\documentclass[a4paper, 11pt]{article}
\usepackage[dvipsnames]{xcolor}
\usepackage{amsmath}
\usepackage{amsthm}
\usepackage{amssymb}
\usepackage{mathcomp}
\usepackage{dsfont}
\usepackage{comment}
\usepackage{authblk}

\usepackage[active]{srcltx}
\usepackage{stackengine}
\usepackage[bookmarksnumbered=true]{hyperref}

\usepackage{geometry}
\newtheorem{definition}{Definition}[section]
\newtheorem{lemma}[definition]{Lemma}
\newtheorem{proposition}[definition]{Proposition}
\newtheorem{theorem}[definition]{Theorem}
\newtheorem{remark}[definition]{Remark}
\newtheorem{corollary}[definition]{Corollary}
\newtheorem{assumption}[definition]{Assumption}

\numberwithin{equation}{section}

\begin{document}
\title{An optimal upper bound for the dilute Fermi gas in three dimensions}

\author{Emanuela L. Giacomelli}
\affil{LMU Munich, Department of Mathematics, Theresienstr. 39, 80333 Munich, Germany}

\maketitle

\abstract{In a system of interacting fermions, the correlation energy is defined as the difference between the energy of the ground state and the one of the free Fermi gas. We consider $N$ interacting spin $1/2$ fermions in the dilute regime, i.e., $\rho\ll 1$ where $\rho$ is the total density of the system. We rigorously derive a first order upper bound for the correlation energy with an optimal error term of the order $\mathcal{O}(\rho^{7/3})$ in the thermodynamic limit. Moreover, we improve the lower bound estimate with respect to previous results getting an error $\mathcal{O}(\rho^{2+1/5})$.}

\tableofcontents
\section{Introduction}
The study of the collective behavior of highly correlated many-body systems has gained a wide attention from the  mathematical physics community in the past decades. The quantum mechanical description of such systems is very complicated and therefore their rigorous understanding is challenging. In this paper, we consider a system of $N$ interacting fermions with spin $\sigma = \{\uparrow, \downarrow\}$ in a box $\Lambda_L := [-L/2,L/2]^3$, with periodic boundary conditions. We suppose the two-body interaction potential to be positive and of short-range type (for more details see Assumption \ref{asu: potential V}). The Hamiltonian associated to the system is 
\begin{equation}
  H_N= -\sum_{i=1}^N\Delta_{x_i} + \sum_{i<j=1}^N V(x_i-x_j).
\end{equation}
 We denote by $N_\sigma$ the number of particles with spin $\sigma$ ($N= N_\uparrow + N_\downarrow$). Correspondingly, let $\rho_\sigma = N_\sigma / L^3$ be the density of particles with a given spin $\sigma$, we write $\rho = \rho_\uparrow + \rho_\downarrow$ to refer to the total density of the system. In particular, we are interested in the dilute regime, thus $\rho_\sigma \ll 1$. The space of allowed wave function is $\mathfrak{h}(N_\uparrow, N_\downarrow):= L^2_a(\Lambda_L^{N_\uparrow}) \otimes L^2_a(\Lambda_L^{N_\downarrow})$, where $L^2_a(\Lambda_L^{N_\sigma})$ is the antisymmetric tensor product of $N_\sigma$ copies of $L^2(\Lambda_L)$. 

We consider the \textit{correlation energy} for the the system just introduced, which is given by the difference between the  energy of the many body ground state and the one of the free Fermi gas. We derive an optimal formula for this quantity: to do that a detailed understanding of the many-body ground state is required. 

In this setting, it is well know \cite{LSS,FGHP} that,  in units such that $\hbar = 1$ and putting the masses of the particles equal to $1/2$, the ground state energy per unit volume can be approximated as
\begin{equation}\label{eq: energy density}
  e(\rho_\uparrow, \rho_\downarrow) = \frac{3}{5}(6\pi^2)^{\frac{2}{3}} (\rho_\uparrow^{\frac{5}{3}} + \rho_\downarrow^{\frac{5}{3}}) + 8\pi a\rho_\uparrow \rho_\downarrow + o(\rho^2)\qquad \mbox{as}\quad \rho\rightarrow 0.
\end{equation} 
The first term in the right hand side of \eqref{eq: energy density} is purely kinetic (kinetic energy of the free Fermi gas) and the fact that this contribution is proportional to $\rho^{5/3}$ is a consequence of the fermionic nature of the wave function. The effect of the interaction appears at the next order in \eqref{eq: energy density}, via the parameter $a$ which is the scattering length of the interaction potential. The role of the correlations is important to derive the asymptotics written in \eqref{eq: energy density} (see also Remark \ref{rem: role correl}). In general, an important aspect of the formula in \eqref{eq: energy density} is its universality, meaning that the ground state energy density is independent of specific details of the system under consideration and depends only on the densities $\rho_\sigma$ and on the scattering length $a$. An outstanding open problem concerning the higher-order corrections to the ground state energy in the low density setting is the proof of the validity of the Huang–Yang formula \cite{HY}, leading, in the case in which $\rho_\uparrow = \rho_\downarrow = \rho/2$, to the conjecture:
\begin{equation}\label{eq: HY}
  e(\rho) = \frac{3}{5}(3\pi^2)^{\frac{2}{3}}\rho^{\frac{5}{3}} + 2\pi a\rho^2 + c_0 \rho^{\frac{7}{3}} + o (\rho^{7/3}),
\end{equation}
with an explicit value of $c_0$.

To study the corrections to the ground state energy of the free Fermi gas, similarly as in \cite{FGHP}, we make use of the almost bosonic nature of the low energy excitations around the Fermi sea. We describe such excitations via Bogoliubov theory. This theory is in general widely used nowadays in the study of many-body systems, in particular for interacting Bose gases. This approach was proposed in 1947 by  Bogoliubov \cite{Bog} and later several works were devoted to its rigorous justification. We mention here some of them. In \cite{LS1comp, LS2comp, Sol} the ground state energy of a one and two components Bose gases is studied, we refer to \cite{YY,FS,BaCS} for the Lee Huang Yang formula for the ground state energy of the Bose gas in the thermodynamic limit (see also \cite{FGJMO} for the corresponding two dimensional case), i.e., 
\begin{equation}\label{eq: LHY}
  e(\rho) = 4\pi a^2 \left[ 1+ \frac{128}{15\sqrt{\pi}} (\rho a^3)^{1/2} + o(\rho a^3)^{1/2})\right].
\end{equation}
The validity of the first order asymptotics in \eqref{eq: LHY} was established in \cite{Dy} for the upper bound and in \cite{LY} for the lower bound.  Note that, differently than in \eqref{eq: energy density}, in bosonic systems the effect of the interaction (through the scattering length) is visible at leading order (see \eqref{eq: LHY}). This is consistent with the fact that fermions have to obey to Pauli's principle which implies that the interaction energy is subleading with respect to the kinetic one. In the fixed volume setting the mathematical results for the Bose gas have been very rich: the Bogoliubov excitation spectrum has been justified in the mean field regime \cite{S,GS, LNSS, DN, Piz, BSP}  as well as in the Gross-Pitaevskii (GP) regime \cite{BBCS2, NT, BSS, HST}. We also refer to \cite{ABS, BCS} for results on Bose gases in intermediate regimes between GP and the thermodynamic limit, and to \cite{DG} for the study of Bose gases in the Thomas-Fermi regime.

Despite the good understanding of the Bose gas, the mathematical progresses in the study of the collective behavior of fermionic systems is not so huge. However, recently many works were devoted to the study of how to go beyond the Hartree Fock approximation in the high density regime. More precisely, for the high density Fermi gas in \cite{BNPSS, CHN} the corrections to the Hartree–Fock energy (random phase approximation) due to non–trivial quantum correlations are studied (see also \cite{BNPSS0, BPSS, BNPSS1}).

As already mentioned, we focus on the low density Fermi gas. The first proof of \eqref{eq: energy density} was proposed in \cite{LSS}, there a large class of two-body interaction potentials was considered (including the hard-core case). Successively, in \cite{Sfermi} the result has been extended to the dilute Fermi gas at positive temperature, providing an expansion for the thermodynamic pressure. Moreover, an asymptotics analogous to the one in \eqref{eq: energy density} is valid for interacting fermions in a lattice, i.e., for the Hubbard model, see \cite{Giu} (upper bound) and \cite{SY} (lower bound). Recently the same asymptotics as the one in \eqref{eq: energy density} was re-derived in \cite{FGHP}. With respect to \cite{LSS}, in \cite{FGHP} more restrictions are put in the interaction potential (see Assumption \ref{asu: potential V}). However, in \cite{FGHP} better error estimates are obtained. More precisely, in \cite{LSS} the error in the asymptotic expansion is lower bounded by $\rho^{2+1/39}$ and upper bounded by $\rho^{2+2/27}$, instead, in \cite{FGHP}, we have $\rho^{2+1/9}$ and $\rho^{2+2/9}$ for the lower bound and the upper bound, respectively. The main novelty in \cite{FGHP} is the  use of a completely different approach in the derivation of \eqref{eq: energy density}. As already mentioned before, we describe the low energy correlations by pairs of fermions (particle-hole pairs) behaving approximately as bosonic particles. The fact that these emergent bosonic particles are responsible for the $8\pi a\rho_\uparrow\rho_\downarrow$ correction in the energy density expansion can be expected also comparing \eqref{eq: energy density} with \eqref{eq: LHY}. 

As already anticipated in \cite{FGHP, Gia} this approach provides a new point of view on the dilute Fermi gas and is useful in getting refined energy asymptotics than the one in \eqref{eq: energy density}. In this paper, by developing further the approach introduced in \cite{FGHP}, we derive the optimal upper bound for the correlation energy:
\begin{equation}\label{eq: energy density opt ub}
  e_L(\rho_\uparrow, \rho_\downarrow) \leq  \frac{3}{5}(6\pi^2)^{\frac{2}{3}} (\rho_\uparrow^{\frac{5}{3}} + \rho_\downarrow^{\frac{5}{3}}) + 8\pi a\rho_\uparrow \rho_\downarrow + \mathcal{O}(\rho^{\frac{7}{3}}), \qquad \mbox{as}\quad \rho\rightarrow 0.
\end{equation} 
The reason why we say that the upper bound above is optimal is that it is in agreement with the Huang-Yang's conjecture (see \eqref{eq: HY}). We expect that \eqref{eq: energy density opt ub} provides useful inputs to eventually reach the Huang-Yang asymptotics in \eqref{eq: HY} (with the correct $c_0\rho^{7/3}$ contribution). As a by product of our analysis, we improve the error estimate in the lower bound from \cite[Theorem 2.1]{FGHP}.\\

\noindent\textbf{Organization of the paper.} In Section \ref{sec: main result} we present the main results and the strategy of the proof. In Section \ref{sec: ferm Bog} we introduce the fermionic Fock space, we take into account the energy of the free Fermi gas and we recall the main properties of the well known particle-hole transformation; we conclude the section by isolating the correlation Hamiltonian and estimating some sub-leading contributions. In Section \ref{sec: correl struct} we define the almost bosonic Bogoliubov transformation which we use to study the correlation energy and we prove some general bounds which are useful in what follows. In Section \ref{sec: propagation estimates} we set up all the technical results that we need to rigorously use the almost-bosonic Bogoliubov transformation. In Section \ref{sec: up bd} and Section \ref{sec: lw bd} we conclude the proof of the optimal upper bound and of the improved lower bound, respectively. Finally, in Appendix \ref{app: scattering} we collect some properties of the solution of the scattering equation which we use in our approach and in Appendix \ref{app: cut-off} we discuss some technical estimates.
\\

\noindent\textbf{Acknowledgments.} The author deeply acknowledge Christian Hainzl, Phan Thành Nam and Marcello Porta for the continuous support and for helpful suggestions and discussions.

\section{Main results} \label{sec: main result} 

We now present the main results of the paper. First of all we recall that the Hamiltonian of the system is given by 
\begin{equation}
  H_N= -\sum_{i=1}^N\Delta_{x_i} + \sum_{i<j=1}^N V(x_i-x_j),
\end{equation}
where $\Delta_{x_i}$ denotes the Laplacian acting on the i-th particle and $V$ is the pair interaction potential which can be thought as the periodization on the box $\Lambda_L$ of a regular potential $V_\infty$ defined on $\mathbb{R}^3$. We write more precisely our assumptions below.  
\begin{assumption}\label{asu: potential V} Let $V$ be such that 
\begin{equation}\label{eq: period V}
  V(x-y) = \frac{1}{L^3}\sum_{p\in\frac{2\pi}{L}\mathbb{Z}^3}\hat{V}_\infty(p)e^{ip\cdot (x-y)}, \qquad \hat{V}_\infty(p) = \int_{\mathbb{R}^3}dx\, V_\infty(x) e^{-ip\cdot x}.
\end{equation}
We suppose $V_\infty$ to be non negative, radial, smooth and compactly supported with $\mathrm{supp} V \subset \{x\in\mathbb{R}^3\, \vert\, |x| < R_0\}\subset \mathbb{R}^3$, for some $R_0>0$.
\end{assumption}
Moreover, we suppose our system to be dilute, i.e., $\rho_\sigma \ll 1$. We are interested in the thermodynamic limit, meaning that $N_\sigma, L \rightarrow \infty$ keeping $\rho_\sigma$ fixed.

We now rigorously introduce the main quantities involved in our analysis. We denote the ground state energy by 
\begin{equation} 
  E_{L}(N_\uparrow, N_\downarrow) = \inf_{\Psi\in\mathfrak{h}(N_\uparrow, N_\downarrow)}\frac{\langle \Psi, H_N \Psi\rangle}{\langle \Psi, \Psi\rangle}.
\end{equation}
Being the Hamiltonian translation invariant, the energy is extensive in the system size. Thus, we can define the ground state energy density as 
\begin{equation}\label{eq: up bound main}
  e(\rho_\uparrow, \rho_\downarrow) = \frac{E_{L}(N_\uparrow, N_\downarrow)}{L^3}.
\end{equation}
\begin{theorem}[Optimal upper bound]\label{thm: optimal up bd} Let $V, V_\infty$ as in Assumption \ref{asu: potential V}. There exists $L_0> 0$ such that for $L \geq L_0$, it holds
\begin{equation} \label{eq: opt up bd}
  e_L(\rho_\uparrow, \rho_\downarrow) \leq \frac{3}{5}(6\pi^2)^{\frac{2}{3}} (\rho_\uparrow^{\frac{5}{3}} + \rho_\downarrow^{\frac{5}{3}}) + 8\pi a\rho_\uparrow\rho_\downarrow + \mathcal{O}(\rho^{7/3}),
  \end{equation}
  where $a$ is the scattering length of the interaction potential $V_\infty$.
\end{theorem}
Following the same ideas of the proof of Theorem \ref{thm: optimal up bd}, we can also improve the error estimate for the lower bound with respect to the one in \cite[Theorem 2.1]{FGHP}.
\begin{theorem}[Improved lower bound]\label{thm: lower bound} Let $V, V_\infty$ be as in Assumption \ref{asu: potential V}. There exists $L_0 > 0$ such that for $L\geq L_0$, it holds
\begin{equation}\label{eq: impr lw bd}
  e_L(\rho_\uparrow, \rho_\downarrow) \geq \frac{3}{5}(6\pi^2)^{\frac{2}{3}} (\rho_\uparrow^{\frac{5}{3}} + \rho_\downarrow^{\frac{5}{3}}) + 8\pi a\rho_\uparrow\rho_\downarrow + \mathcal{O}(\rho^{2+ 1/5}),
\end{equation}
where $a$ is the scattering length of the interaction potential $V_\infty$.
\end{theorem}
\begin{remark}[The role of correlations]\label{rem: role correl} The correlations between particles play an important role in the derivation of \eqref{eq: opt up bd} and \eqref{eq: impr lw bd}. Indeed, the $\rho^2$ dependence in the energy per volume of the free Fermi gas explicitly reads as $\hat{V}(0)\rho_\uparrow\rho_\downarrow$. Being $\hat{V}(0) > 8\pi a$, we see that the most uncorrelated fermionic state, i.e., the free Fermi gas, does not reproduce the right asymptotics.  
\end{remark}
\begin{remark}[On the condensation estimate]\label{rem: cond est} As a consequence of Theorem \ref{thm: optimal up bd} and Theorem \ref{thm: lower bound}, we get (see \cite[Remark 6.3]{FGHP}) a bound for the difference between the reduced one-particle density matrix of any approximate ground state\footnote{By approximate ground state we refer to a state which is close enough to ground state of the system (for a rigorous definition, see Definition \ref{def: approx gs})} $\psi$  and the reduced one-particle density matrix of the free Fermi gas $\omega$, see, e.g., \cite{BPS}. In particular from the bounds in \eqref{eq: opt up bd} and \eqref{eq: impr lw bd}, it immediately follows that 
\begin{equation}
  \mathrm{tr} \gamma^{(1)}_\psi (1-\omega) \leq CL^3\rho^{\frac{19}{15}},
\end{equation}
which improves the analogous estimate in \cite[Eq. (6.59)]{FGHP}.
\end{remark}
\subsection{Strategy of the proofs}\label{sec: proof strategy}
As already mentioned above, in the approach we follow here, we develop further the ideas introduced in \cite{FGHP}. To start, we stress that, as in \cite{FGHP}, working in the thermodynamic limit we need some kind of localization arguments: we implement them introducing cut-off functions on the relevant quantities used in our analysis. On one hand this is convenient to work directly in the thermodynamic limit. Indeed, this allows us to gain a fast decay in configuration space of the quantities involved in the definition of the almost-bosonic Bogoliubov transformation. On the other hand, via the cut-off, we introduce by hand many error terms whose estimate is not trivial. The proofs of Theorem \ref{thm: optimal up bd} and Theorem \ref{thm: lower bound} are based on the same ideas. The reason why the lower bound is not as satisfactory as the upper bound estimate is that, we do not get optimal estimates for some of the error terms induced by the cut-off and supported for momenta $k$ of the order $\rho^{1/3}$ (for more details see Proposition \ref{pro: reg}). We emphasize that this is the only point of the proof in which we do not get optimal estimates $\mathcal{O}(\rho^{7/3})$ in the analysis of the lower bound.

In what follows we first recall the main ideas of the approach introduced in \cite{FGHP} and then we underline the improvements we achieve here. As we mentioned above, the general strategy is to describe the low energy excitations around the Fermi ball making use of almost bosonic operators. We denote the almost-bosonic creation (resp. annihilation) operators by $\hat{b}_{p,\sigma}^\ast$ (resp. $\hat{b}_{p,\sigma}$) for a given $p\in (2\pi/L)\mathbb{Z}^3$ and a given spin $\sigma$. Explicitly they are defined through pairs of fermionic creation (resp. annihilation) operators one of which has momentum inside the Fermi ball and the other outside (see Remark \ref{rem: bog in p space}). The reason why we refer to them as \textit{almost-bosonic operators} is that when acting on states with few particles they approximately behave as bosonic creation/annihilation operators (i.e., they approximately satisfy the canonical commutation relations). Note that these operators are the same ones used in the bosonization approach in the fermionic high density regime (see e.g., \cite{BPSS,CHN}). The next step is then to express the effective correlation Hamiltonian (see \eqref{eq: def H corr} for a concrete expression) in terms of $\hat{b}_{p,\sigma}, \hat{b}^\ast_{p,\sigma}$ and to extract the constant contribution $8\pi a \rho_\uparrow\rho_\downarrow$ through a (almost-bosonic) Bogoliubov transformation. The Bogoliubov transformation we use explicitly reads as 
\begin{equation}\label{eq: bog intro}
  T = \exp\left\{\frac{1}{L^3}\sum_{p\in\frac{2\pi}{L}\mathbb{Z}^3} \hat{\varphi}(p)\hat{b}_{p,\uparrow}\hat{b}_{-p,\downarrow} - \mathrm{h.c.}\right\}.
\end{equation}
Note that treating $\hat{b}_{p,\sigma}$, $\hat{b}^\ast_{p,\sigma}$ as true bosonic operators, the unitary operator $T$ implements a bosonic Bogoliubov transformation (see, e.g., \cite{BBCS2}). Moreover, as we are going to underline later, the choice of $\hat{\varphi}$ is responsible for getting the right dependence on the scattering length in the constant term we want to extract, i.e., $8\pi a \rho_\uparrow\rho_\downarrow$. 
For more details about the heuristic discussion above, we refer to\footnote{Note that in \cite[Section 4]{FGHP} a rescaled version of the almost-bosonic operators is used, i.e., $\hat{c}_{p,\sigma} = \rho_\sigma^{-1/2}\hat{b}_{p,\sigma}$ (resp. $\hat{c}_{p,\sigma}^\ast = \rho_\sigma^{-1/2}\hat{b}_{p,\sigma}^\ast$)} \cite[Section 4]{FGHP}. Here we only underline that, as already explained in \cite{FGHP}, even at the heuristic level, in order to justify this approach we need to introduce a cut-off in the momenta living in the support of $\hat{\varphi}(p)$ (compare the expression of $T$ above with the one in \eqref{eq: def T reg p space}). We achieve this constraint by regularizing $\varphi(x)$, so that is supported on a ball of radius proportional to $\rho^{-1/3}$.

We are now ready to explain in which sense we developed further the approach first proposed in \cite{FGHP}. To start, a difference here with respect to \cite{FGHP} is the definition of the cut-off functions to regularize $\hat{b}_{p,\sigma}$, $\hat{b}^\ast_{p,\sigma}$. In particular, differently than in \cite{FGHP}, in the present work we neglect all the momenta $k$ such that $k_F\leq |k| \leq 2k_F$, where $k_F$ denotes the Fermi momentum. This different way of regularizing $\hat{b}_{p,\sigma}$, $\hat{b}^\ast_{p,\sigma}$ allow us to simplify many proofs with respect to the ones proposed in \cite{FGHP}, getting better error estimates (see e.g., Remark \ref{rem: N_>}).

However, the main difference with respect to \cite{FGHP} is the choice of the kernel of the almost bosonic Bogoliubov transformation in \eqref{eq: bog intro}. In \cite{FGHP} we define $\varphi(x)$ through the solution of the scattering equation in a ball in $\mathbb{R}^3$ satisfying Neumann boundary conditions (see \cite[Eq. (4.28)]{FGHP}). Here, we take $\varphi(x)$ to be the periodization of a localized version of the solution of the zero energy scattering equation in $\mathbb{R}^3$. This has some advantages as we are going to explain later. To be more concrete, let $\varphi_0$ be the solution of the zero energy scattering equation in $\mathbb{R}^3$, i.e.,  
\begin{equation}\label{eq: scatter eq phi0}
 2\Delta\varphi_0 + V_\infty(1-\varphi_0) = 0 \qquad \mbox{in}\,\,\,\mathbb{R}^3, \quad \mbox{with}\quad \varphi_0(x)\rightarrow 0 \quad \mbox{as}\,\,\, |x| \rightarrow \infty,
\end{equation}
and let $\chi$ be a smooth cut-off function in $\mathbb{R}^3$, which is such that $0\leq \chi\leq 1$ and
\begin{equation}\label{eq: def chi}
  \chi(x) = \chi(|x|) = \begin{cases} 1, &\mbox{if}\,\,\, |x| \leq 1, \\ 0, &\mbox{if}\,\,\, |x| \geq 2. \end{cases}
\end{equation}
We then take into account a localization of $\varphi_0$ made via $\chi$, i.e., we define for all $x\in\mathbb{R}^3$,
\begin{equation}\label{eq: def phi infty}
  \varphi_\infty(x) := \varphi_0(x)\chi(x/\rho^{-1/3}).
\end{equation}
Looking at $\varphi_\infty$, we are localizing the scattering solution in a ball of radius $\sim \rho^{-\gamma}$, with $\gamma = 1/3$. This choice of $\gamma$ turns out to be the optimal one in order to get our main results (see also Remark \ref{rem: cut-off phi}). As a consequence of our localization, the equation satisfied by $\varphi_\infty$ differs from \eqref{eq: scatter eq phi0} only for $|x|$ large, i.e., for $|x|\sim \rho^{-1/3}$. In particular, $\varphi_\infty = \varphi_0$ in the support of the interaction potential $V_\infty$. Thus, $\varphi_\infty$ is such that 
\begin{equation}\label{eq: equation for tilde phi approx}
  2\Delta\varphi_\infty + V_\infty(1-\varphi_\infty) \sim -2a\left[\frac{2\rho^{\frac{1}{3}}}{|x|^2} + \frac{\rho^{\frac{2}{3}}}{|x|}\right] \chi(|x|\sim \rho^{-\frac{1}{3}}),
\end{equation}
where $a$ is the scattering length associated to the interaction potential and $\chi(|x|\sim \rho^{-\frac{1}{3}})$ means that the right hand side above is not vanishing only for $\rho^{-1/3}\leq |x| \leq 2\rho^{-1/3}$. Note that we are localizing $\varphi_0$ through a cut-off which is not constant far away from the compact support of the interaction potential. This  allows us to write the explicit expression of the renormalized potential appearing in the right side of \eqref{eq: equation for tilde phi approx}, simplifying the study of it. Moreover, the factors $\rho$ in the right hand side of \eqref{eq: equation for tilde phi approx} are related to the fact that the cut-off $\chi$ is not constant only in a region of order $\rho^{-1/3}$.  

After having introduced $\varphi_\infty$, in order to make our choice of $\varphi(x)$ compatible with the periodic boundary conditions, we introduce a periodization of $\varphi_\infty$ in the box $\Lambda_L$, which we denote by $\varphi$ and which explicitly reads as 
\begin{equation}\label{eq: def phi}
  \varphi(x) = \sum_{n\in\mathbb{Z}^3}\varphi_\infty(x + nL),
\end{equation}
or, equivalently,
\[
  \varphi(x) = \frac{1}{L^3}\sum_{p\in\frac{2\pi}{L}\mathbb{Z}^3}\hat{\varphi}_\infty(p)e^{ip\cdot x},\qquad \hat{\varphi}_\infty(p) = \int_{\mathbb{R}^3} dx\, \varphi_\infty(x) e^{-ip\cdot x}.
\]
As already mentioned, with respect to \cite{FGHP}, we are taking the periodization of a different scattering equation defined over $\mathbb{R}^3$. This new choice plays a crucial role in our analysis in order to improve the error estimates both for the upper bound and for the lower bound (see Remark \ref{rem: comparison kinetic energy}) and to, in general, improve our previous result.

To conclude, we underline that, for $L$ large enough, the $\varphi$ defined in \eqref{eq: def phi} is such that 
\[
  \hat{V}(0) - \int_{\Lambda_L}dx\, V(x)\varphi(x)= 8\pi a. 
\]  
This will play a role in our analysis. Indeed, acting through the unitary operator $T$ (almost-bosonic Bogoliubov transformation) introduced above, we extract the constant term $-\rho_\uparrow\rho_\downarrow\int_{\Lambda_L}dx V(x)\varphi(x)$ from the correlation energy. Summing it with the interaction energy (per unit volume) coming from the free Fermi gas, i.e., $\rho_\uparrow\rho_\downarrow\hat{V}(0)$, we reconstruct the first order correction to the correlation energy, i.e., $8\pi a \rho_\uparrow\rho_\downarrow $.

To conclude we also stress that many technical improvements are needed with respect to the approach followed in \cite{FGHP}. In order to simplify the discussion here, we will highlight them during the proofs. We only do few more comments. The bounds in Lemma \ref{lem: bound b phi}, which are an improved version of those in \cite[Lemma 5.3]{FGHP}, play a crucial role in getting refined estimates in both the upper bound and the lower bound. Moreover, differently than in \cite{FGHP}, to bound the ``\textit{cubic terms}'' (e.g., the terms proportional to $b^\ast a^\ast a$ or $a^\ast a b$) a more careful analysis is required (see Section \ref{sec: Q3}) with respect to the one done in \cite{FGHP}, where only a priori estimates were used.

In general, in what follows, we write the statements of all the intermediate results in full generality, in order to directly use them to conclude both the proof of Theorem \ref{thm: optimal up bd} (see Section \ref{sec: up bd}) and the one of Theorem \ref{thm: lower bound} (see Section \ref{sec: lw bd}).

\section{The Hartree Fock approximation}\label{sec: ferm Bog}
In this section we discuss the Hartree Fock approximation. In particular, we make use of a unitary transformation (the so called \textit{particle-hole transformation}) to compare the ground state energy with the free Fermi gas energy. This approach is quite standard and, for the setting we are interested in, is explained in details in \cite[Section 3]{FGHP}. Therefore, we only underline the main features.

\subsection{The fermionic Fock space} Being convenient to work in second quantization, we start by fixing the notation for the fermionic Fock space. We define   
\begin{equation} 
  \mathcal{F} = \bigoplus_{n\geq 0} \mathcal{F}^{(n)}, \qquad \mathcal{F}^{(n)} = L^2(\Lambda_L; \mathbb{C}^2)^{\wedge n}, \qquad \mathcal{F}^{(0)} := \mathbb{C}.
\end{equation}
Any $\psi\in\mathcal{F}$ is then of the form $\psi = (\psi^{(0)}, \psi^{(1)}, \cdots , \psi^{(n)}, \cdots)$ with $\psi^{(n)}\in \mathcal{F}^{(n)}$ and $\psi^{(n)} \equiv \psi^{(n)} ( (x_1, \sigma_1), \cdots, (x_n, \sigma_n))$ where $(x,\sigma)\in \Lambda_L \times \{\uparrow, \downarrow\}$. In the following, we will work with the vacuum vector $\Omega$, i.e., the zero particle state:
\[
  \Omega = (1, 0, 0, \cdots, 0, \cdots).
\]
Let now $f\in L^2(\Lambda_L; \mathbb{C}^2)$, $f= (f_\uparrow, f_\downarrow)$, the fermionic annihilation operator $a(f): \mathcal{F}^{(n)}\rightarrow \mathcal{F}^{(n-1)}$ is such that $a(f)\Omega = 0$ and it explicitly acts as 
\begin{multline}\label{eq: def annihil}
  (a(f)\psi)^{(n)} \left(x_1, \sigma_1), \cdots (x_n, \sigma_n)\right)
  = \sqrt{n+1} \sum_{\sigma} \int_{\Lambda_L} dx\, \overline{f_\sigma(x)}\psi^{(n+1)} ( (x,\sigma), (x_1, \sigma_1), \cdots (x_n, \sigma_n)).
\end{multline}
Correspondingly, the fermionic creation operator $a^\ast_\sigma(f): \mathcal{F}^{(n-1)}\rightarrow\mathcal{F}^{(n)}$ is defined as 
\begin{multline}\label{eq: def creat}
  (a^\ast(f)\psi)^{(n)}( (x_1,\sigma_1), \cdots, (x_n, \sigma_n)) 
  \\
  = \frac{1}{\sqrt{n}} \sum_{j=1}^n (-1)^{j+1} f_{\sigma_j}(x_j) \psi^{(n-1)} ( (x_1, \sigma_1), \cdots (x_{j-1}, \sigma_{j-1}), (x_{j+1}, \sigma_{j+1}), \cdots , (x_n,\sigma_n)).
\end{multline}
It is well known that for any $f,g\in L^2(\Lambda_L; \mathbb{C}^2) $, the creation and annihilation operators satisfy the canonical anticommutation relations (CAR), i.e, 
\[
  \{a(f), a(g)\} = \{a^\ast(f), a^\ast(g)\} = 0, \quad \{a(f), a^\ast(g)\} = \langle f,g \rangle_{L^2(\Lambda_L;\mathbb{C}^2)}.
\]
As a consequence one gets the following bounds for the operator norm:
\[
  \|a(f)\|, \|a^\ast(f)\| \leq \|f\|_{L^2(\Lambda_L;\mathbb{C}^2)}.
\]
Taking $f$ to be a normalized plane wave, i.e., $f(x) \equiv f_k(x) = e^{ik\cdot x}/L^{3/2}$, for $k\in (2\pi/L)\mathbb{Z}^3$, we can express the creation/annihilation operators in momentum space, i.e., 
\begin{equation}\label{eq: def creat/annihi momentum}
  \hat{a}_{k,\sigma} \equiv a_\sigma(f_k) = \frac{1}{L^{\frac{3}{2}}}\int_{\Lambda_L} dx\, a_{x,\sigma} e^{-ik\cdot x}, \qquad \hat{a}^\ast_{k, \sigma} = (a_\sigma(f_k))^\ast,
\end{equation}
where we explicitly used that, according to our definitions in \eqref{eq: def annihil} and \eqref{eq: def creat}, $a^\ast(f)$ is the adjoint of $a(f)$ for any $f\in L^2(\Lambda_L; \mathbb{C}^2)$. In \eqref{eq: def creat/annihi momentum} we denoted by $a_{x,\sigma}$, $a^\ast_{x,\sigma}$ the creation/annihilation operator-valued distributions which, formally, are such that $a_{x,\sigma} = a(\delta_{x,\sigma})$ and $a^\ast_{x,\sigma} = a^\ast (\delta_{x,\sigma})$, with 
\[
  \delta_{x,\sigma}(y,\sigma^\prime) = \delta_{\sigma,\sigma^\prime}\delta(x-y) = \frac{\delta_{\sigma,\sigma^\prime}} {L^3}\sum_{k\in\frac{2\pi}{L}\mathbb{Z}^3} e^{ik\cdot (x-y)},
\]
where $\delta_{\sigma,\sigma^\prime}$ denotes the Kronecker delta and $\delta(x-y)$ is the Dirac delta distribution periodic over $\Lambda_L$.
We now fix the notation for the number operator:
\begin{equation} 
  \mathcal{N} = \sum_{\sigma = \uparrow, \downarrow} \mathcal{N}_\sigma  \equiv \sum_{\sigma = \uparrow, \downarrow} \int_{\Lambda_L}\, dx\, a^{\ast}_{x,\sigma} a_{x,\sigma} \equiv \sum_{\sigma = \uparrow, \downarrow} \sum_{k\in \frac{2\pi}{L}\mathbb{Z}^3} \hat{a}^\ast_{k,\sigma}\hat{a}_{k,\sigma}.
\end{equation}
In what follows we may use the notation $\mathcal{N}_\sigma$ to refer to the number operator with a fixed spin, i.e., 
\begin{equation}
  \mathcal{N}_\sigma = \sum_k \hat{a}_{k,\sigma}\hat{a}_{k,\sigma}.
\end{equation}
Finally, the second-quantized Hamiltonian reads as
\[
  \mathcal{H} = \sum_\sigma \int_{\Lambda_L} dx\, \nabla_x a^\ast_{x,\sigma} \nabla_x a_{x,\sigma} + \frac{1}{2}\sum_{\sigma,\sigma^\prime} \int_{\Lambda_L\times \Lambda_L} dxdy\, V(x-y) a^\ast_{x,\sigma} a^\ast_{y,\sigma^\prime}a_{y,\sigma^\prime}a_{x,\sigma}.
\]
It then follows that $(\mathcal{H}\psi)^{(n)} = H_n \psi^{(n)}$ for $n\geq 1$. Note that the $n$-th particle Hamiltonian $H_n$ is independent of the spin. Thus, if we then denote by $\mathcal{F}(N_\uparrow, N_\downarrow)\subset\mathcal{F}$ the subset of $\mathcal{F}$ given by $N$-particle states with $N = N_\uparrow + N_\downarrow$ particles, we can rewrite the ground state energy as 
\[
  E_L(N_\uparrow, N_\downarrow) = \inf_{\psi\in\mathcal{F}^{(N_\uparrow, N_\downarrow)}}\frac{\langle \psi, \mathcal{H}\psi\rangle}{\langle \psi, \psi\rangle}.
\]  

\subsection{The free Fermi gas}
We now write explicitly the state representing the free Fermi gas and given by the anti-symmetric product of plane waves, i.e., 
\begin{equation}
  \psi_{\mathrm{FFG}} (\{x_i, \uparrow\}_{i=1}^{N_\uparrow}, \{y_j,\downarrow\}_{j=1}^{N_\downarrow}) = \frac{1}{\sqrt{N}_\uparrow}\frac{1}{\sqrt{N}_\downarrow}(\mathrm{det}f_{k_\ell}^\uparrow(x_i))_{1\leq i,\ell\leq N_\uparrow}(\mathrm{det}f^\downarrow_{k_m}(y_j))_{1 \leq m,j\leq N_\downarrow},
\end{equation}
here $f^\sigma_k(x) \equiv f_k(x) = e^{ik\cdot x}/L^{3/2}$ with $k\in\mathcal{B}_F^\sigma$ where $\mathcal{B}_F^\sigma$ is the Fermi ball, i.e.,
\begin{equation}
  \mathcal{B}_F^\sigma = \big\{ k\in (2\pi/L)\mathbb{Z}^3\, \vert\, |k| \leq k_F^\sigma\big\}.
\end{equation} 
Here and in the following, $k_F$ denotes the Fermi momentum which, for fixed densities and in the limit $L\rightarrow \infty$, can be written as
\begin{equation}
  k_F^\sigma = (6\pi^2)^{\frac{1}{3}} \rho_\sigma^{\frac{1}{3}} + o(1). 
\end{equation}
Note that for a fixed value of the spin, all the momenta involved in the definition of $\psi_{\mathrm{FFG}}$ are different, the state then respect to the Pauli principle.
\begin{remark}[Filled Fermi sea] We suppose the Fermi ball to be completely filled, i.e., we suppose to only take into account the values $N_\sigma$ such that $N_\sigma = |\mathcal{B}_F^\sigma|$. This is not a loss of generality since the corresponding densities $\rho_\sigma$ form a dense subset of $\mathbb{R}_+$ as $L\rightarrow \infty$ (see \cite[Section 3]{FGHP} for more details).
\end{remark}
\begin{remark}[Fully antisymmetrized version of $\psi_{\mathrm{FFG}}$] Being the Hamiltonian $H_N$ spin-independent, it is convenient to introduce the fully antisymmetrized version of $\psi_{\mathrm{FFG}}$ which we denote by $\Psi_{\mathrm{FFG}}$ and which explicitly reads as
\begin{equation}\label{eq: fully ant ffg}
  \Psi_{\mathrm{FFG}}(x_1, \cdots, x_N) = \frac{1}{\sqrt{N!}} \mathrm{det}(f^{\sigma_i}_i(x_j))_{1\leq i,j\leq N}, \qquad \sigma_i = \uparrow, \downarrow. 
\end{equation}
As explained in \cite[Section 3]{FGHP}, it holds
\begin{equation}
  \langle \psi_{\mathrm{FFG}}, H_N \psi_{\mathrm{FFG}}\rangle = \langle \Psi_{\mathrm{FFG}}, H_N \Psi_{\mathrm{FFG}}\rangle.
\end{equation} 
\end{remark}
We denote by $\omega:L^2(\Lambda_L;\mathbb{C}^2) \rightarrow L^2(\Lambda_L;\mathbb{C}^2)$  the reduced one particle density matrix associated to $\Psi_{\mathrm{FFG}}$, i.e., the operator with integral kernel given by 
\begin{equation}\label{eq: 1particle density matrix w}
  \omega_{\sigma,\sigma^\prime}(x;y) = \frac{\delta_{\sigma,\sigma^\prime}}{L^3}\sum_{k\in\mathcal{B}_F^\sigma} e^{ik\cdot (x-y)}.
\end{equation}
One can then prove, proceeding similarly as in \cite[Section 3.2.1]{FGHP}, that
\begin{equation}
  \langle \psi_{\mathrm{FFG}}, H_N \psi_{\mathrm{FFG}}\rangle = \langle \Psi_{\mathrm{FFG}}, H_N \Psi_{\mathrm{FFG}}\rangle = E_{\mathrm{HF}}(\omega),
\end{equation}
where $E_{\mathrm{HF}}(\omega)$ is the Hartree-Fock energy functional, i.e., 
\begin{equation}\label{eq: HF funcitional}
  E_{\mathrm{HF}}(\omega) = -\mathrm{tr}\Delta\omega + \frac{1}{2}\sum_{\sigma,\sigma^\prime}\int_{\Lambda_L\times \Lambda_L}dxdy\, V(x-y) \big(\omega_{\sigma,\sigma}(x;x)\omega_{\sigma^\prime, \sigma^\prime}(y;y) - |\omega_{\sigma, \sigma^\prime}(x;y)|^2\big).
\end{equation}
We now quantify $\mathcal{E}_{\mathrm{HF}}(\omega)$. Taking into account the exchange term:
\begin{eqnarray}
  -\frac{1}{2}\sum_{\sigma,\sigma^\prime}\int_{\Lambda_L \times \Lambda_L} dxdy\, V(x-y)|\omega_{\sigma,\sigma^\prime}(x;y)|^2 &=& -\frac{1}{2}\sum_{\sigma,\sigma^\prime} \frac{\delta_{\sigma,\sigma^\prime}}{L^3}\sum_{k,k^\prime\in \mathcal{B}_F^\sigma}\hat{V}(k-k^\prime)
  \\
  &=& -\frac{1}{2}\sum_\sigma L^3\hat{V}(0)\rho_\sigma^2 + \mathcal{O}(L^3\rho^{\frac{8}{3}}),\nonumber
\end{eqnarray}
where the last estimate holds for $L$ large. In particular, we used the fact that $k,k^\prime \in \mathcal{B}_F^\sigma$, which implies $|k-k^\prime| \leq C\rho^{1/3}$. As a consequence, one gets $\hat{V}(k - k^\prime) = \hat{V}(0) + \mathcal{O}(\rho^{2/3})$, using the fact that the interaction potential is radial.
To calculate the remaining contributions in $E_{\mathrm{HF}}(\omega)$ we refer to \cite[Section 3]{FGHP}. In the end, taking $L$ large enough, we get that 
\begin{equation}\label{eq: HF unit volume}
  \frac{E_{\mathrm{HF}}(\omega)}{L^3} = \frac{3}{5}(6\pi^2)^{\frac{2}{3}}(\rho_\uparrow^{\frac{5}{3}} + \rho_\downarrow^{\frac{5}{3}}) + \hat{V}(0)\rho_\uparrow\rho_\downarrow + \mathcal{O}(\rho^{\frac{8}{3}}).
\end{equation}
\subsection{The fermionic Bogoliubov transformation (\textit{particle-hole transformation})}
Given the reduced one-particle density matrix of the free Fermi gas, $\omega_{\sigma,\sigma^\prime} = \delta_{\sigma,\sigma^\prime}\sum_{k \in \mathcal{B}_F^\sigma}|f_k\rangle \langle f_k|$, we define $u, v :L^2(\Lambda_L; \mathbb{C}^2) \rightarrow L^2(\Lambda_L; \mathbb{C}^2)$ to be two operators with integral kernel given by 
\begin{equation} \label{eq: def u,v}
  v_{\sigma,\sigma^\prime} (x;y)  =  \delta_{\sigma,\sigma^\prime} \sum_{k\in\mathcal{B}_{F}^\sigma} |\overline{f_k}\rangle \langle f_k|, \qquad u_{\sigma, \sigma^\prime}(x;y) = \delta_{\sigma,\sigma^\prime}\left(\delta(x-y) -\omega_{\sigma,\sigma^\prime}(x;y)\right),
\end{equation}
where $\delta(x-y)$ denotes the periodic Dirac delta distribution. Note that $\omega = \overline{v}v$ and that $u\overline{v} = 0$. Thorough $u,v$ we now define the fermionic Bogoliubov transformation. In the following, to simplify the notation we write $u_{\sigma}$ and $v_\sigma$ in place of $u_{\sigma,\sigma}$ and $v_{\sigma,\sigma}$, respectively. Moreover, often we simply write $\sum_\sigma$ for $\sum_{\sigma= \uparrow, \downarrow}$ and $\int dx$ for $\int_{\Lambda_L} dx$. Similarly, often we shall use the notation $\|\cdot\|_p$ for $\|\cdot\|_{L^p(\Lambda_L)}$.
\begin{definition}[Particle-hole transformation]\label{def: ferm bog}
Let $u$,$v: L^2(\Lambda_L;\mathbb{C}^2)\rightarrow L^2(\Lambda_L;\mathbb{C}^2)$ be the operators with integral kernel defined in \eqref{eq: def u,v}. The particle-hole transformation is a unitary operator $R: \mathcal{F}\rightarrow\mathcal{F}$ such that the following properties hold: 
\begin{itemize}
  \item[(i)] The state $R\Omega$ is such that $(R\Omega)^{(n)} = 0$ whenever $n\neq N$ and $(R\Omega)^{(N)} = \Psi_{\mathrm{FFG}}$.
  \item[(ii)] It holds
  \begin{equation}\label{eq: prop II R}
    R^\ast a_{x,\sigma} R = a_\sigma(u_x) + a^\ast_\sigma(\overline{v}_x),
  \end{equation}
  where
  \[
    a_\sigma(u_x) = \int_{\Lambda_L} dy\, \overline{u_\sigma(y;x)}a_{y,\sigma}, \qquad a^\ast_\sigma(\overline{v}_x) = \int_{\Lambda_L} dy\,\overline{v_\sigma(y;x)} a^\ast_{y,\sigma}.
  \] 
\end{itemize}
\end{definition}

The existence and the well-posedness of the unitary $R$ is a consequence of Shale-Stinespring theorem, see \cite{Sol1}.
\begin{remark}[The Hartree Fock energy via $R$] From point (i) in Definition \ref{def: ferm bog}, it holds that 
\begin{equation}
  \langle R\Omega, H_N R\Omega\rangle = E_{\mathrm{HF}}(\omega).
\end{equation}

\end{remark}
\begin{remark}[Action of $R$ in momentum space] From \eqref{eq: prop II R} and using that $\hat{a}_{k,\sigma} = a_\sigma(f_k)$ it follows that
\begin{equation}
  R^\ast \hat{a}_{k,\sigma} R = \begin{cases} \hat{a}_{k,\sigma} &\mbox{if}\,\,\, k\notin\mathcal{B}_F^\sigma, \\ \hat{a}^\ast_{k,\sigma} &\mbox{if}\,\,\, k\in \mathcal{B}_F^\sigma, \end{cases}
\end{equation} 
which justifies the name \textit{particle-hole transformation}.
\end{remark}
The next proposition is useful to compare the ground state energy with the Hartree Fock energy via conjugation under the particle-hole transformation $R$. 
\begin{proposition}[Conjugation under $R$]\label{pro: fermionic transf}
  Let $V$ be as in Assumption \ref{asu: potential V}. Let $\psi\in \mathcal{F}$ be a normalized state, such that $\langle \psi, \mathcal{N}_\sigma\psi\rangle = N_\sigma$ and $N = N_\uparrow + N_\downarrow$. Then 
  \begin{equation} 
    \langle \psi, \mathcal{H}\psi\rangle = E_{\mathrm{HF}}(\omega) + \langle R^\ast\psi, \mathbb{H}_0 R^\ast\psi\rangle + \langle R^\ast\psi, \mathbb{X}R^\ast\psi\rangle + \langle R^{\ast} \psi, \mathbb{Q}R^\ast\psi\rangle,
  \end{equation}
  where $E_{\mathrm{HF}}(\omega)$ is the Hartree-Fock energy functional introduced in \eqref{eq: HF funcitional}. The operators $\mathbb{H}_0$ and $\mathbb{X}$ are given by 
\begin{eqnarray}
  \mathbb{H}_0 &=& \sum_\sigma\sum_k ||k|^2 -\mu_\sigma| \hat{a}_{k,\sigma}^\ast \hat{a}_{k,\sigma}, \qquad \mu_\sigma = (k_F^\sigma)^2,
  \\
  \mathbb{X} &=& \sum_\sigma \int dxdy\, V(x-y)\omega_\sigma(x-y)\left(a^\ast_\sigma(u_x)a_\sigma(u_y) -a^\ast_\sigma(\overline{v}_y)a_\sigma(\overline{v}_x)\right).
\end{eqnarray}
The operator $\mathbb{Q}$ can be written as $\mathbb{Q} = \sum_{i=1}^4\mathbb{Q}_i$, with 
\begin{eqnarray}
  \mathbb{Q}_1 &=& \frac{1}{2}\sum_{\sigma,\sigma^\prime}\int dxdy\, V(x-y) \left( a^\ast_\sigma(u_x)a^\ast_{\sigma^\prime}(u_y)a_{\sigma^\prime}(u_y)a_\sigma(u_x)\right)\nonumber,
  \\ 
  \mathbb{Q}_2 &=& \frac{1}{2}\sum_{\sigma,\sigma^\prime}\int dxdy\, V(x-y) a^\ast_\sigma(u_x)a^\ast_{\sigma^\prime}(\overline{v}_x)a_{\sigma^\prime}(\overline{v}_y)a_{\sigma^\prime}(u_y)\nonumber
  \\
  && + \frac{1}{2}\sum_{\sigma,\sigma^\prime}\int dxdy\, V(x-y) \left(a^\ast_\sigma(\overline{v}_x)a^\ast_{\sigma^\prime}(\overline{v}_y)a_{\sigma^\prime}(\overline{v}_y)a_\sigma(\overline{v}_x) -2a^\ast_\sigma(u_x)a^\ast_{\sigma^\prime}(\overline{v}_y)a_{\sigma^\prime}(\overline{v}_y)a_\sigma(u_x)\right)\nonumber,
  \\
  \mathbb{Q}_3 &=& -\sum_{\sigma,\sigma^\prime} \int dxdy\, V(x-y)\left(a^\ast_\sigma(u_x) a^\ast_{\sigma^\prime}(u_y) a^\ast_{\sigma^\prime}(\overline{v}_x)a_{\sigma^\prime}(u_y) - a^\ast_\sigma(u_x) a^\ast_{\sigma^\prime}(\overline{v}_y) a^\ast_{\sigma}(\overline{v}_x)a_{\sigma^\prime}(\overline{v}_y)\right) + \mathrm{h.c.}\nonumber,
  \\
  \mathbb{Q}_4 &=& \frac{1}{2}\sum_{\sigma,\sigma^\prime} \int dxdy\, V(x-y) a^\ast_\sigma(u_x)a^\ast_\sigma(u_y)a^\ast_{\sigma^\prime}(\overline{v}_y)a^\ast_{\sigma^\prime}(\overline{v}_x) + \mathrm{h.c.}\nonumber
\end{eqnarray}
Moreover, the following inequality holds true: 
\begin{equation}
  \langle \psi, \mathcal{H}\psi\rangle \geq E_{\mathrm{HF}}(\omega) + \langle R^\ast\psi, \mathbb{H}_0 R^\ast\psi\rangle + \langle R^\ast \psi, \mathbb{X}R^\ast\psi\rangle + \rangle R^\ast \psi, \widetilde{\mathbb{Q}} R^\ast \psi\rangle, 
\end{equation}
where $\widetilde{\mathbb{Q}} = \sum_{i=1}^4 \widetilde{\mathbb{Q}}_i$ and each operator $\widetilde{\mathbb{Q}}_i$ is defined as the corresponding $\mathbb{Q}_i$ with the sums over $\sigma, \sigma^\prime$ replaced by the sums over $\sigma\neq \sigma^\prime$. 
\end{proposition}
For the proof of the Proposition \ref{pro: fermionic transf} we refer to \cite{BPS, BJPSS, HPR, FGHP}. The idea is to conjugate the Hamiltonian under $R$, to use \eqref{eq: prop II R} and to put the operators in normal order. The resulting Hamiltonian is useful to describe the correlations between particles:
\begin{equation}\label{eq: def H corr}
  \mathcal{H}_{\mathrm{corr}} := R\mathcal{H}R^\ast - E_{\mathrm{HF}}(\omega) =  \mathbb{H}_0 + \mathbb{X} + \sum_{i=1}^4\mathbb{Q}_i .
\end{equation}
%
Later we show that the effective Hamiltonian responsible for the correct asymptotics of the correlation energy to leading order is given by 
\begin{equation}
  \mathcal{H}_{\mathrm{eff}} = \mathbb{H}_0 + \widetilde{\mathbb{Q}}_1 + \widetilde{\mathbb{Q}}_4.
\end{equation}
In the next proposition we then bound the expectations of some sub-leading operators.
\begin{proposition}[Bounds for $\mathbb{X}$, $\mathbb{Q}_2$]\label{pro: X,Q2} Let $V$ be as in Assumption \ref{asu: potential V}. Let $\psi\in \mathcal{F}$ be a normalized state, such that $\langle \psi, \mathcal{N}_\sigma\psi\rangle = N_\sigma$ and $N = N_\uparrow + N_\downarrow$. Then, the following holds
  \begin{equation} 
    |\langle \psi, \mathbb{X}\psi\rangle| \leq C\rho\langle \psi, \mathcal{N}\psi\rangle,\qquad |\langle \psi, \mathbb{Q}_2\psi\rangle| \leq C\rho\langle \psi, \mathcal{N}\psi\rangle.
  \end{equation}
\end{proposition}
For the proof of Proposition \ref{pro: X,Q2} we refer to \cite[Proposition 3.3]{FGHP}.

Later, in our analysis, we will need to estimate the expectation value of $\mathcal{N}$, $\mathbb{H}_0$, $\mathbb{Q}_1$ and $\widetilde{\mathbb{Q}}_1$ over states which are close enough to ground state of the system. In order to quantify this, similarly as in \cite{FGHP}, we introduce here the the notion of being an \textit{approximate ground state}.
\begin{definition}[Approximate ground state]\label{def: approx gs} Let $\psi\in\mathcal{F}$ be a normalized state such that $\langle \psi, \mathcal{N}_\sigma\psi\rangle = N_\sigma$ and $N= N_\uparrow + N_\downarrow$. We say that $\psi$ is an approximate ground state of $\mathcal{H}$ if it holds 
\begin{equation}
	\left|\langle \psi, \mathcal{H}\psi\rangle - \sum_{\sigma= \uparrow, \downarrow}\sum_{k\in\mathcal{B}_F^\sigma} |k|^2\right| \leq CL^3\rho^2.
\end{equation}
\end{definition}
Following \cite{FGHP}, we can prove the next lemma.
\begin{lemma}[A priori estimates for $\mathbb{H}_0$, $\mathcal{N}$, $\mathbb{Q}_1$]\label{lem: a priori} Let $\psi$ be an approximate ground state. Under the same assumptions of Theorem \ref{thm: optimal up bd} and Theorem \ref{thm: lower bound}, we have
\begin{equation}
  \langle R^\ast\psi, \mathbb{H}_0 R^\ast\psi\rangle \leq CL^3\rho^2, \qquad \langle R^\ast \psi, \mathcal{N} R^\ast \psi\rangle \leq CL^3\rho^{\frac{7}{6}}
  \end{equation}
  \begin{equation} \langle R^\ast \psi,\mathbb{Q}_1 R^\ast \psi\rangle \leq CL^3\rho^2, \qquad \langle R^\ast \psi,\widetilde{\mathbb{Q}}_1 R^\ast \psi\rangle \leq CL^3\rho^2. 
\end{equation}
\end{lemma}
For the proof of Lemma \ref{lem: a priori}, we refer to \cite[Lemma 3.5, Lemma 3.5, Corollary 3.7, Lemma 3.9]{FGHP}. 

 \section{The correlation structure}\label{sec: correl struct}

In this section we rigorously introduce the almost-bosonic Bogoliubov transformation $T$ which we use to extract the first order correction to the correlation energy. 

Before proceeding further, we fix some notations. We denote by $|\cdot|$ the Euclidean distance on $\mathbb{R}^3$ a{}nd by $|\cdot|_L$ the distance on the torus, i.e., $|x-y|_L = \min_{n\in\mathbb{Z}^3}|x-y + nL|$.

As already explained in Section \ref{sec: proof strategy}, we need to introduce some cut-off on the relevant quantities involved in our analysis. We then introduce a regularization for the operators $u$ and $v$ defined in Section \ref{sec: ferm Bog}. We define
\begin{equation}\label{eq: def ur vr}
  v^r_{\sigma, \sigma^\prime} = \frac{\delta_{\sigma, \sigma^\prime}}{L^3}\sum_k \hat{v}^r_\sigma(k) \overline{|f_k\rangle}\langle f_k|,\qquad u^{r}_{\sigma,\sigma^\prime} = \frac{\delta_{\sigma,\sigma^\prime}}{L^3} \sum_k u^r_\sigma(k) |f_k\rangle \langle f_k|,
\end{equation}
where we recall that $f_k(x) = L^{-3/2} e^{ik\cdot x}$ and where $0\leq \hat{v}^r(k), \hat{u}^r(k) \leq 1$ are smooth radial functions such that
\begin{equation}
  \hat{v}^{r}_{\sigma}(k) = \begin{cases} 1 &\mbox{for }\quad |k| \leq k_F^\sigma -\rho^{\alpha}_\sigma, \\ 
    0 &\mbox{for}\quad |k| \geq k_F^\sigma, 
    \end{cases}
    \qquad \hat{u}^{r}_{\sigma}(k) = \begin{cases} 0 &\mbox{for}\quad |k| \leq 2k_F^\sigma, \\ 
    1 &\mbox{for}\quad  3k_F^\sigma \leq |k| \leq \rho^{-\beta}_\sigma,
    \\
    0  &\mbox{for}\quad |k|\geq 2\rho_\sigma^{-\beta},\end{cases}
\end{equation}
with $\alpha = \frac{1}{3} + \frac{\epsilon}{3}$, for some $\epsilon >0$ to be chosen later and $\beta >0$. We denote, the regularized version of the one particle reduced density matrix is given by $\omega^r = \overline{v}^r v^r$. Note that under our choices for the regularization, $u^r$ and $v^r$ are still orthogonal, i.e., $u^r\overline{v}^r = 0$. In the following, we use the notation $u_{x,\sigma}(\cdot) := u_\sigma(\cdot\,;x)$, $v_{x,\sigma}(\cdot):= v_\sigma(\cdot\,;x)$ (and analogously for similar quantities). 

\begin{proposition}\label{pro: L1, L2 v}
   For $L$ large enough, one has
  \begin{equation}\label{eq: bounds u,v}
    \|v^r_{x, \sigma}\|_{2} \leq C\rho^{\frac{1}{2}},\qquad \|u^r_{x,\sigma}\|_{2} \leq C\rho^{-\frac{3}{2}\beta}, \qquad \|\omega^r_{x,\sigma}\|_{1} \leq C \rho_\sigma^{-\frac{\epsilon}{3}},\qquad \|u^r_{x,\sigma}\|_{1} \leq C.
  \end{equation}
   Moreover, for any $n\in\mathbb{N}$ it holds 
  \begin{equation}\label{eq: est infty norm omega}
    \|D^\alpha\omega^r_{x,\sigma}\|_\infty \leq C\rho^{1+\frac{n}{3}}, \qquad \mbox{with}\,\,\, \alpha = (\alpha_1, \alpha_2, \alpha_3) \in \mathbb{N}^3_0\,\,\,\mbox{and}\,\,\, |\alpha|= n.
  \end{equation}
\end{proposition}
\begin{proof}
For the proof of the bounds in \eqref{eq: bounds u,v}, we refer to \cite[Proposition 4.2]{FGHP}. Note that in \cite{FGHP} a different regularization for $u$ was used. However, anything changes in the proof, since also in our regularization $\hat{u}^r$ interpolates smoothly between $0$ and $1$ on the scales $\rho^{1/3}$ and $\rho^{-\beta}$. The estimate in \eqref{eq: est infty norm omega} directly follows from the definition of $\omega^r$.
\end{proof}
\begin{remark}\label{rem: u2 and tilde omega}
Note that in some of the estimates in what follows, we will use also $\widetilde{w}^r = \overline{v}^r v$ and 
\begin{equation}
(u^r)^2_{\sigma,\sigma^\prime} = \frac{\delta_{\sigma,\sigma^\prime}}{L^3}\sum_{k\in\frac{2\pi}{L}\mathbb{Z}^3}(\hat{u}^r_\sigma(k))^2 |f_k\rangle \langle f_k|.
\end{equation}
Both $\widetilde{\omega}^r_{x,\sigma}(\cdot)$ and $(u^r)^2_{x,\sigma}(\cdot)$ satisfy the same bounds in Proposition \ref{pro: L1, L2 v} as $\omega^r_{x,\sigma}(\cdot)$ and $u^{r}_{x,\sigma}(\cdot)$, respectively.
\end{remark}
In what follows, we introduce the almost bosonic Bogoliubov transformation. In particular, we recall that we denote by $\varphi_0$ the solution of the zero energy scattering equation in $\mathbb{R}^3$, i.e., the solution of the equation in \eqref{eq: scatter eq phi0}, which reads as
\begin{equation}\label{eq: def phi not app}
 2\Delta\varphi_0 + V_\infty(1-\varphi_0) = 0 \qquad \mbox{in}\,\,\,\mathbb{R}^3, \quad \mbox{with}\quad \varphi_0(x)\rightarrow 0 \quad \mbox{as}\,\,\, |x| \rightarrow \infty.
\end{equation}
We also recall that we denote by $\varphi_\infty$ a localized version of $\varphi_0$, i.e.,
\[
  \varphi_\infty(x) = \varphi(x)\chi(x/\rho^{-1/3}),
\]
where $\chi$ is the smooth cut-off introduced in \eqref{eq: def chi}. From now on we use the following notation $\chi_{\sqrt[3]{\rho}}(x):= \chi (x/ \rho^{-\frac{1}{3}})$.
As we anticipated in \eqref{eq: equation for tilde phi approx}, $\varphi_\infty$ satisfies a different equation than the one for $\varphi_0$. We now write it rigorously. It is  
\begin{equation}\label{eq: equation for tilde phi}
  2\Delta\varphi_\infty + V_\infty(1-\varphi_\infty) = \mathcal{E}_{\{\varphi_0, \chi_{\sqrt[3]{\rho}}\}}^\infty,
\end{equation}
with $\mathcal{E}_{\{\varphi_0, \chi_{\sqrt[3]{\rho}}\}}^\infty(x) = -4\nabla \varphi_0(x) \nabla \chi_{\sqrt[3]{\rho}}(x) - 2\varphi_0(x) \Delta\chi_{\sqrt[3]{\rho}}(x)$. Note that, thanks to the compact support of the interaction potential, we have an explicit formula for the renormalized interaction $\mathcal{E}_{\{\varphi_0, \chi_{\sqrt[3]{\rho}}\}}^\infty$: this is convenient in our analysis. More precisely, due to the fact $\chi_{\sqrt[3]\rho}$ is not constant only in a region far from the support of the interaction potential, we have that $\varphi_0(x) = a/|x|$ in the support of $\mathcal{E}^\infty_{\varphi, \chi_{\sqrt[3]\rho}}$.

As already explained, in our analysis we need to use a periodization of $\varphi_\infty$ in the box $\Lambda_L$, which we denote by $\varphi$ and which explicitly reads as (see \eqref{eq: def phi})
\begin{equation}
  \varphi(x) = \sum_{n\in\mathbb{Z}^3}\varphi_\infty(x + nL),
\end{equation}
or, equivalently,
\[
  \varphi(x) = \frac{1}{L^3}\sum_{p\in\frac{2\pi}{L}\mathbb{Z}^3}\hat{\varphi}_\infty(p)e^{ip\cdot x},\qquad \hat{\varphi}_\infty(p) = \int_{\mathbb{R}^3} dx\, \varphi_\infty(x) e^{-ip\cdot x}.
\]
From now on we take $L$ large enough such that $\mathrm{supp}\varphi_\infty(x) \subset \Lambda_L$.
Note that this is approach is similar to the one followed in \cite{FGHP}, also there we used a scattering equation defined over $\mathbb{R}^3$ and we took into account the periodization of it. The difference, as already underlined, is the scattering equation we use, which is crucial to get a final error proportional to $\rho^{7/3}$.

\subsection{The almost bosonic Bogoliubov transformation}
Making us of the $\varphi$ introduced above, we can now define the almost bosonic Bogoliubov transformation.
\begin{definition}[The transformation $T$]\label{def: transf T}
  Let $\lambda\in [0,1]$. We define the unitary operator $T_\lambda: \mathcal{F}\rightarrow\mathcal{F}$ as 
  \begin{equation}
    T_\lambda := e^{\lambda(B-B^\ast)}, \quad B:= \int_{\Lambda_L\times \Lambda_L} dzdz^\prime\, \varphi(z-z^\prime)\, a_\uparrow(u^r_z)a_\uparrow(\overline{v}^r_z)a_\downarrow(u^r_{z^\prime})a_\downarrow(\overline{v}^r_{z^\prime}), 
  \end{equation}
where $\varphi$ is as in \eqref{eq: def phi} . Moreover, in the sequel, we set $T\equiv T_1$.
\end{definition}
Note that despite the presence of two volume integrations, the operator $B$ is bounded proportionally to $L^3$. To see this, we notice that 
\[
  \|B\|\leq \int_{\Lambda_L\times \Lambda_L} dzdz^\prime\, |\varphi(z-z^\prime)| \|a_\uparrow(u^r_z)\|\|a_\uparrow(\overline{v}^r_z)\|\|a_\downarrow(u^r_{z^\prime})\|\|a_\downarrow(\overline{v}^r_{z^\prime})\|.
\]
Using then the regularization in $\hat{u}^r$, we can bound $\|a_\sigma(u^r_z)\|\leq \|u^r_z\|_2 \leq C\rho^{-3\beta/2}$. Similarly, using the support properties of $\hat{v}^r$, we have $\|a_\sigma(v^r_{z})\|\leq \|v^r_z\|_2 \leq C\rho^{1/2}$. Moreover, from the definition of $\varphi$, taking $L$ large enough, we have that $\|\varphi\|_{L^1(\Lambda_L)}\leq C\rho^{-2/3}$ (see Lemma \ref{lem: bound phi}). Using then the periodicity of $\varphi$, we can conclude that 
\[
  \|B\|\leq CL^3\rho^{1-3\beta - \frac{2}{3}}.
\]

\begin{remark}[Definition of $T$ in momentum space]\label{rem: bog in p space} In Definition \ref{def: transf T} we introduced the almost-bosonic Bogoliubov transformation written in configuration space since this is the expression we are going to use in what follows. To underline the almost-bosonic structure, we rewrite it in momentum space. We define\footnote{For a complete understanding, we underline that the almost bosonic operators $\hat{b}_{p,\sigma}, \hat{b}^\ast_{p,\sigma}$ mentioned in Section \ref{sec: proof strategy} correspond to \[\hat{b}_{p,\sigma} = \sum_{k\in\frac{2\pi}{L}\mathbb{Z}^3}\hat{u}_\sigma(k+p) \hat{v}_\sigma(k)\hat{a}_{k+p, \sigma}\hat{a}_{k,\sigma} \equiv \sum_{\substack{k\in\frac{2\pi}{L}\mathbb{Z}^3 \\ k \in \mathcal{B}_F^\sigma, \, \,k+p \notin \mathcal{B}_F^\sigma}}\hat{a}_{k+p, \sigma}\hat{a}_{k,\sigma},\]
and analogously for $\hat{b}^\ast_{p,\sigma}$.}
\begin{equation}
  \hat{b}_{p,\sigma}^r = \sum_{k\in\frac{2\pi}{L}\mathbb{Z}^3}\hat{u}^r_\sigma(k+p) \hat{v}^r_\sigma(k)\hat{a}_{k+p, \sigma}\hat{a}_{k,\sigma},\qquad (\hat{b}_{p,\sigma}^r)^\ast = \sum_{k\in\frac{2\pi}{L}\mathbb{Z}^3}\hat{u}^r_\sigma(k+p) \hat{v}^r_\sigma(k)\hat{a}_{k, \sigma}^\ast\hat{a}_{k+p,\sigma}^\ast. 
\end{equation}
As a consequence, we can rewrite the almost-bosonic Bogoliubov transformation $T$ as 
\begin{equation}\label{eq: def T reg p space}
T_\lambda = \exp\left\{\frac{\lambda}{L^3}\sum_{p\in\frac{2\pi}{L}\mathbb{Z}^3} \hat{\varphi}(p)\hat{b}_{p,\uparrow}^r\hat{b}_{-p,\downarrow}^r - \mathrm{h.c.}\right\},
\end{equation}
where $\hat{\varphi}(p)$, $p\in(2\pi/L)\mathbb{Z}^3$ are the Fourier coefficients of $\varphi(x)$, i.e., $\hat{\varphi}(p) = \int_{\Lambda_L}dx\, \varphi(x) e^{-ip\cdot x}$.
\end{remark}
\begin{remark}[Cut-off induced on $\varphi$]\label{rem: cut-off phi} We notice that the regularizations over $\hat{u}$ and $\hat{v}$ (i.e., $\hat{u}^r$, $\hat{v}^r$) induce a cut-off on $\varphi$ implying that $\hat{\varphi}(p)$ is supported for $k_F \leq |p| \leq 3\rho^{-\beta}$.  Note that the localization in configuration space for $\varphi$ is consistent with the condition $|p| \geq k_F \sim \rho^{1/3}$: we are indeed localizing $\varphi(x)$ in a ball of radius $\sim \rho^{-1/3}$. We also underline that the localization at order $|x|\sim \rho^{-1/3}$ is the optimal one to get a final error estimate (in the upper bound) of the order $\mathcal{O}(\rho^{7/3})$.
\end{remark}
In the lemmas below, we collect useful properties of $\varphi$ which we use often in what follows.
\begin{lemma}[Bounds for $\varphi$]\label{lem: bound phi} Let $V$ as in Assumption \ref{asu: potential V}. Let $\varphi$ as in \eqref{eq: def phi}. Taking $L$ large enough, the following holds
\begin{itemize}
  \item[(i)] For all $x\in \Lambda_L$, there exists $C>0$ such that
  \begin{equation}\label{eq: uniform norm phi >}
   \qquad|\varphi(x)| \leq C.
  \end{equation}
  \item[(ii)] It holds
\begin{equation}
  \|\varphi\|_{L^2(\Lambda_L)} \leq C\rho^{-\frac{1}{6}},\qquad \|\nabla\varphi\|_{L^2(\Lambda_L)}\leq C,
  \end{equation}
 and
  \begin{equation}  \|\varphi\|_{L^1(\Lambda_L)} \leq C\rho^{-\frac{2}{3}}, \qquad \|\nabla\varphi\|_{L^1(\Lambda_L)} \leq C\rho^{-\frac{1}{3}}, \qquad \|\Delta\varphi\|_{L^1(\Lambda_L)} \leq C. 
  \end{equation}
  Moreover, 
  \begin{equation}
  \|D^2\varphi\|_{L^1(\Lambda)} \leq C|\log\rho|, \qquad \|D^3\varphi\|_{L^1(\Lambda)} \leq C.
  \end{equation}
\end{itemize}
\end{lemma}

In the next result we state some decay properties for the Fourier coefficients of $\varphi$, i.e., $\hat{\varphi}(p) = \int_{\Lambda_L}dx\, \varphi(x) e^{-ip\cdot x}$, $p\in (2\pi/L)\mathbb{Z}^3$. These estimates are going to be useful when dealing with the cut-off we introduced for momenta of the order $|k| \sim \rho^{-\beta}$.
\begin{lemma}[Decay estimates for $\varphi_>$]\label{lem: decay phi} Let $V$ as in Assumption \ref{asu: potential V}. Let $\varphi$ as in \eqref{eq: def phi}. It holds
\begin{equation}\label{eq: est decay phi>}
  |\hat{\varphi}(p)|\leq C_n\frac{\rho^{-\frac{2}{3}}}{(1+\rho^{-\frac{2}{3}}|p|^n)}, \qquad \forall n\in \mathbb{N}, \,\, n\geq 2,
\end{equation}
for some $C_n>0$.
\end{lemma}
For the proof of Lemma \ref{lem: bound phi} and Lemma \ref{lem: decay phi}, we refer the reader to Appendix \ref{app: scattering}. 
\subsection{Useful bounds}
In this section we prove some bounds which are going to be useful later. In particular, the bounds in the lemma below are an improved version of those in \cite[Lemma 5.3]{FGHP}.
\begin{lemma}\label{lem: bound b phi}Let $\varphi$ as in \eqref{eq: def phi}.
  Let $b_{z^\prime,\sigma} := a_\sigma(u^r_{z^\prime})a_\sigma(\overline{v}^r_{z^\prime})$ and let
  \begin{equation}\label{eq: def b phi}
      b_\sigma(\varphi_{z}) := \int_{\Lambda_L} dz^\prime\, \varphi(z-z^\prime)\, b_{z^\prime, \sigma}, \qquad b_\sigma(\partial_j\varphi_{z}) := \int_{\Lambda_L} dz^\prime\,\partial_j\varphi(z-z^\prime)\,b_{z^\prime,\sigma},
  \end{equation}
  for each fixed $j=1,2,3$. It holds
  Then 
  \begin{equation} \label{eq: op bound b phi}
      \| b_\sigma(\varphi_z)\| \leq C \rho^{\frac{1}{3}}, \qquad \| b_\sigma(\partial_j \varphi_{z})\| \leq C \rho^{\frac{1}{2}},
  \end{equation}
  and the same holds for the adjoint operators.
Furthermore, let $\widetilde{b}_{z^\prime,\sigma}:= a_\sigma(u^r_{z^\prime})a_\sigma(\partial_\ell\overline{v}^r_{z^\prime})$ for some $\ell=1,2,3$ and let 
\[
  \widetilde{b}_\sigma(\varphi_{z}) := \int_{\Lambda_L}dz^\prime\, \varphi(z-z^\prime)\widetilde{b}_{z^\prime,\sigma}, \qquad \widetilde{b}_\sigma(\partial_j\varphi_{z}) := \int_{\Lambda_L}dz^\prime\, \partial_j\varphi(z-z^\prime)\widetilde{b}_{z^\prime,\sigma},
\]
for some $j=1,2,3$.
We have 
\begin{equation}
  \| \widetilde{b}_\sigma(\varphi_z)\| \leq C \rho^{\frac{2}{3}}, \qquad \|\widetilde{b}_\sigma(\partial_j\varphi_{z})\|\leq C\rho^{\frac{5}{6}},
\end{equation}  
and the same is true for the adjoint operators.
\end{lemma}
\begin{proof}[Proof of Lemma \ref{lem: bound b phi}] 
   We now estimate the operator norm of $b_\sigma(\varphi_z)$, the proof for the adjoint operator works the same. To start, we omit the spin dependence which does not play any role and we rewrite the operator $b(\varphi_z)$ in momentum space. We recall that following the same ideas as in \cite[Lemma 5.3]{FGHP}, we have that $\|b(\varphi_z)\|\leq C\rho^{-\frac{1}{6} - \frac{3}{2}\beta}$. Here we write  
\begin{equation}
  b(\varphi_z) = \frac{1}{L^3}\sum_{p,k}\hat{\varphi}(p) e^{ip\cdot z} \hat{u}^r(k+p)\hat{v}^r(k)\hat{a}_{k+p}\hat{a}_{k}.
\end{equation}
We want now to show that the operator $b(\varphi_z)$ can be rewritten as 
\begin{eqnarray}\label{eq: decomposition b phi}
  b(\varphi_z) &=& -\int dz^\prime \,\Delta \varphi(z-z^\prime)a(\widetilde{u}^r_{z^\prime})a(\overline{v}^r_{z^\prime}) + 2 \sum_{j=1}^3\int dz^\prime\, \partial_j \varphi(z-z^\prime)a(\widetilde{u}^r_{z^\prime})a(\partial_j\overline{v}^r_{z^\prime})\nonumber\\
  && - \int dz^\prime\, \varphi(z-z^\prime)a(\widetilde{u}^r_{z^\prime})a(\Delta\overline{v}^r_{z^\prime}),
\end{eqnarray}
where $\widetilde{u}^r_x(\cdot):= \widetilde{u}^r(\cdot\,;x)$ (we will use analogous notations also later in the proof), with
\[
  \widetilde{u}^r(t;x) = \frac{1}{L^3}\sum_{p\in\frac{2\pi}{L}\mathbb{Z}^3}\frac{\hat{u}^r(p)}{|p|^2} e^{ip\cdot (t-x)}.
\]
and where (omitting the spin dependence)
\[
  a(\partial_j\overline{v}^r_{z^\prime}) = \int_{\Lambda_L}dt\, \partial_jv^r(t;z^\prime) a_{t} =    \int_{\Lambda_L}dt\, \frac{1}{L^3}\sum_{k\in\frac{2\pi}{L}\mathbb{Z}^3} (-ik_j)\hat{v}^r(k) e^{-ik\cdot (t+z^\prime)} a_{t},
\]
\[
  \partial_j \varphi(z-z^\prime) = \frac{1}{L^3}\sum_{k\in\frac{2\pi}{L}\mathbb{Z}^3} (ik_j)\hat{\varphi}(k)e^{ik\cdot (z-z^\prime)}.
\]
To prove that \eqref{eq: decomposition b phi} holds, one can write each of the term in the right hand side above in momentum space. In particular,
\[
  -\int dz^\prime \,\Delta \varphi(z-z^\prime)a(\widetilde{u}^r_{z^\prime})a(\overline{v}^r_{z^\prime}) = \frac{1}{L^3}\sum_{p,k} |p|^2\hat{\varphi}(p) e^{ip\cdot z}\frac{\hat{u}^r(k+p)}{|k+p|^2}\hat{v}^r(k)\hat{a}_{k+p}\hat{a}_{k},
\]
\[
  2\sum_{j=1}^3\int dz^\prime\, \partial_j \varphi(z-z^\prime)a(\widetilde{u}^r_{z^\prime})a(\partial_j\overline{v}^r_{z^\prime}) = \frac{2}{L^3}\sum_{p,k} (p\, \hat{\varphi}(p))\cdot(k\hat{v}^r(k)) e^{ip\cdot z} \frac{\hat{u}^r(k+p)}{|k+p|^2}\hat{a}_{k+p}\hat{a}_{k},
\]
\[
  - \int dz^\prime\, \varphi(z-z^\prime)a(\widetilde{u}^r_{z^\prime})a(\Delta\overline{v}^r_{z^\prime}) =  \frac{1}{L^3}\sum_{p,k} (|k|^2\hat{v}^r(k))\hat{\varphi}(p) e^{ip\cdot z} \frac{\hat{u}^r(k+p)}{|k+p|^2}\hat{a}_{k+p}\hat{a}_{k},
\]
which yields
\begin{multline}
  -\int dz^\prime \,\Delta \varphi(z-z^\prime)a(\widetilde{u}^r_{z^\prime})a(\overline{v}^r_{z^\prime}) +2\sum_{j=1}^3 \int dz^\prime\, \partial_j \varphi(z-z^\prime)a(\widetilde{u}^r_{z^\prime})a(\partial_j\overline{v}^r_{z^\prime}) - \int dz^\prime\, \varphi(z-z^\prime)a(\widetilde{u}^r_{z^\prime})a(\Delta\overline{v}^r_{z^\prime})
  \\
  = \frac{1}{L^3}\sum_{p,k}\hat{\varphi}(p) e^{ip\cdot z} \hat{u}^r(k+p)\hat{v}^r(k)\hat{a}_{k+p}\hat{a}_{k}.
\end{multline}
It then follows that 
\begin{eqnarray}
\|b(\varphi_z)\| &\leq& \int\, dz^\prime |\Delta\varphi(z-z^\prime)|\|a(\widetilde{u}^r_{z^\prime})\|\|a(\overline{v}^r_{z^\prime})\|\| +  2\sum_{j=1}^3\int\, dz^\prime |\partial_j\varphi(z-z^\prime)|\|a(\widetilde{u}^r_{z^\prime})\|\|a(\partial_j\overline{v}^r_{z^\prime})\|\|\nonumber
\\
&& + \int\, dz^\prime |\varphi(z-z^\prime)|\|a(\widetilde{u}^r_{z^\prime})\|\|a(\Delta\overline{v}^r_{z^\prime})\|.
\end{eqnarray}
Now, using that (see Lemma \ref{lem: bound phi}):
\begin{equation}\label{eq: L1 norms phi}
  \|\Delta\varphi\|_{L^1(\Lambda_L)} \leq C, \qquad \|\nabla\varphi\|_{L^1(\Lambda_L)}\leq
   C\rho^{-\frac{1}{3}},\qquad \|\varphi\|_{L^1(\Lambda_L)}\leq C\rho^{-\frac{2}{3}},
\end{equation}
and that 
\begin{equation}\label{eq: L2 norm lemma b phi}
\|a(\partial^n_j\overline{v}^r_{z^\prime})\| \leq \|\partial_j^n\overline{v}^r_{z^\prime}\|_{L^2(\Lambda_L)} \leq C\rho^{\frac{1}{2} +\frac{n}{3}},\qquad \|a(\widetilde{u}^r_{z^\prime})\| \leq \|\widetilde{u}^r_{z^\prime}\|_{L^2(\Lambda_L)} \leq C\rho^{-\frac{1}{6}},
\end{equation}
we can estimate 
\begin{equation}
  \|b(\varphi_z)\| \leq C\rho^{\frac{1}{2} -\frac{1}{6}}\|\Delta\varphi\|_{L^1(\Lambda_L)} + C\rho^{\frac{1}{2} +\frac{1}{3} -\frac{1}{6}}\|\nabla\varphi\|_{L^1(\Lambda_L)} + C\rho^{\frac{1}{2} +\frac{2}{3} -\frac{1}{6}}\|\varphi\|_{L^1(\Lambda_L)} \leq C\rho^{\frac{1}{3}}.
\end{equation}
The proof for $\widetilde{b}_\sigma(\varphi_z)$ can be done similarly as above, we omit the details. We now take into account the operator $b(\partial_j\varphi_{z})$ for a fixed $j=1,2,3$. To do not lose powers of $\rho$ depending on $\beta$ in the final estimate, it is convenient to first split the operator $b(\partial_j\varphi_z)$ as follows 
\begin{equation}
  b(\partial_j\varphi_z) = \int dz^\prime\, \partial_j\varphi(z-z^\prime)a(u^<_{z^\prime})a(\overline{v}^r_{z^\prime}) + \int dz^\prime\, \partial_j\varphi(z-z^\prime)a(u^>_{z^\prime})a(\overline{v}^r_{z^\prime}) =: b^<( \partial_j\varphi_z) + b^>( \partial_j\varphi_z),
\end{equation}
where $\hat{u}^<(k)$ and $\hat{u}^>(k)$ are defined as $\hat{u}^<(k) := \hat{u}^r(k)\hat{\chi}(k)$ and $\hat{u}^>(k) := \hat{u}^r(k)\hat{\eta}(k)$, respectively. Here $\hat{\chi}$ and $\hat{\eta}$ are radial smooth functions such that $0\leq \hat{\chi}, \hat{\eta}\leq 1$, $\hat{\chi}+ \hat{\eta} = 1$ and such that 
\[
  \hat{\chi}(k) := \begin{cases} 1, &\mbox{if}\,\,\, |k| \leq 1, \\ 0, &\mbox{if}\,\,\, |k| >2, \end{cases}\qquad \hat{\eta}(k) = \begin{cases} 0, &\mbox{if}\,\,\, |k| \leq 1 \\ 1, &\mbox{if}\,\,\, |k| >2. \end{cases}
\]
We now take into account $b^<(\partial_j\varphi_z)$, we can rewrite it as 
\begin{equation}\label{eq: splitting b<}
  b^<( \partial_j\varphi_z) = -\int dz^\prime \varphi(z-z^\prime)a (\partial_j u^<_{z^\prime})a(\overline{v}^r_{z^\prime}) + \int dz^\prime \varphi(z-z^\prime)a (u^<_{z^\prime})a(\partial_j\overline{v}^r_{z^\prime}),
\end{equation}
where, similarly as above (omitting the spin)
\[
  a(\partial_j u^<_{z^\prime}) = \int_{\Lambda_L} \overline{\partial_j u^<(t;z^\prime)}a_{t} = \int_{\Lambda_L}\frac{1}{L^3}\sum_{k\in\frac{2\pi}{L}\mathbb{Z}^3}(-ik_j)\hat{u}^<(k) e^{-ik\cdot (t-z^\prime)} a_{t}
\]
We can now proceed similarly as above to bound both the operators in the right hand side of \eqref{eq: splitting b<}. More precisely, if we define 
\[
  \breve{u}^{<,j}(z;x) := \frac{1}{L^3}\sum_{p\in\frac{2\pi}{L}\mathbb{Z}^3} (ip_j)\frac{\hat{u}^<(p)}{|p|^2}e^{ip\cdot (z-x)}, \qquad \widetilde{u}^<(z;x) := \frac{1}{L^3}\sum_{p\in\frac{2\pi}{L}\mathbb{Z}^3} \frac{\hat{u}^<(p)}{|p|^2}e^{ip\cdot (z-x)},
\] 
we can write
\begin{eqnarray}
  \left\|\int dz^\prime\varphi(z-z^\prime) a(\partial_j u^<_{z^\prime})a(\overline{v}^r_{z^\prime})\right\| &\leq& \int dz^\prime \,|\Delta \varphi(z-z^\prime)|\|a(\breve{u}^{<,j}_{z^\prime})\|\|a(\overline{v}^r_{z^\prime})\|\nonumber
  \\
  && +  \int dz^\prime\, |\varphi(z-z^\prime)|\|a(\breve{u}^{<,j}_{z^\prime})\|\|a(\Delta\overline{v}^r_{z^\prime})\| \nonumber\\
  && +2\sum_{m=1}^3 \int dz^\prime\, |\partial_m \varphi(z-z^\prime)|\|a(\breve{u}^{<,j}_{z^\prime})\|\|a(\partial_m\overline{v}^r_{z^\prime})\|,
\end{eqnarray}
and, similarly,
\begin{eqnarray}
  \left\|\int dz^\prime\varphi(z-z^\prime) a(u^<_{z^\prime})a(\partial_j\overline{{v}}^r_{z^\prime})\right\| &\leq& \int dz^\prime \,|\Delta \varphi(z-z^\prime)| \|a(\widetilde{u}^<_{z^\prime})\| \|a(\partial_j\overline{{v}}^r_{z^\prime})\|\nonumber \\
  && +\sum_{m=1}^3\int dz^\prime\, |\varphi(z-z^\prime)|\|a(\widetilde{u}^<_{z^\prime})\| \|a(\partial^2_m \partial_j\overline{{v}}^r_{z^\prime})\|\nonumber \\
  && +2\sum_{m=1}^3 \int dz^\prime\, |\partial_m \varphi(z-z^\prime)|\|a(\widetilde{u}^<_{z^\prime})\|\|a(\partial_m\partial_j\overline{{v}}^r_{z^\prime})\| .
\end{eqnarray}
Using then the following estimates,  
\[
  \|\breve{u}^{<,j}_{z^\prime}\|_2 \leq C, \qquad \|\widetilde{u}^<_{z^\prime}\|_2 \leq C\rho^{-\frac{1}{6}}, \qquad\| D^\alpha\overline{v}^r_{z^\prime}\|_2 \leq C\rho^{\frac{1}{2} + \frac{n}{3}},\quad\mbox{with}\,\,\, \alpha = (\alpha_1, \alpha_2,\alpha_3)\in \mathbb{N}^3_0, \,\,\,|\alpha| = n,
\]
together with the bounds for $\|\Delta\varphi\|_1, \|\nabla\varphi\|_1, \|\varphi\|_1$ from Lemma \ref{lem: bound phi}, we get 
\begin{equation}\label{eq: bound b<}
  \|b^<(\partial_j\varphi_z)\| \leq C\rho^{\frac{1}{2}}.
\end{equation}
To study $b^>( \partial_j\varphi_z)$, we can proceed as for the bound of $b(\varphi_z)$. More precisely, if we define
\[
  \widetilde{u}^>(z;x) = \frac{1}{L^3}\sum_{p\in\frac{2\pi}{L}\mathbb{Z}^3} \frac{\hat{u}^>(k)}{|k|^2}e^{ik\cdot (z-x)},
\]
then we can bound
\begin{eqnarray}
   \|b^>( \partial_j\varphi_z)\| &\leq& \sum_{j=1}^3\int\, dz^\prime\, |\partial^2_m\partial_j\varphi(z-z^\prime)|\|a(\widetilde{u}^>_{z^\prime})\|\|a(\overline{v}^r_{z^\prime})\| \nonumber
   \\
   &&+ 2\sum_{m=1}^3\int\,dz^\prime |\partial_m\partial_j\varphi(z-z^\prime)|\|a(\widetilde{u}^>_{z^\prime})\|\|a(\partial_m\overline{v}^r_{z^\prime})\|\nonumber
   \\
   && + \int\, dz^\prime |\partial_j\varphi(z-z^\prime)|\|a(\widetilde{u}^>_{z^\prime})\|\|a(\Delta\overline{v}^r_{z^\prime})\|,
\end{eqnarray}

Using now the fact that $\|\widetilde{u}^>\|_2 \leq C$, the bounds in \eqref{eq: L2 norm lemma b phi} together with the ones in Lemma \ref{lem: bound phi}, we get that 
\begin{equation}\label{eq: bound b>}
  \|b^>(\partial_j\varphi_z)\| \leq C\rho^{\frac{1}{2}}\|D^3\varphi\|_1 + C\rho^{\frac{1}{2} +\frac{1}{3}}\|D^2\varphi\|_1 + C\rho^{\frac{7}{6}}\|\nabla\varphi\|_1 \leq C\rho^{\frac{1}{2}}.
\end{equation}
Combining then $\eqref{eq: bound b<}$ and \eqref{eq: bound b>}, we get $\|b(\partial_j\varphi_{z})\|\leq C\rho^{1/2}$. The estimate for $\widetilde{b}(\partial_j\varphi_{z})$ can be done in the same way, we omit the details.
\end{proof}
Recall now that $\varphi_\infty$ does not satisfies the zero energy scattering equation exactly in $\mathbb{R}^3$: there is a correction due to the localization. It explicitly reads as $\mathcal{E}_{\{\varphi_0, \chi_{\sqrt[3]{\rho}}\}}^\infty(x) = -4\nabla \varphi_0(x) \nabla \chi_{\sqrt[3]{\rho}}(x) - 2\varphi_0(x) \Delta\chi_{\sqrt[3]{\rho}}(x)$. In what follows we denote by $\mathcal{E}_{\{\varphi_0, \chi_{\sqrt[3]{\rho}}\}}(x)$ the corresponding periodization in the box $\Lambda_L$, i.e., 
\begin{equation}\label{eq: period error scattering}
  \mathcal{E}_{\{\varphi_0, \chi_{\sqrt[3]{\rho}}\}}(x) = \sum_{n\in\mathbb{Z}^3}\mathcal{E}_{\{\varphi_0, \chi_{\sqrt[3]{\rho}}\}}^\infty(x + nL).
\end{equation}
\begin{lemma}\label{lem: b error scatt}
Let $\mathcal{E}_{\{\varphi_0, \chi_{\sqrt[3]{\rho}}\}}$ be as in \eqref{eq: period error scattering}. Let
\[
  b_\sigma((\mathcal{E}_{\{\varphi_0, \chi_{\sqrt[3]{\rho}}\}})_z) := \int\, dz^\prime\, \mathcal{E}_{\{\varphi_0, \chi_{\sqrt[3]{\rho}}\}}(z-z^\prime)b_{z^\prime,\sigma}, \qquad b_{z^\prime, \sigma} = a_\sigma(u^r_{z^\prime})a_\sigma(\overline{v}^r_{z^\prime})
\]
It holds
\begin{equation}\label{eq: bound b err scatt}
  \|b_\sigma((\mathcal{E}_{\{\varphi_0, \chi_{\sqrt[3]{\rho}}\}})_z)\| \leq C\rho,
\end{equation}
and the same holds for the adjoint operator $b_\sigma^\ast((\mathcal{E}_{\{\varphi_0, \chi_{\sqrt[3]{\rho}}\}})_z)$.
\end{lemma}
\begin{remark}
Note that the analogous proof of Lemma \ref{lem: b error scatt} gives that $\|b_\sigma (V_z)\|\leq C\rho^{1/3}$ (where $b_\sigma(V)= \int dz^\prime\, V(z-z^\prime) a_\sigma(u^r_{z^\prime})a_{\sigma}(\overline{v}^r_{z^\prime})$). The fact that the bound for $b_\sigma((\mathcal{E}_{\{\varphi_0, \chi_{\sqrt[3]{\rho}}\}})_z)$ is better (see \eqref{eq: bound b err scatt}) shows that the renormalized interaction $\mathcal{E}_{\{\varphi_0, \chi_{\sqrt[3]\rho}\}}$ is much smaller than the original one. 
\end{remark} 
\begin{proof}
The proof is as the one of Lemma \ref{lem: bound b phi}. Following the same calculation we can prove that, taking $L$  large enough,
\begin{equation}\label{eq: oper est err}
  \|b_\sigma((\mathcal{E}_{\{\varphi_0, \chi_{\sqrt[3]{\rho}}\}})_z)\| \leq C\rho^{\frac{1}{3}}\|\Delta\mathcal{E}_{\{\varphi_0, \chi_{\sqrt[3]{\rho}}\}}\|_{L^1(\Lambda_L)} + C\rho^{\frac{2}{3}}\|\nabla\mathcal{E}_{\{\varphi_0, \chi_{\sqrt[3]{\rho}}\}}\|_{L^1(\Lambda_L)} + C\rho \|\mathcal{E}_{\{\varphi_0, \chi_{\sqrt[3]{\rho}}\}}\|_{L^1(\Lambda_L)}.
\end{equation}
Moreover, using the estimates for $\varphi_0$ stated in Lemma \ref{lem: bound phi0}, the properties of $\chi_{\sqrt[3]\rho}$ and recalling \eqref{eq: period error scattering}, we have 
\begin{eqnarray}\label{eq: L1 norm err scatt}
  \|\mathcal{E}_{\{\varphi_0, \chi_{\sqrt[3]{\rho}}\}}\|_{L^1(\Lambda_L)} &\leq&   C\int_{\mathbb{R}^3}dx\, \nabla\varphi_0(x)\nabla \chi_{\sqrt[3]\rho}(x) + C\int_{\mathbb{R}^3}\varphi_0(x)\Delta \chi_{\sqrt[3]\rho}(x) 
  \\
  &\leq& C\rho^{\frac{1}{3}}\int_{\rho^{-\frac{1}{3}}\leq |x| \leq 2\rho^{-\frac{1}{3}}}dx\, \frac{1}{|x|^2} + C\rho^{\frac{2}{3}}\int_{\rho^{-\frac{1}{3}}\leq |x| \leq 2\rho^{-\frac{1}{3}}}dx\, \frac{1}{|x|} \leq C,\nonumber
\end{eqnarray}
where we used $\nabla^n\chi_{\sqrt[3]\rho}$ ($n\geq 1$) is non-vanishing only for $\rho^{-1/3}\leq |x| \leq 2\rho^{-1/3}$, that $\|\nabla^n\chi_{\sqrt[3]\rho}\|_\infty\leq C\rho^{n/3}$ and that in the support of $\nabla^n\chi_{\sqrt[3]\rho}$ the scattering energy solution is such that $\varphi_0(x) = a/|x|$ due to the compact support of the interaction potential (see Appendix \ref{app: scattering}). Proceeding similarly and using the bound proved in Lemma \ref{lem: bound phi}), we can also prove that 
\begin{equation}\label{eq: L1 norm der err scatt}
  \|\nabla \mathcal{E}_{\{\varphi_0, \chi_{\sqrt[3]{\rho}}\}}\|_{L^1(\Lambda_L)} \leq C\rho^{\frac{1}{3}}, \qquad \|\Delta\mathcal{E}_{\{\varphi_0, \chi_{\sqrt[3]{\rho}}\}}\|_{L^1(\Lambda_L)} \leq C\rho^{\frac{2}{3}}.
\end{equation}
Inserting then the bounds \eqref{eq: L1 norm err scatt} and \eqref{eq: L1 norm der err scatt} in \eqref{eq: oper est err}, the result holds.
\end{proof}
We now use Lemma \ref{lem: bound b phi} to prove a bound for the number operator, as shown in the next proposition. From now one we will work over states 
\begin{equation}\label{eq: def xi lambda}
  \xi_\lambda:= T^\ast_\lambda R^\ast \psi,
\end{equation}
with $\psi$ being an approximate ground state in the sense of Definition \ref{def: approx gs}.
\begin{proposition}[Improved estimate for $\mathcal{N}$]\label{pro: impr est N} Let $\lambda\in [0,1]$ and let $\xi_\lambda$ be a state as in \eqref{eq: def xi lambda}. Then 
\begin{equation}\label{eq: est N}
\langle\xi_\lambda, \mathcal{N}\xi_\lambda\rangle \leq C\langle\xi_1, \mathcal{N}\xi_1 \rangle + CL^{3}\rho^{\frac{5}{3}}.
\end{equation}
\end{proposition}
\begin{remark}[Simpler proof than in \cite{FGHP}] The proof of Proposition \ref{pro: impr est N}  which we propose here is simpler than the analogous one in \cite[Proposition 5.1]{FGHP}. This is due to the fact that in the present paper we are able to show an improved bound for the operator $b(\varphi_z)$ (see Lemma \ref{lem: bound b phi}). 
\end{remark}
\begin{proof} The proof of \eqref{eq: est N} is a direct consequence of Gr\"onwall's Lemma. We start by computing 
\begin{eqnarray}
  \partial_\lambda \langle \xi_\lambda, \mathcal{N}\xi_\lambda\rangle  &=&  -\langle\xi_\lambda, [\mathcal{N}, (B - B^\ast)]\xi_\lambda\rangle 
  \\
  &=& -4\int dzdz^\prime\, \varphi(z-z^\prime)\langle \xi_\lambda, a_\uparrow(u^r_z)a_\uparrow(\overline{v}^r_z) a_\downarrow(u^r_{z^\prime})a_\downarrow(\overline{v}^r_{z^\prime})\xi_\lambda\rangle + \mathrm{c.c.}\nonumber
  \\
  &=& 4\int dz \,\langle \xi_\lambda, b_\downarrow(\varphi_z)a_\uparrow(\overline{v}^r_z)a_\uparrow(u^r_{z})\xi_\lambda\rangle + \mathrm{c.c.},\nonumber
\end{eqnarray}
where we are using the notations introduced in \eqref{eq: def b phi}. From Lemma \ref{lem: bound b phi} and using the bound $\|a_\uparrow(\overline{v}^r_z)\|\leq \|\overline{v}^r_z\|_2 \leq C\rho^{1/2}$ (see Proposition \ref{pro: L1, L2 v}), we get, by Cauchy-Schwarz 
\begin{equation*}
   |\partial_\lambda \langle \xi_\lambda, \mathcal{N}\xi_\lambda\rangle| \leq C\rho^{\frac{1}{2} + \frac{1}{3}} \int dz\, \|a(u^r_z)\xi_\lambda\|\leq CL^{\frac{3}{2}}\rho^{\frac{5}{6}}\|\mathcal{N}^{\frac{1}{2}}\xi_\lambda\| \leq \langle \xi_\lambda, \mathcal{N} \xi_\lambda\rangle + CL^3 \rho^{\frac{5}{3}} .
\end{equation*}
Then we have
\begin{equation}\label{eq: est der N}
  \partial_\lambda\langle \xi_\lambda, \mathcal{N}\xi_\lambda\rangle \leq C\langle \xi_\lambda, \mathcal{N}\xi_\lambda\rangle + CL^3\rho^{\frac{5}{3}}.
\end{equation}
Thus, using Gr\"onwall's Lemma, we get
\begin{equation}\label{eq: 1}
  \langle \xi_\lambda, \mathcal{N}\xi_\lambda\rangle \leq C\langle\xi_1, \mathcal{N}\xi_1\rangle + CL^3 \rho^{\frac{5}{3}}. 
\end{equation}
\end{proof}
\begin{remark}[Non optimal bound for $\mathcal{N}$] From \eqref{eq: est der N}, by Gronwall's Lemma, we get
\begin{equation}\label{eq: non optimal bound N}
   \langle \xi_\lambda, \mathcal{N}\xi_\lambda\rangle \leq \langle\xi_0, \mathcal{N}\xi_0\rangle + CL^3 \rho^{\frac{5}{3}} =  \langle R^\ast\psi, \mathcal{N} R^\ast \psi\rangle + CL^3 \rho^{\frac{5}{3}} \leq CL^3\rho^{\frac{7}{6}}, 
\end{equation}
where in the last inequality we used Lemma \ref{lem: a priori}.
\end{remark}
To estimate some of the error terms, we will often use the number operator counting the particles having momentum $k$ such that $|k| \geq 2k_F$, i.e.
\begin{equation}
  \mathcal{N}_> := \sum_\sigma\sum_{|k|\geq 2k_F} \hat{a}^\ast_{k,\sigma}\hat{a}_{k,\sigma},
\end{equation}
\begin{remark}[On the operator $\mathcal{N}_>$]\label{rem: N_>} Note that the new way of regularizing $\hat{u}^r$ proposed in this paper implies that $\hat{u}^r(k) = 0$ for $|k| \leq 2k_F$. This allows us to estimate many error terms with respect to the operator $\mathcal{N}_>$. Indeed one can easily see that 
\begin{equation}
\int_{\Lambda_L}dx\, \|a_\sigma(u_x^r)\psi\|^2 \leq C\sum_{|k|\geq 2k_F}\langle \psi, \hat{a}_{k,\sigma}^\ast \hat{a}_{k,\sigma}\psi\rangle.
\end{equation}
The advantage of that is that we can bound the expectation of $\mathcal{N}_>$ with respect to the kinetic operator $\mathbb{H}_0$ uniformly in the volume (see \eqref{eq: N> wrt H0}). In this way we improve many estimates when deriving the lower bound.
\end{remark}
\begin{proposition}[Bound for $\mathcal{N}_>$]\label{lem: bound N>} Let $\lambda\in [0,1]$ and let $\xi_\lambda$ be a state as in \eqref{eq: def xi lambda}. It holds true that 
\begin{equation}\label{eq: N> lambda}
  \langle\xi_\lambda, \mathcal{N}_> \xi_\lambda\rangle \leq C\rho^{-\frac{2}{3}}\langle \xi_\lambda, \mathbb{H}_0 \xi_\lambda\rangle.
\end{equation}
Moreover, we have
\begin{equation}\label{eq: est N>}
  \langle\xi_\lambda, \mathcal{N}_>\xi_\lambda\rangle \leq C\rho^{-\frac{2}{3}}\langle \xi_1,\mathbb{H}_0 \xi_1\rangle + CL^3\rho^{\frac{5}{3}}.
\end{equation}
\end{proposition}
\begin{proof}The proof of \eqref{eq: N> lambda} is a consequence of the fact that for each $k$ such that $|k| \geq 2k_F$, it holds $||k|^2 - k_F^2| \geq C\rho^{\frac{2}{3}}$. Thus,
\begin{equation}\label{eq: N> wrt H0}
  \mathcal{N}_> = \sum_\sigma\sum_{|k| \geq 2k_F} \hat{a}^\ast_{k,\sigma}\hat{a}_{k,\sigma}\leq C\rho^{-\frac{2}{3}}\sum_\sigma\sum_{k}||k|^2 - k_F^2|\, \hat{a}^\ast_{k,\sigma}\hat{a}_{k,\sigma} = C\rho^{-\frac{2}{3}}\mathbb{H}_0.
\end{equation}
Moreover, repeating similar calculations as in Proposition \ref{pro: impr est N}, one gets 
\begin{equation}
   \partial_\lambda\langle\xi_\lambda, \mathcal{N}_>\xi_\lambda\rangle \leq C\langle \xi_\lambda, \mathcal{N}_>\xi_\lambda\rangle + CL^3\rho^{\frac{5}{3}} .
\end{equation}
It then follows, by Gronwall's Lemma combined with \eqref{eq: N> lambda} that 
\begin{equation}\label{eq: est N_>}
  \langle \xi_\lambda, \mathcal{N}_>\xi_\lambda \rangle \leq C\langle \xi_1, \mathcal{N}_>\xi_1\rangle + CL^{3}\rho^{\frac{5}{3}} \leq C\rho^{-\frac{2}{3}}\langle\xi_1, \mathbb{H}_0\xi_1\rangle + CL^3\rho^{\frac{5}{3}}.
\end{equation}
\end{proof}
Proceeding similarly as in Proposition \ref{lem: bound N>}, one gets the result below.
\begin{corollary}\label{cor: N tilde}
Let $\lambda\in[0,1]$. Let $\alpha > 1/3$. 
It holds true that 
\[
  \sum_\sigma\sum_{|k| \geq k_F + \rho^{\alpha}}\langle \xi_\lambda, \hat{a}^\ast_{k,\sigma}\hat{a}_{k,\sigma}\xi_\lambda\rangle \leq  \sum_\sigma\sum_{|k| \geq k_F + \rho^{\alpha}}\langle \xi_1, \hat{a}^\ast_{k,\sigma}\hat{a}_{k,\sigma}\xi_1\rangle + CL^3\rho^{\frac{5}{3}}.
\]
\end{corollary}
\begin{proof}
This is a consequence of Gronwall's Lemma. Indeed, proceeding as in Proposition \ref{lem: bound N>}, we have 
\[
  \sum_\sigma\sum_{|k| \geq k_F + \rho^{\alpha}}\partial_\lambda\langle \xi_\lambda, \hat{a}^\ast_{k,\sigma}\hat{a}_{k,\sigma}\xi_\lambda\rangle \leq C\langle \xi_\lambda, \mathcal{N}_> \xi_\lambda\rangle + CL^3\rho^{\frac{5}{3}}\leq \sum_\sigma\sum_{|k| \geq k_F + \rho^{\alpha}}\langle \xi_\lambda, \hat{a}^\ast_{k,\sigma}\hat{a}_{k,\sigma}\xi_\lambda\rangle + CL^3\rho^{\frac{5}{3}},
\]
using Gronwall's Lemma the result follows.
\end{proof}


\section{Propagation estimates}\label{sec: propagation estimates}
To prove Theorem \ref{thm: optimal up bd} and Theorem \ref{thm: lower bound}, we need to propagate the bounds in Lemma \ref{lem: a priori} via a Gr\"onwall-type argument. To do that, the next propositions are needed. 
\subsection{Bounds for $\mathbb{H}_0$ on interpolating states}
\begin{proposition}[Propagation estimate for ${}\mathbb{H}_0$ - Part I]\label{pro: H0} Let $\lambda\in [0,1]$ and let $\xi_\lambda$ be a state as in \eqref{eq: def xi lambda}. Under the same assumptions of Theorem \ref{thm: optimal up bd} and Theorem \ref{thm: lower bound}, it holds  
\begin{equation} 
  \frac{d}{d\lambda}\langle \xi_\lambda, \mathbb{H}_0 \xi_\lambda\rangle = -\langle \xi_\lambda, \mathbb{T}_1 \xi_\lambda\rangle + \mathcal{E}_{\mathbb{H}_0}(\xi_\lambda),
\end{equation}
where 
\begin{equation}\label{eq: def T1}
  \mathbb{T}_1 = 2\int dxdy\,  \Delta\varphi(x-y) a_\uparrow(u^r_x)a_\uparrow(\overline{v}^r_x)a_\downarrow(u^r_y) a_\downarrow(\overline{v}^r_y) + \mathrm{h.c.}
\end{equation}
and
\begin{equation}\label{eq: error kin energy}
  \mathcal{E}_{\mathbb{H}_0}(\xi_\lambda) = 2\sum_{j=1}^3\int dxdy\, \partial_j \varphi(x-y) \left[a_\uparrow(u^r_x)a_\uparrow(\partial_j\overline{v}^r_x)a_\downarrow(u^r_y) a_\downarrow(\overline{v}^r_y) - a_\uparrow(u^r_x)a_\downarrow(\overline{v}^r_x)a_\downarrow(u^r_y) a_\downarrow(\partial_j\overline{v}^r_y)\right] + \mathrm{h.c.}
\end{equation}
\end{proposition}
\begin{proof} 
  We compute 
  \begin{equation*} 
    \frac{d}{d\lambda}\langle \xi_\lambda, \mathbb{H}_0\xi_\lambda\rangle = - \langle \xi_\lambda, [\mathbb{H}_0,B]\xi_\lambda \rangle + \mathrm{c.c.}
  \end{equation*}
  Recall that 
  \begin{equation*}
    \mathbb{H}_0 = \sum_\sigma\sum_k \varepsilon(k) \hat{a}_{k,\sigma}^\ast\hat{a}_{k,\sigma}, \qquad \varepsilon(k) := ||k|^2 - (k_F^\sigma)^2|
.  \end{equation*}
  It is then convenient to write $B$ in momentum space, we have 
  \begin{equation*} 
    B = \frac{1}{L^3}\sum_k\sum_{s,s^\prime}\hat{\varphi}(k)\hat{u}^{r}(s+k)\hat{u}^r(s^\prime - k)\hat{v}^{r}(s)\hat{v}^r (s^\prime)\,\hat{a}_{k+s,\uparrow}\hat{a}_{s,\uparrow}\hat{a}_{s^\prime -k, \downarrow}\hat{a}_{s^\prime,\downarrow} - \mathrm{h.c.}
  \end{equation*}
We calculate
\[
  [\mathbb{H}_0, B] = \frac{1}{L^3}\sum_\sigma\sum_{k,q}\sum_{s,s^\prime}\varepsilon(k)\hat{\varphi}(q) \hat{u}^{r}(s+q)\hat{u}^r(s^\prime - q)\hat{v}^{r}(s)\hat{v}^r (s^\prime)[\hat{a}^\ast_{k,\sigma} , \hat{a}_{q+s,\uparrow}\hat{a}_{s,\uparrow}\hat{a}_{s^\prime -q, \downarrow}\hat{a}_{s^\prime,\downarrow} ]\hat{a}_{k,\sigma} + \mathrm{h.c.}
\]
We have that 
\begin{equation*}
  [\hat{a}_{k,\sigma}^\ast, \hat{a}_{q+s,\uparrow}\hat{a}_{s,\uparrow}\hat{a}_{s^\prime -q, \downarrow}\hat{a}_{s^\prime,\downarrow} ] = [\hat{a}^\ast_{k,\sigma}, \hat{a}_{q+s,\uparrow}\hat{a}_{s,\uparrow}]\hat{a}_{s^\prime -q, \downarrow}\hat{a}_{s^\prime,\downarrow} +\hat{a}_{q+s,\uparrow}\hat{a}_{s,\uparrow} [\hat{a}^\ast_{k,\sigma}, \hat{a}_{s^\prime -q, \downarrow}\hat{a}_{s^\prime,\downarrow}],
\end{equation*}
and being 
\begin{equation*}\label{eq: comm [a*,b]}
 \sum_{s}  [\hat{a}^\ast_{k,\sigma}, \hat{a}_{s+q,\sigma^\prime}\hat{a}_{s,\sigma^\prime}] = \sum_{s}\delta_{\sigma,\sigma^\prime}\delta_{k,s+q} \hat{a}_{s,\sigma^\prime} - \sum_{s}\delta_{\sigma,\sigma^\prime}\delta_{k,s}\hat{a}_{s+q,\sigma^\prime},
\end{equation*}
we get
\begin{multline*}
[\mathbb{H}_0, B] = -\frac{1}{L^3}\sum_k\sum_{s,s^\prime} (\varepsilon(s+k) +\varepsilon (s) + \varepsilon (s^\prime - k) +\varepsilon(s^\prime))\hat{\varphi}(k)\, (\hat{u}^r(s+k) \hat{u}^r(s^\prime -k)\hat{v}^r(s)\hat{v}^r(s^\prime))\cdot
\\
\cdot\hat{a}_{s+k,\uparrow} \hat{a}_{s,\uparrow} \hat{a}_{s^\prime-k,\downarrow} \hat{a}_{s^\prime,\downarrow} + \mathrm{h.c.}
\end{multline*}
Using that $(s+k),(s^\prime-k) \notin\mathcal{B}_F$ and that $s,s^\prime\in \mathcal{B}_F$, we have
\begin{equation} 
  \varepsilon(s+k) +\varepsilon (s) + \varepsilon (s^\prime - k) +\varepsilon(s^\prime) = 2|k|^2 + 2k\cdot(s-s^\prime).
\end{equation}
We then define
\begin{equation} 
  \mathbb{T}_1 := -\frac{2}{L^3}\sum_k\sum_{s,s^\prime} |k|^2 \hat{\varphi}(k)   (\hat{u}^r(s+k) \hat{u}^r(s^\prime -k)\hat{v}^r(s)\hat{v}^r(s^\prime))\,\hat{a}_{s+k,\uparrow} \hat{a}_{s,\uparrow} \hat{a}_{s^\prime-k,\downarrow} \hat{a}_{s^\prime,\downarrow} + \mathrm{h.c.},
\end{equation}
which yields
\begin{equation*}
  \mathbb{T}_1 = 2\int dxdy\, \Delta\varphi(x-y) a_\uparrow(u^r_x)a_\uparrow(\overline{v}^r_x)a_\downarrow(u^r_y) a_\downarrow(\overline{v}^r_y) + \mathrm{h.c.}
\end{equation*}
We can then write 
\begin{equation*} 
  \frac{d}{d\lambda} \langle\xi_\lambda, \mathbb{H}_0\xi_\lambda\rangle = -\langle\xi_\lambda, \mathbb{T}_1\xi_\lambda\rangle + \mathcal{E}_{\mathbb{H}_0}(\xi_\lambda) ,
\end{equation*}
where
\begin{multline*}
  \mathcal{E}_{\mathbb{H}_0}(\xi_\lambda) :=  -\frac{2}{L^3}\sum_k\sum_{s,s^\prime} (k\hat{\varphi}(k)\cdot (s-s^\prime)\hat{v}^r(s)\hat{v}^r(s^\prime))   \hat{u}^r(s+k) \hat{u}^r(s^\prime -k)\cdot
  \\
  \cdot\langle \xi_\lambda,\hat{a}_{s+k,\uparrow} \hat{a}_{s,\uparrow} \hat{a}_{s^\prime-k,\downarrow} \hat{a}_{s^\prime,\downarrow}\xi_\lambda\rangle + \mathrm{c.c.},
\end{multline*}
which can also be written as
\begin{multline*}
  \mathcal{E}_{\mathbb{H}_0}(\xi_\lambda) = 2\sum_{j=1}^3\int dxdy\, \partial_j\varphi(x-y) \cdot 
  \\
  \cdot \langle \xi_\lambda, (a_\uparrow(u^r_x)a_\uparrow(\partial_j\overline{v}^r_x)a_\downarrow(u^r_y) a_\downarrow(\overline{v}^r_y) - a_\uparrow(u^r_x)a_\uparrow(\overline{v}^r_x)a_\downarrow(u^r_y) a_\downarrow(\partial_j\overline{v}^r_y))\xi_\lambda\rangle + \mathrm{c.c.}\nonumber
\end{multline*}
\end{proof}
In the next lemma we bound the error term $\mathcal{E}_{\mathbb{H}_0}(\xi_\lambda)$ in \eqref{eq: error kin energy}.
\begin{lemma}[Estimate of $\mathcal{E}_{\mathbb{H}_0}(\xi_\lambda)$]\label{lem: error kinetic energy} Let $\lambda\in [0,1]$ and let $\xi_\lambda$ be a state as in \eqref{eq: def xi lambda}. Let $\mathcal{E}_{\mathbb{H}_0}(\xi_\lambda)$ be as in \eqref{eq: error kin energy}. Let $0<\epsilon\leq 1$. Under the same assumptions of Theorem \ref{thm: optimal up bd} and Theorem \ref{thm: lower bound}, for $0 <\delta <1$, it holds 
\begin{eqnarray}\label{eq: est err kin}
  |\mathcal{E}_{\mathbb{H}_0}(\xi_\lambda)| &\leq& \delta\langle\xi_1, \mathbb{H}_0\xi_1\rangle + CL^3\rho^{\frac{7}{3}}  + C\rho^{\frac{7}{6} -\frac{\epsilon}{3}}\langle \xi_\lambda, \mathcal{N}_> \xi_\lambda\rangle + C\rho^{\frac{2}{3}-\frac{\epsilon}{3}}\langle \xi_\lambda, \mathbb{H}_0 \xi_\lambda\rangle  + C\rho\langle \xi_\lambda, \mathcal{N}\xi_\lambda\rangle\nonumber
  \\
  && + C\rho^{\frac{5}{6} -\frac{\epsilon}{3}}\|\mathcal{N}_>^{\frac{1}{2}}\xi_\lambda\|\big(\|\widetilde{\mathbb{Q}}_1^{\frac{1}{2}}\xi_\lambda\| + \|\mathbb{H}_0^{\frac{1}{2}}\xi_\lambda\|\big).
\end{eqnarray}
\end{lemma}

\begin{remark}[Comparison with \cite{FGHP}]\label{rem: comparison kinetic energy} In Lemma \ref{lem: error kinetic energy} we get refined estimates with respect to the ones in \cite[Proposition 5.5]{FGHP}. More precisely, as a consequence of Lemma \ref{lem: error kinetic energy}, we will prove that, up to a positive contribution,\footnote{Note that in the upper bound we do not have this term thanks to the choice of the trial state.} i.e., $\langle \xi_1, \mathbb{H}_0 \xi_1\rangle$, $\mathcal{E}_{\mathbb{H}_0}(\xi_\lambda) = \mathcal{O}(L^3\rho^{7/3})$. Note that in the proof of Lemma \ref{lem: error kinetic energy} we explicitly use that $\varphi_\infty$ satisfies $2\Delta\varphi_\infty + V(1-\varphi_\infty) = \mathcal{E}_\infty(\varphi_0, \chi_{\sqrt[3]\rho})$ in $\mathbb{R}^3$ (see \eqref{eq: equation for tilde phi}): this plays a crucial role. For the sake of completeness, we recall that in \cite{FGHP}, $\mathcal{E}_{\mathbb{H}_0}(\xi_\lambda)$ was bounded as
\[
  -CL^3 \rho^{2+\frac{1}{9}}\leq \mathcal{E}_{\mathbb{H}_0}(\xi_\lambda)\leq CL^3\rho^{2+ \frac{2}{9}}.
\]
\end{remark}
\begin{proof}[Proof of Lemma \ref{lem: error kinetic energy}] Recall that 
\begin{multline*}
  \mathcal{E}_{\mathbb{H}_0}(\xi_\lambda) = 2\sum_{j=1}^3\int dxdy\, \partial_j \varphi(x-y) \cdot 
  \\
  \cdot \langle \xi_\lambda, (a_\uparrow(u^r_x)a_\uparrow(\partial_j\overline{v}^r_x)a_\downarrow(u^r_y) a_\downarrow(\overline{v}^r_y) - a_\uparrow(u^r_x)a_\uparrow(\overline{v}^r_x)a_\downarrow(u^r_y) a_\downarrow(\partial_j\overline{v}^r_y))\xi_\lambda\rangle + \mathrm{c.c.}\nonumber
\end{multline*}
Both the two contributions  above can be estimated in the same way; we focus on the first one.
We first define 
\begin{equation}\label{eq: def A}
\mathbb{A}:=  2\sum_{j=1}^3\int\, dxdy\, \partial_j \varphi(x-y)\, a_\uparrow^\ast(u^r_x)a_\uparrow^\ast(\partial_j\overline{v}^r_x)a_\downarrow^\ast(u^r_y) a_\downarrow^\ast(\overline{v}^r_y).
\end{equation}
It is convenient to write:
\begin{equation}\label{eq: prop A}
  \langle \xi_\lambda, \mathbb{A} \xi_\lambda \rangle  = \langle \xi_1, \mathbb{A}\xi_1\rangle  + \int_1^{\lambda} d\lambda^\prime \frac{d}{d\lambda}\langle\xi_{\lambda^\prime}, \mathbb{A}\xi_{\lambda^\prime}\rangle = \langle \xi_1, \mathbb{A}\xi_1\rangle  - \int_1^{\lambda} d\lambda^\prime\langle\xi_{\lambda^\prime}, [\mathbb{A}, B]\xi_{\lambda^\prime}\rangle.
\end{equation}
Before calculating the commutator, we estimate $\langle\xi_1, \mathbb{A}\xi_1\rangle$. 
We can rewrite $\mathbb{A}$ as follows: 
\begin{eqnarray}\label{eq: 1st int parts err H0}
 \mathbb{A} &=& 2\sum_{j=1}^3\int dxdy\, \partial_j \varphi(x-y)\, a_\uparrow^\ast(u^r_x)a_\uparrow^\ast(\partial_j\overline{v}^r_x)a_\downarrow^\ast(u^r_y) a_\downarrow^\ast(\overline{v}^r_y)\nonumber
 \\
 &=& -2\int dxdy\, \varphi(x-y)\, a_\uparrow^\ast (u^r_x) a_\uparrow^\ast(\Delta\overline{v}^r_x)a_\downarrow^\ast(u^r_y) a_\downarrow^\ast(\overline{v}^r_y)\nonumber \\
  && + 2\sum_{j=1}^3\int dxdy\, \varphi(x-y) \, a_\uparrow^\ast(\partial_j u^r_x) a_\uparrow^\ast(\partial_j\overline{v}^r_x)a_\downarrow^\ast(u^r_y) a_\downarrow^\ast(\overline{v}^r_y) \nonumber
  \\
  &\equiv&\mathbb{A}_1 + \mathbb{A}_2 .
\end{eqnarray}
We then do Cauchy-Schwarz to get, for some\footnote{Here and in the following we write $\delta$ to denote a positive constant strictly smaller than $1$. This choice will be convenient later.} $0<\delta<1$,
\begin{eqnarray}\label{eq: IA}
  |\langle\xi_1, \mathbb{A}_1\xi_1\rangle| &\leq& 2\int dx\, \|b_\downarrow(\varphi_x)a_\uparrow(\Delta\overline{v}^r_x)a_\uparrow(u^r_x)\xi_1\| \leq CL^{\frac{3}{2}}\rho^{\frac{1}{3}+ \frac{1}{2} + \frac{2}{3}} \|\mathcal{N}^{\frac{1}{2}}_>\xi_1\| 
  \\
  &\leq&  CL^{\frac{3}{2}}\rho^{\frac{3}{2} -\frac{1}{3}}\|\mathbb{H}_0^{\frac{1}{2}}\xi_1\| \leq  CL^3\rho^{\frac{7}{3}} + \delta\langle \xi_1, \mathbb{H}_0\xi_1, \rangle\nonumber,
\end{eqnarray}
where we used Lemma \ref{lem: bound b phi} and \eqref{eq: N> lambda}. We now estimate $\langle\xi_1, \mathbb{A}_2 \xi_1\rangle$. Again from Lemma \ref{lem: bound b phi}, we have 
\begin{equation*}
  |\langle \xi_1, \mathbb{A}_2\xi_1\rangle | \leq C\sum_{j=1}^3\int dx\, \|b_\downarrow(\varphi_x) a_\uparrow(\partial_j\overline{v}^r_x)a_\uparrow(\partial_j u^r_x)\xi_1\| \leq C\rho^{\frac{2}{3} + \frac{1}{2}}\sum_{j=1}^3\int \, dx\, \|a_\uparrow(\partial_j u^r_x)\xi_1\|,
\end{equation*}
where we also used that $\|a_\uparrow(\partial_j\overline{v}^r_x)\|\leq C\rho^{1/2 + 1/3}$.
Now we can proceed as follows
\begin{equation}\label{eq: est grad u kin}
  \int dx \, \|a_\uparrow(\partial_j u^{r}_x)\xi_1\|^2 \leq \sum_{|k| \geq 2k_F}|k|^2\| \hat{a}_{k,\uparrow}\xi_1 \|^2 \leq \langle \xi_1, \mathbb{H}_0 \xi_1 \rangle + C\rho^{\frac{2}{3}}\langle \xi_1 \mathcal{N}_> \xi_1\rangle\leq C\langle\xi_1,\mathbb{H}_0 \xi_1\rangle,
\end{equation}
where we used again \eqref{eq: N> lambda}. Then, by Cauchy-Schwarz, we get
\begin{equation}\label{eq: A2 xi1}
  |\langle \xi_1, \mathbb{A}_2\xi_1\rangle| \leq CL^{\frac{3}{2}}\rho^{\frac{2}{3} + \frac{1}{2}}\|\mathbb{H}_0^{\frac{1}{2}}\xi_1\| \leq C L^3 \rho^{\frac{7}{3}} + \delta\langle \xi_1,\mathbb{H}_0\xi_1\rangle,
\end{equation}
for some $0<\delta <1$. We now bound the second term in the right side of \eqref{eq: prop A}. We then have to calculate the following commutator
\begin{multline*}
   [\mathbb{A}, B - B^\ast] = \sum_{j=1}^3\int dxdydzdz^\prime \partial_j\varphi(x-y)\varphi(z-z^\prime) \, 
   \\
   \cdot [a_\uparrow^\ast( u^r_x) a_\downarrow^\ast(u^r_y)a_\downarrow^\ast(\overline{v}^r_y)a_\uparrow^\ast(\partial_j\overline{v}^r_x) ,a_\uparrow(u^r_z)a_\uparrow(\overline{v}^r_z)a_\downarrow(u^r_{z^\prime})a_\downarrow(\overline{v}^r_{z^\prime})] + \mathrm{h.c.}
\end{multline*}
We rewerite the commutator above as
\begin{eqnarray}\label{eq: comm B A}
  &&[a_\uparrow^\ast(u^r_x) a_\downarrow^\ast(u^r_y)a_\downarrow^\ast(\overline{v}^r_y)a_\uparrow^\ast(\partial_j\overline{v}^r_x) ,a_\uparrow(u^r_z)a_\uparrow(\overline{v}^r_z)a_\downarrow(u^r_{z^\prime})a_\downarrow(\overline{v}^r_{z^\prime})] 
  \\
  &&\qquad=  -a_\uparrow^\ast( u^r_x) a_\downarrow^\ast(u^r_y)[a_\downarrow^\ast(\overline{v}^r_y)a_\uparrow^\ast(\partial_j\overline{v}^r_x),a_\uparrow(\overline{v}^r_z)a_\downarrow(\overline{v}^r_{z^\prime})]a_\uparrow(u^r_z)a_\downarrow(u^r_{z^\prime})\nonumber 
  \\
  && \qquad \quad- a_\uparrow(\overline{v}^r_z)a_\downarrow(\overline{v}^r_{z^\prime})[a_\uparrow^\ast(u^r_x) a_\downarrow^\ast(u^r_y), a_\uparrow(u^r_z)a_\downarrow(u^r_{z^\prime})]a^\ast_\downarrow(\overline{v}^r_y)a^\ast_\uparrow(\partial_j\overline{v}^r_x)\nonumber.
\end{eqnarray}
In what follows, we estimate differently the error terms. We start by the first line in the right side of \eqref{eq: comm B A}. More precisely, we compute
\begin{eqnarray}
  [a_\downarrow^\ast(\overline{v}^r_y)a_\uparrow^\ast(\partial_j\overline{v}^r_x),a_\uparrow(\overline{v}^r_z)a_\downarrow(\overline{v}^r_{z^\prime})] &=& \partial_j\omega^r_\uparrow(x;z)a^\ast_\downarrow(\overline{v}^r_y)  a_\downarrow (\overline{v}^r_{z^\prime}) - a_\uparrow(\overline{v}^r_z)\omega^r_\downarrow(y;z^\prime)a^\ast_\uparrow(\partial_j\overline{v}^r_x)
  \\
  &=& -  \partial_j\omega^r_\uparrow(x;z)\omega^r_\downarrow (y;z^\prime) + \partial_j\omega^r_\uparrow(x;z)a^\ast_\downarrow(\overline{v}^r_y) a_\downarrow (\overline{v}^r_{z^\prime})\nonumber 
  \\
  && + \omega^r_\downarrow(y;z^\prime) a^\ast_\uparrow(\partial_j\overline{v}^r_x) a_\uparrow(\overline{v}^r_z)\nonumber.
\end{eqnarray}
Correspondingly we have to estimate three different error terms:
\begin{equation}
  \mathrm{I}_a :=  \sum_{j=1}^3\int dxdydzdz^\prime\, \partial_j\varphi(x-y)\varphi(z-z^\prime)\partial_j\omega^r_\uparrow(x;z) \omega^r_\downarrow (y;z^\prime) \langle\xi_\lambda, a^\ast_\uparrow(u^r_x)a^\ast_\downarrow(u^r_y)  a_\uparrow(u^r_z)a_\downarrow(u^r_{z^\prime})\xi_\lambda\rangle,
\end{equation}
\begin{equation}\label{eq: Ib}
  \mathrm{I}_b := -\sum_{j=1}^3\int dxdydzdz^\prime\, \partial_j \varphi(x-y)\varphi(z-z^\prime)\partial_j\omega^r_\uparrow(x;z) \langle\xi_\lambda, a^\ast_\uparrow( u^r_x)a^\ast_\downarrow(u^r_y) a_\downarrow^\ast(\overline{v}^r_y)a_\downarrow(\overline{v}^r_{z^\prime}) a_\uparrow(u^r_z)a_\downarrow(u^r_{z^\prime})\xi_\lambda\rangle.
\end{equation}
\begin{equation}\label{eq: Ic}
  \mathrm{I}_c := -\sum_{j=1}^3\int dxdydzdz^\prime\, \partial_j \varphi(x-y)\varphi(z-z^\prime)\omega^r_\downarrow(y;z^\prime) \langle\xi_\lambda, a^\ast_\uparrow(u^r_x)a^\ast_\downarrow(u^r_y)a^\ast_\uparrow(\partial_j\overline{v}^r_x) a_\uparrow(\overline{v}^r_z) a_\uparrow(u^r_z)a_\downarrow(u^r_{z^\prime})\xi_\lambda\rangle.
\end{equation}
We start by estimating $\mathrm{I}_a$. To do that, it is convenient to use the following bound for all $j=1,2,3$:
\begin{equation}\label{eq: est oper err H0}
  \left\|\int dxdz\, \partial_j\varphi (x-y)\varphi(z-z^\prime) \partial_j\omega^r_\uparrow(z;x) a^\ast_\uparrow(u^r_x) a_\uparrow(u^r_{z})\right\| \leq C\rho^{\frac{7}{6}}.
\end{equation}
To get the estimate \eqref{eq: est oper err H0}, one can proceed similarly as in Lemma \ref{lem: bound b phi}. In particular, we can rewrite 
\[
  \partial_j\omega^r_\uparrow(x;z) = -\{a^\ast_\uparrow(\overline{v}^r_x), a_\uparrow(\partial_j \overline{v}^r_z)\} = -a^\ast_\uparrow(\overline{v}^r_x) a_\uparrow(\partial_j \overline{v}^r_z) - a_\uparrow(\partial_j \overline{v}^r_z)a^\ast_\uparrow(\overline{v}^r_x),
\]
which implies that 
\begin{multline}
   \int dxdz \,\partial_j\varphi(x-y)\varphi(z-z^\prime) \partial_j\omega^r_\uparrow(x;z) a^\ast_\uparrow(u^r_x) a_\uparrow(u^r_{z}) 
   \\
   = b^\ast_\uparrow(\partial_j\varphi_y)\widetilde{b}_\uparrow(\varphi_{z^\prime}) - \widetilde{b}_\uparrow(\varphi_{z^\prime}) b^\ast_\uparrow(\partial_j\varphi_y) -   \int dxdz\, \partial_j\varphi(x-y)\varphi(z-z^\prime)(u^r)^2_\uparrow(z;x) a_\uparrow(\partial_j\overline{v}^r_z)a_\uparrow^\ast(\overline{v}^r_x),
\end{multline}
where $(u^r)^2_\uparrow(z;x)$ is defined as in Remark \ref{rem: u2 and tilde omega}.
Using then Lemma \ref{lem: bound b phi} (to bound $\|\widetilde{b}_\uparrow(\varphi_{z^\prime})\|\leq C\rho^{2/3}$ and $\|b^\ast_\uparrow(\partial_j\varphi_y)\|\leq C\rho^{1/2}$), together with Lemma \ref{lem: bound phi} and Proposition \ref{pro: L1, L2 v}, we find 
\begin{multline}
  \left\|\int dxdz\, \partial_j\varphi(x-y)\varphi(z-z^\prime) \partial_j\omega^r_\uparrow(x;z) a^\ast_\uparrow(u^r_x) a_\uparrow(u^r_{z}) \right\| 
  \\
  \leq 2\|\widetilde{b}_\uparrow(\varphi_{z^\prime})\|\|b^\ast_\uparrow(\partial_j\varphi_y)\| + \int dxdz\, |\partial_j\varphi(x-y)||\varphi(z-z^\prime)| |(u^r)^2_\uparrow(z;x)|\|a(\partial_j \overline{v}^r_x)\|\|a^\ast_\uparrow (\overline{v}^r_x)\|
  \\ \leq C\rho^{\frac{7}{6}}+ C\rho^{1+\frac{1}{3}}\left(\int dxdz\, |\nabla\varphi(x-y)|^2 |(u^r)^2_\uparrow(z;x)|\right)^{\frac{1}{2}}\left(\int dxdz\, |\varphi(z-z^\prime)|^2 |(u^r)^2_\uparrow(z;x)|\right)^{\frac{1}{2}}
  \\
  \leq C\rho^{\frac{7}{6}} + C\rho^{\frac{4}{3}}\|\nabla\varphi\|_2 \|\varphi\|_2 \leq C\rho^{\frac{7}{6}},
\end{multline}
where we also used that $\|a_\uparrow(\partial_j \overline{v}^r_z)\|\leq \|\partial_j \overline{v}^r_z\|_2 \leq C\rho^{1/2 + 1/3}$, $\|a_\uparrow^\ast(\overline{v}^r_x)\|\leq \|\overline{v}^r_x\|_2 \leq \rho^{1/2}$ and that $\|(u^r)^2_{x,\sigma}\|_1\leq C$ (see Remark \ref{rem: u2 and tilde omega}).
Then we have, 
\begin{equation}
  |\mathrm{I}_a | \leq C\rho^{\frac{7}{6}} \int dydz^\prime\, |\omega^r_\downarrow(y;z^\prime)| \|a_\downarrow(u^r_y)\xi_\lambda\|\|a_\downarrow (u^r_{z^\prime})\xi_\lambda\| \leq C\rho^{\frac{7}{6} -\frac{\epsilon}{3}}\langle\xi_\lambda, \mathcal{N}_>\xi_\lambda\rangle\nonumber
\end{equation}
where we also used \eqref{eq: est N>} to compare $\mathcal{N}_>$ with $\mathbb{H}_0$ and Proposition \ref{pro: L1, L2 v} to bound $\|\omega^r_{x,\sigma}\|_1\leq C\rho^{-\epsilon/3}$. We now estimate $\mathrm{I}_b$. We can rewrite $\mathrm{I}_b$ defined in \eqref{eq: Ib} as
\begin{eqnarray}
  \mathrm{I}_b \hspace{-0.2cm}&=&\hspace{-0.2cm} \int dxdydzdz^\prime\, \Delta\varphi(x-y)\varphi(z-z^\prime)\omega^r_\uparrow(x;z)\langle\xi_\lambda, a^\ast_\uparrow( u^r_x)a^\ast_\downarrow(u^r_y) a_\downarrow^\ast(\overline{v}^r_y)a_\downarrow(\overline{v}^r_{z^\prime}) a_\uparrow(u^r_z)a_\downarrow(u^r_{z^\prime})\xi_\lambda\rangle\nonumber
  \\
  &&\hspace{-0.3cm} - \sum_{j=1}^3\int dxdydzdz^\prime\, \partial_j\varphi(x-y)\varphi(z-z^\prime)\omega^r_\uparrow(x;z)\langle\xi_\lambda, a^\ast_\uparrow( \partial_j u^r_x)a^\ast_\downarrow(u^r_y) a_\downarrow^\ast(\overline{v}^r_y)a_\downarrow(\overline{v}^r_{z^\prime}) a_\uparrow(u^r_z)a_\downarrow(u^r_{z^\prime})\xi_\lambda\rangle \nonumber
  \\
  &\equiv& \mathrm{I}_{b;1} + \mathrm{I}_{b;2}.
\end{eqnarray}
Note that to bound $\mathrm{I}_{b;1}$ is important to use the scattering equation in \eqref{eq: equation for tilde phi}.
We first take into account the error term $\mathrm{I}_{b;1}$. To do that, we can use the periodicity of all the quantities we are integrating. Using then that for each $x\in \Lambda_L$, for $L$ large enough, $\Delta\varphi(x)= \Delta\varphi_\infty(x)$, $V(x) = V_\infty(x)$, $\varphi(x) = \varphi_\infty(x)$ and $\mathcal{E}_{\varphi_0, \chi_{\sqrt[3]\rho}}(x) = \mathcal{E}^\infty_{\varphi_0, \chi_{\sqrt[3]\rho}}(x)$, we get 
\begin{multline*}
  \mathrm{I}_{b;1} = -\int dxdydz\, \Delta\varphi(x-y) \omega^r_\uparrow(x;z) \langle a_\downarrow (\overline{v}^r_y) a_\uparrow(u^r_x)a_\downarrow(u^r_y)\xi_\lambda, b_\downarrow(\varphi_z)a_\uparrow(u^r_z)\xi_\lambda\rangle
  \\
  = -\frac{1}{2}\int dxdydz\, (\mathcal{E}_{\{\varphi_0,\chi_{\sqrt[3]\rho}\}} - V(1-\varphi))(x-y) \omega^r_\uparrow(x;z)\langle a_\downarrow (\overline{v}^r_y) a_\uparrow(u^r_x)a_\downarrow(u^r_y)\xi_\lambda, b_\downarrow(\varphi_z)a_\uparrow(u^r_z)\xi_\lambda\rangle 
\\
=: \mathrm{I}_{b;1;1} + \mathrm{I}_{b;1;2}.
\end{multline*}
We now bound $\mathrm{I}_{b;1;1}$. We have 
\begin{equation}
  \mathrm{I}_{b;1;1} = -\frac{1}{2}\int dxdz \,\omega^r_\uparrow(x;z)\langle b_\downarrow((\mathcal{E}_{\varphi_0, \chi_{\sqrt[3]\rho}})_x)a_\uparrow(u^r_x)\xi_\lambda, b_\downarrow(\varphi_z)a_\uparrow(u^r_z)\xi_\lambda\rangle.
\end{equation}
Thus, from Lemma \ref{lem: bound b phi} and Lemma \ref{lem: b error scatt}, we get 
\begin{eqnarray}
  |\mathrm{I}_{b;1;1}| &\leq& C\int dxdz |\omega^r_\uparrow(x;z)| \|b_\downarrow((\mathcal{E}_{\varphi_0, \chi_{\sqrt[3]\rho}})_x)\|\|b_\downarrow(\varphi_z)\| \|a_\uparrow(u^r_x)\xi_\lambda\|\|a_\uparrow(u^r_z)\xi_\lambda\|
\\
  &\leq& C\rho^{1+ \frac{1}{3}-\frac{\epsilon}{3}} \langle \xi_\lambda, \mathcal{N}_> \xi_\lambda\rangle \leq C\rho^{\frac{2}{3} -\frac{\epsilon}{3}}\langle \xi_\lambda, \mathbb{H}_0 \xi_\lambda\rangle,
\end{eqnarray}
where we used that $\mathcal{N}_> \leq C\rho^{-2/3}\mathbb{H}_0$ and Proposition \ref{pro: L1, L2 v}. We now estimate the contribution coming form $\mathrm{I}_{b;1;2}$.
Before proceeding further, we replace in the term $\mathrm{I}_{b;1;2}$ the operators $a_\uparrow(u^r_x)$, $a_\downarrow(u^r_y)$ with $a_\uparrow (u_x)$, $a_\downarrow(u_y)$. More precisely, we write 
\begin{equation}\label{eq: a+b cut-off}
  \mathrm{I}_{b;1;2} =  \frac{1}{2}\int dxdydz\, V(1-\varphi)(x-y) \omega^r_\uparrow(x;z) \langle a_\downarrow (\overline{v}^r_y) a_\uparrow(u_x)a_\downarrow(u_y)\xi_\lambda, b_\downarrow(\varphi_z)a_\uparrow(u^r_z)\xi_\lambda\rangle + \mathrm{B} \equiv \mathrm{A} + \mathrm{B},
\end{equation}
where $\mathrm{B}$ is an error term due to the replacement. We start by estimating $\mathrm{A}$. Using that $\varphi$ is bounded (see Lemma \ref{lem: bound phi}) we have that
\begin{equation}\label{eq: Ib11}
  |\mathrm{A}| \leq C\rho^{\frac{1}{2} + \frac{1}{3}}\int dxdydz V(x-y)|\omega^r_\uparrow(x;z)|\|a_\uparrow(u_x)a_\downarrow(u_y)\xi_\lambda\|\|a_\uparrow(u^r_z)\|\leq C\rho^{\frac{1}{2} + \frac{1}{3}-\frac{\epsilon}{3}}\|\widetilde{\mathbb{Q}}_{1}^{\frac{1}{2}}\xi_\lambda\|\|\mathcal{N}_>^{\frac{1}{2}}\xi_\lambda\| 
\end{equation}
where we also used Proposition \ref{pro: L1, L2 v}, Lemma \ref{lem: bound b phi} and \eqref{eq: est N>}. We now take into account the error term $\mathrm{B}$. The estimate of this contribution is similar to the one of some error terms in \cite{FGHP} (see, e.g., \cite[Appendix C.2]{FGHP}). To begin, we write $\hat{u}_\sigma(k) = \hat{\alpha}_\sigma(k) + \hat{u}^r_\sigma(k) + \hat{\delta}^>_\sigma(k)$, where $\hat{\alpha}_\sigma$ is supported for $k\notin \mathcal{B}_F^\sigma$ and $|k| \leq 3k_F^\sigma$ and $\hat{\delta}^>_\sigma$ is supported for $|k| \geq \rho^{-\beta}$.
 The case in which at  least one between $a_\uparrow (u^r_x)$ and $a_\downarrow(u^r_y)$ is replaced by $a_\uparrow(\delta^>_x)$ or $a_\downarrow(\delta^>_y)$ is discussed in Appendix \ref{app: cut-off}. In this case, all the errors can be made $o(\rho^{7/3})$ and thus are negligible. We now take into account all the other types of error terms in $\mathrm{B}$. We have to study the case in which at least one between $a_\uparrow (u^r_x)$ and $a_\downarrow(u^r_y)$ is replaced by $a_\uparrow(\alpha_x)$ or $a_\downarrow(\alpha_y)$. We bound then, for instance,
\begin{equation}
  \widetilde{\mathrm{B}}_{1} = \int dxdydz\, V(1-\varphi)(x-y) \omega^r_\uparrow(x;z) \langle a_\uparrow (\overline{v}^r_y) a_\uparrow(\alpha_x)a_\downarrow(u_y)\xi_\lambda, b_\downarrow(\varphi_z)a_\uparrow(u^r_z)\xi_\lambda\rangle.
\end{equation}
We then have 
\begin{eqnarray}
  |\widetilde{\mathrm{B}}_{1}| &\leq& C\int dxdydz\, V(x-y) |\omega^r_\uparrow(x;z)| \|a_\uparrow (\overline{v}^r_y)\|\|a_\uparrow(\alpha_x)\|\|b_\downarrow(\varphi_z)\|\|a_\downarrow(u_y)\xi_\lambda\|\|a_\uparrow(u^r_z)\xi_\lambda\|\nonumber
  \\
  &\leq& C\rho^{1+\frac{1}{3}-\frac{\epsilon}{3}}\langle \xi_\lambda, \mathcal{N}\xi_\lambda\rangle,
\end{eqnarray}
where we used that $\|V(1-\varphi)\|_1 \leq C\|\varphi\|_\infty\|V\|_1 \leq C$ (see Lemma \ref{lem: bound phi}), $\|a_\uparrow(\alpha_x)\|\leq \|\alpha_{x,\uparrow}\|\leq C\rho^{1/2}$ together with Proposition \ref{pro: L1, L2 v}. We then find that,
\begin{eqnarray}
  |\mathrm{I}_{b;1}| &\leq& |\mathrm{I}_{b;1;1}| + |\mathrm{I}_{b;1;2}|
  \\
  &\leq & C\rho^{\frac{2}{3} -\frac{\epsilon}{3}}\langle \xi_\lambda, \mathbb{H}_0\xi_\lambda\rangle + C\rho^{\frac{5}{6} -\frac{\epsilon}{3}}\|\widetilde{\mathbb{Q}}_1^{\frac{1}{2}}\xi_\lambda\|\|\mathcal{N}_>^{\frac{1}{2}}\xi_\lambda\| + C\rho^{1+\frac{1}{3} -\frac{\epsilon}{3}}\langle \xi_\lambda, \mathcal{N}\xi_\lambda\rangle\nonumber
\end{eqnarray}

We now estimate $\mathrm{I}_{b;2}$. From Lemma \ref{lem: bound b phi}, we have
\begin{eqnarray} 
  |\mathrm{I}_{b;2}| &\leq& \sum_{j=1}^3\int dxdz^\prime \, |\omega^r_\uparrow(x;z)| \|a_\uparrow(\partial_j u^r_x)\xi_\lambda\|\|b_\downarrow(\partial_j\varphi_x)\|\|b_\downarrow(\varphi_z)\|\|a_\uparrow(u^r_z)\xi_\lambda\|
  \\
  &\leq& C\rho^{\frac{1}{2} + \frac{1}{3}}\sum_{j=1}^3\int\, dxdz^\prime \, |\omega^r_\uparrow(x;z)| \|a_\uparrow(\partial_j u^r_x)\xi_\lambda\|\|a_\uparrow(u^r_z)\xi_\lambda\| \nonumber
  \\
  &\leq& C\rho^{\frac{5}{6}-\frac{\epsilon}{3}}\|\mathbb{H}_0^{\frac{1}{2}}\xi_\lambda\|\|\mathcal{N}_>^{\frac{1}{2}}\xi_\lambda\|, \nonumber
\end{eqnarray}
where the last inequality follows from \eqref{eq: est grad u kin}. 
We now work with $\mathrm{I}_c$ in \eqref{eq: Ic}, i.e., 
  \begin{equation}
    \mathrm{I}_c = - \sum_{j=1}^3\int dydz^\prime\,\omega^r_\downarrow(y;z^\prime)\langle \xi_\lambda, a^\ast_\downarrow(u^r_y)\widetilde{b}_\uparrow^\ast(\partial_j\varphi_y) b_\uparrow(\varphi_{z^\prime}) a_\downarrow(u^r_{z^\prime})\xi_\lambda\rangle,\nonumber
  \end{equation}
  where we recall that the notation $\widetilde{b}_\uparrow^\ast(\partial_j\varphi_y)$ is for the following operator (see Lemma \ref{lem: bound b phi}):
  \[
    \widetilde{b}_\uparrow^\ast((\partial_j\varphi)_y) = \int dx\, \partial_j\varphi(x-y)a_\uparrow^\ast(u^r_x)a_\uparrow^\ast(\partial_j \overline{v}^r_x).
  \]
Then, from Proposition \ref{pro: L1, L2 v}, Lemma \ref{lem: bound b phi} and \eqref{eq: est N>}, we have
  \begin{equation}
    |\mathrm{I}_c| \leq \int dydz^\prime\,|\omega^r_\downarrow(y;z^\prime)|\|\widetilde{b}_\uparrow^\ast((\partial_j\varphi)_y)\|\|b_\uparrow(\varphi_{z^\prime})\| \|\|a_\downarrow(u^r_y)\xi_\lambda\|\|a_\downarrow(u^r_{z^\prime})\xi_\lambda\|\leq C\rho^{\frac{7}{6}-\frac{\epsilon}{3}} \langle\xi_\lambda, \mathcal{N}_> \xi_\lambda\rangle.
  \end{equation}
We now estimate all the contributions coming from the second line in the right hand side of \eqref{eq: comm B A}. To start, we rewrite the commutator as 
\begin{eqnarray}\label{eq: cmm A B second line}
   &&- a_\uparrow(\overline{v}^r_z)a_\downarrow(\overline{v}^r_{z^\prime})[a_\uparrow^\ast(u^r_x) a_\downarrow^\ast(u^r_y), a_\uparrow(u^r_z)a_\downarrow(u^r_{z^\prime})]a^\ast_\downarrow(\overline{v}^r_y)a^\ast_\uparrow(\partial_j\overline{v}^r_x) 
  \\
  &&\qquad=  - a^\ast_\downarrow(\overline{v}^r_y)a^\ast_\uparrow(\partial_j\overline{v}^r_x) [a_\uparrow^\ast(u^r_x) a_\downarrow^\ast(u^r_y), a_\uparrow(u^r_z)a_\downarrow(u^r_{z^\prime})] a_\uparrow(\overline{v}^r_z)a_\downarrow(\overline{v}^r_{z^\prime}) \nonumber
  \\
  && \qquad \quad - [a_\uparrow^\ast(u^r_x) a_\downarrow^\ast(u^r_y), a_\uparrow(u^r_z)a_\downarrow(u^r_{z^\prime})][a_\uparrow(\overline{v}^r_z)a_\downarrow(\overline{v}^r_{z^\prime}) ,  a^\ast_\downarrow(\overline{v}^r_y)a^\ast_\uparrow(\partial_j\overline{v}^r_x)]\nonumber
\end{eqnarray}
Moreover, we have 
\begin{eqnarray}\label{eq: 1}
  [a_\uparrow^\ast(u^r_x) a_\downarrow^\ast(u^r_y), a_\uparrow(u^r_z)a_\downarrow(u^r_{z^\prime})] &=& (u^r)^2_\uparrow(z;x) a_\downarrow(u^r_{z^\prime})a^\ast_\downarrow(u^r_y) - (u^r)^2_\downarrow(z^\prime;y)a^\ast_\uparrow(u^r_x)a_\uparrow(u^r_z)
  \\
  &=&  (u^r)^2_\downarrow(z^\prime;y)(u^r)^2_\uparrow(z;x) - (u^r)^2_\uparrow(z;x) a^\ast_\downarrow(u^r_y)a_\downarrow(u^r_{z^\prime}) \nonumber
  \\
  &&-(u^r)^2_\downarrow(z^\prime;y)a^\ast_\uparrow(u^r_x)a_\uparrow(u^r_z), \nonumber
\end{eqnarray}
where $(u^r)^2(\cdot\,;\cdot)$ is defined as in Remark \ref{rem: u2 and tilde omega}. 
Correspondingly, we have two types of error terms to estimate (we take into account them for each $j=1,2,3$ fixed):
\begin{equation}
  \mathrm{II}_a = -\int dxdydzdz^\prime \partial_j\varphi(x-y)\varphi(z-z^\prime)(u^r)^2_\downarrow(z^\prime;y)(u^r)^2_\uparrow(z;x) \langle \xi_\lambda,  a^\ast_\downarrow(\overline{v}^r_y)a^\ast_\uparrow(\partial_j\overline{v}^r_x)a_\uparrow(\overline{v}^r_z)a_\downarrow(\overline{v}^r_{z^\prime}) \xi_\lambda\rangle,
  \end{equation}
\begin{equation}
  \mathrm{II}_b =  \int dxdydzdz^\prime \partial_j\varphi(x-y)\varphi(z-z^\prime)(u^r)^2_\downarrow(z^\prime;y) \langle \xi_\lambda,  a^\ast_\downarrow(\overline{v}^r_y)a^\ast_\uparrow(\partial_j\overline{v}^r_x)a^\ast_\uparrow(u^r_x) a_\uparrow(u^r_z) a_\uparrow(\overline{v}^r_z)a_\downarrow(\overline{v}^r_{z^\prime}) \xi_\lambda\rangle.
\end{equation}
From Proposition \ref{pro: L1, L2 v} and Lemma \ref{lem: bound phi}, we get
\begin{eqnarray}
  |\mathrm{II}_a | \hspace{-0.3cm}&\leq &\hspace{-0.3cm} \int dxdydzdz^\prime |\partial_j\varphi(x-y)||\varphi(z-z^\prime)||(u^r)^2_\downarrow(z^\prime;y)||(u^r)^2_\uparrow(z;x)| \|a^\ast_\uparrow(\partial_j\overline{v}^r_x)\|\|a_\uparrow(\overline{v}^r_z)\|\|a_\downarrow(\overline{v}^r_y)\xi_\lambda\|\|a_\downarrow(\overline{v}^r_{z^\prime}) \xi_\lambda\|\nonumber
  \\
  &\leq& C\rho^{1+\frac{1}{3}}\|\nabla\varphi\|_1 \|\varphi\|_\infty \langle \xi_\lambda, \mathcal{N}\xi_\lambda\rangle \leq C\rho\langle\xi_\lambda, \mathcal{N}\xi_\lambda\rangle,
\end{eqnarray}
where we also used that $\|a_\uparrow^\ast(\partial_j\overline{v}_x)\|_2\leq C\rho^{1/2 + 1/3}$ and that $\|(u^r)^2_{x,\sigma}\|_1\leq C$. Similarly, using Proposition \ref{pro: L1, L2 v} and \eqref{eq: N> lambda}, we get 
\begin{equation}
  |\mathrm{II}_b| \leq C\rho^{2+\frac{1}{3}} \|\nabla\varphi\|_1\|\varphi\|_1 \langle \xi_\lambda, \mathcal{N}_> \xi_\lambda\rangle  \leq C\rho^{\frac{4}{3}}\langle \xi_\lambda, \mathcal{N}_> \xi_\lambda\rangle \leq C\rho^{\frac{2}{3}}\langle \xi_\lambda, \mathbb{H}_0 \xi_\lambda\rangle,
\end{equation}
We now compute the second line in \eqref{eq: cmm A B second line}, recall \eqref{eq: 1} and
\begin{multline}\label{eq: comm 4vs}
 [a_\downarrow^\ast(\overline{v}^r_y)a_\uparrow^\ast(\partial_j\overline{v}^r_x),a_\uparrow(\overline{v}^r_z)a_\downarrow(\overline{v}^r_{z^\prime})] 
 \\
 = - \partial_j \omega^r_\uparrow(x;z)\omega^r_\downarrow (y;z^\prime) + \partial_j\omega^r_\uparrow(x;z)a^\ast_\downarrow(\overline{v}^r_y) a_\downarrow (\overline{v}^r_{z^\prime})  + \omega^r_\downarrow(y;z^\prime) a^\ast_\uparrow(\partial_j\overline{v}^r_x) a_\uparrow(\overline{v}^r_z).
\end{multline} 
We then estimate all the possible error terms coming from the product of \eqref{eq: 1} with \eqref{eq: comm 4vs} (for each $j=1,2,3$). We start by taking into account:
\begin{equation}
  \mathrm{III}_a = -\int dxdydzdz^\prime \partial_j\varphi(x-y)\varphi(z-z^\prime)(u^r)^2_\downarrow(z^\prime;y)(u^r)^2_\uparrow(z;x)\partial_j\omega^r_\uparrow(x;z) \langle \xi_\lambda,a^\ast_\downarrow(\overline{v}^r_y) a_\downarrow (\overline{v}^r_{z^\prime})\xi_\lambda\rangle.\nonumber
\end{equation} 
Form Proposition \ref{pro: L1, L2 v} and Lemma \ref{lem: bound phi}, we have (using also that $\|u^r_{x,\sigma}\|_1 \leq C$ and $\|\partial_j \omega^r\|_\infty\leq C\rho^{4/3}$):
\begin{eqnarray}
  |\mathrm{III}_a | &\leq& \int dxdydzdz^\prime |\partial_j\varphi(x-y)||\varphi(z-z^\prime)||u^r_\downarrow(z^\prime;y)||u^r_\uparrow(z;x)||\partial_j\omega^r_\uparrow(z;x)| \|a_\downarrow(\overline{v}^r_y)\xi_\lambda\|\|a_\downarrow (\overline{v}^r_{z^\prime})\xi_\lambda\|\nonumber
  \\
  &\leq& C\rho^{\frac{4}{3}}\|\varphi\|_\infty  \|\nabla\varphi\|_1 \langle \xi_\lambda, \mathcal{N}\xi_\lambda\rangle \leq C\rho\langle \xi_\lambda, \mathcal{N}\xi_\lambda\rangle.
  \end{eqnarray}
Another error term is 
\begin{equation}
 \mathrm{III}_b = -\int dxdydzdz^\prime \partial_j\varphi(x-y)\varphi(z-z^\prime)(u^r)^2_\downarrow(z^\prime;y)(u^r)^2_\uparrow(z,x)\omega^r_\downarrow(y;z^\prime) \langle \xi_\lambda, a^\ast_\uparrow(\partial_j\overline{v}^r_x) a_\uparrow(\overline{v}^r_z)\xi_\lambda\rangle.
\end{equation}
Similarly as before, using that $\|\omega^r\|_\infty\leq C\rho$, we have 
\begin{equation}
  |\mathrm{III}_b| \leq C\rho^{1+\frac{1}{3}}\|\varphi\|_\infty\|\nabla\varphi\|_1\|\langle \xi_\lambda, \mathcal{N}\xi_\lambda\rangle \leq C\rho\langle\xi_\lambda, \mathcal{N}\xi_\lambda\rangle,
\end{equation}
where we used that $\int dx\, \|a(\partial_j\overline{v}^r_x)\xi_\lambda\|^2 \leq C\rho^{1/3}\langle \xi_\lambda, \mathcal{N}\xi_\lambda\rangle$. We now estimate
\begin{equation}
 \mathrm{III}_c =  -\int dxdydzdz^\prime \partial_j\varphi(x-y)\varphi(z-z^\prime) \partial_j \omega^r_\uparrow(x;z)\omega^r_\downarrow (y;z^\prime) (u^r)^2_\uparrow(z;x) \langle\xi_\lambda, a^\ast_\downarrow(u^r_y)a_\downarrow(u^r_{z^\prime})\xi_\lambda\rangle. 
\end{equation}
We have 
\begin{equation}
  |\mathrm{III}_c| \leq C\rho^{2+ \frac{1}{3}}\|\varphi\|_1\|\nabla\varphi\|_1 \langle \xi_\lambda, \mathcal{N}_> \xi_\lambda\rangle \leq C\rho^{2+ \frac{1}{3} -\frac{1}{3} -\frac{2}{3}}\langle \xi_\lambda, \mathcal{N}_> \xi_\lambda\rangle  \leq C\rho^{\frac{2}{3}}\langle \xi_\lambda, \mathbb{H}_0 \xi_\lambda\rangle.
\end{equation}
where we used again \eqref{eq: N> lambda}, Lemma \ref{lem: bound phi} and Proposition \ref{pro: L1, L2 v}. The next type of error term is
\begin{equation}
 \mathrm{III}_d = \int dxdydzdz^\prime \partial_j\varphi(x-y)\varphi(z-z^\prime)(u^r)^2_\uparrow(z;x)\omega^r_\uparrow(y;z^\prime) \langle \xi_\lambda, a^\ast_\downarrow(u^r_y) a^\ast_\uparrow(\partial_j\overline{v}^r_x) a_\uparrow(\overline{v}^r_{z})a_\downarrow (u^r_{z^\prime}) \xi_\lambda\rangle, 
\end{equation}
can be estimated as $\mathrm{III}_c$ replacing the estimate $\|\partial_j\omega^r\|_\infty\leq C\rho^{1+1/3}$ with $\|a^\ast_\uparrow(\partial_j \overline{v}^r_x)a_\uparrow(\overline{v}^r_z)\|\leq C\rho^{1+1/3}$ to get
\begin{equation}
  |\mathrm{III}_d| \leq C\rho^{\frac{2}{3}}\langle\xi_\lambda, \mathbb{H}_0\xi_\lambda\rangle .
\end{equation}
Next we bound
\begin{equation}
\mathrm{III}_e = \int dxdydzdz^\prime \partial_j\varphi(x-y)\varphi(z-z^\prime)(u^r)^2_\uparrow(z;x)\partial_j\omega^r_\uparrow(x;z) \langle \xi_\lambda, a^\ast_\downarrow(u^r_y) a^\ast_\uparrow(\overline{v}^r_y)  a_\downarrow(\overline{v}^r_{z^\prime})a_\downarrow (u^r_{z^\prime})\xi_\lambda\rangle.
\end{equation}
From Proposition \ref{pro: L1, L2 v} and Lemma \ref{lem: bound phi}, similarly as before, we get  
\begin{equation}
  |\mathrm{III}_e| \leq C\rho^{\frac{2}{3}}\langle\xi_\lambda, \mathbb{H}_0\xi_\lambda\rangle .
\end{equation}
The last term we have to take into account is the constant one, i.e., neglecting the spin,
\begin{equation}\label{eq: term IIIf H0}
  \mathrm{III}_f = \int dxdydzdz^\prime \partial_j\varphi(x-y)\varphi(z-z^\prime)(u^r)^2_\downarrow(z^\prime;y)(u^r)^2_\uparrow(z;x)\partial_j \omega^r_\uparrow(x;z)\omega^r_\downarrow (y;z^\prime).
\end{equation} 
The idea is to estimate $\mathrm{III}_f$ replacing $(u^r)^2_\uparrow(z;x)(u^r)^2_\downarrow(z^\prime;y)$ with the Dirac delta distributions up to errors. To do that we write $(u^r)_\uparrow^2(z;x) = \delta^r_\uparrow(z;x) - \nu^r_\uparrow(z;x)$, where $\hat{\delta}^r_\uparrow(k)$ and $\hat{\nu}^r_\uparrow(k)$ are such that
\begin{equation}\label{eq: def delta nu}
  \hat{\delta}^r_\uparrow(k) = \begin{cases} 1 &\mbox{for }\,\, |k|\leq \rho^{-\beta} \\ 0 &\mbox{for }\,\, |k| \geq 2\rho^{-\beta},  \end{cases}\qquad \hat{\nu}^r_\uparrow(k) = \begin{cases} 1 &\mbox{for }\,\, |k| \leq 2k_F^\sigma \\ 0 &\mbox{for }\,\, |k| \geq 3k_F^\sigma,  \end{cases}
\end{equation}
with smooth interpolation between $1$ and $0$.
Moreover, $\nu_\uparrow(z;x) \equiv \nu_\uparrow(z-x)$ and $\delta^r_\uparrow(z;y) \equiv \delta^r_\uparrow(z-y)$
One can  prove that $\|\delta^r_{x,\sigma}\|_1, \|\nu^r_{x,\sigma}\|_1 \leq C$ (the proof can be done as the one of Proposition \ref{pro: L1, L2 v}, see \cite{FGHP} for more details). Doing the same for $(u^r)^2_\downarrow(z^\prime;y)$, we decompose $\mathrm{III}_f$ as sum of different terms. The first one is of the following type
\begin{equation}
  \mathrm{III}_{f;1} = -\int dxdydzdz^\prime \partial_j\varphi(x-y)\varphi(z-z^\prime)\delta^r_\downarrow(z^\prime;y)\nu^r_\uparrow(z;x)\partial_j \omega^r_\uparrow(x;z)\omega_\downarrow^r (y;z^\prime).
\end{equation}
Using Lemma \ref{lem: bound phi} together with the fact that $\|\delta^r_{x,\sigma}\|_1 \leq C$ and $\|\nu^r\|_\infty \leq C\rho$, we have
\begin{eqnarray}
  |\mathrm{III}_{f;1}| &\leq& C\rho^{\frac{7}{3} +1}  \int dxdydzdz^\prime |\partial_j\varphi(x-y)| |\varphi(z-z^\prime)||\delta^r_\downarrow(z^\prime;y)|\nonumber \\
  &\leq& CL^3\rho^{\frac{7}{3} + 1}\|\nabla\varphi\|_1 \|\varphi\|_1  \leq CL^3\rho^{\frac{7}{3}}\nonumber.
\end{eqnarray}
 Another error term is given by 
\begin{equation}
  \mathrm{III}_{f;2} = \int dxdydzdz^\prime \partial_j\varphi(x-y)\varphi(z-z^\prime)\nu^r_\downarrow(z^\prime;y)\nu^r_\uparrow(z;x)\partial_j \omega^r_\uparrow(x;z)\omega^r_\downarrow (y;z^\prime).
\end{equation}
Exactly as for $\mathrm{III}_{f;1}$, using that $\|\nu^r_{x,\sigma}\|_1 \leq C$, we get $|\mathrm{III}_{f;2}| \leq CL^3\rho^{\frac{7}{3}}$. The last possible term we have to take into account is
\begin{equation}
  \mathrm{III}_{f;3} = \int\, dxdydzdz^\prime \partial_j\varphi(x-y)\varphi(z-z^\prime)\delta^r_\downarrow(z^\prime;y)\delta^r_\uparrow(z;x)\partial_j \omega^r_\uparrow(z;x)\omega^r_\downarrow (z^\prime;y).
\end{equation}
We are going to prove that also $\mathrm{III}_{f;3}$ is $\mathcal{O}(L^3\rho^{7/3})$. To do that we replace each $\delta^r$ in $\mathrm{III}_{f;3}$ with the periodic Dirac delta distribution. More precisely, we write $\delta^r = \delta - \delta^>$, where $\hat{\delta}^>$ is supported for $|k| \geq \rho^{-\beta}$ and we bound all the possible terms. We start by taking into account
\begin{eqnarray}
  \mathrm{III}_{f;3;1} &=& \int dxdydzdz^\prime \partial_j\varphi(x-y)\varphi(z-z^\prime)\delta(z^\prime-y)\delta(z-x)\partial_j \omega^r_\uparrow(z;x)\omega^r_\downarrow (z^\prime;y)
  \\
  &=& \int\, dxdy\, \partial_j\varphi(x-y)\varphi(x-y)\partial_j\omega^r_\uparrow(x;x)\omega^r_\downarrow (y;y). 
\end{eqnarray}
We now notice that $\partial_j\omega^r_\uparrow(x;x)\omega^r_\downarrow(y;y)$ are fixed and independent of $x,y$. Thus, we take into account $\int dxdy\partial_j\varphi(x-y)\varphi(x-y)$. First, the integral is bounded. Moreover, from 
\begin{equation}
  \int_{\Lambda_L \times\Lambda_L}\partial_j\varphi(x-y)\varphi(x-y) = \sum_{p\in\frac{2\pi}{L}\mathbb{Z}^3}-ip_j \hat{\varphi}(p)\hat{\varphi}(-p),
\end{equation}
it follows that $\mathrm{III}_{f;3;1}$ is vanishing. The last two type of terms to bound are:
\begin{equation}
  \mathrm{III}_{f;3;2} = -\int dxdydzdz^\prime \partial_j\varphi(x-y)\varphi(z-z^\prime)\delta^>_\downarrow(z^\prime;y)\delta_\uparrow(z-x)\partial_j \omega^r_\uparrow(x;z)\omega^r_\downarrow (y;z^\prime),
\end{equation}
and
\begin{equation}
  \mathrm{III}_{f;3;3} = \int\, dxdydzdz^\prime \partial_j\varphi(x-y)\varphi(z-z^\prime)\delta^>_\downarrow(z^\prime;y)\delta^>_\uparrow(z;x)\partial_j \omega^r_\uparrow(x;z)\omega^r_\downarrow (y;z^\prime).
\end{equation}
We first bound $\mathrm{III}_{f;3;2}$, we have 
\begin{eqnarray}
  \mathrm{III}_{f;3;2} &=& -\int dxdydz^\prime \partial_j\varphi(x-y)\varphi(x-z^\prime)\delta^>_\downarrow(z^\prime;y)\partial_j \omega^r_\uparrow(x;x)\omega^r_\downarrow (y;z^\prime)
  \\
  &=& \frac{1}{L^6} \sum_{k_1,k_2,k_3} (k_3-k_1)_j\hat{\varphi}(k_3 - k_1)\hat{\varphi}(k_1 - k_3) \hat{\delta}^>(k_1)((k_2)_j \hat{\omega}^r(k_2))\hat{\omega}^r(k_3),\nonumber
\end{eqnarray}
where $(k_1 -k_3)_j$ denotes the $j-th$ component of the vector $k_1-k_3$
Note that above and in what follows we write $\sum_k$ in place of $\sum_{k\in (2\pi/L)\mathbb{Z}^3}$. We now proceed very similar to \cite[Appendix C]{FGHP}.
Thus, we only underline that using the support properties of $\hat{\delta}^>(k_1)$ and $\hat{\omega}(k_3)$, we have that $|k_1 -k_3| \geq C \rho^{-\beta}$. Then, we can use Lemma \ref{lem: decay phi} to bound 
\[
  |\hat{\varphi}(k_3 - k_1)| \leq  C_n\rho^{n\beta}.
\]
Moreover, using the decay properties of $\hat{\varphi}(k)$ (see Lemma \ref{lem: decay phi}) together with the support properties of $\hat{\omega}^r(k_2)$ and $\hat{\omega}^r(k_3)$, we can bound the series above and we get
\begin{equation*}
  |\mathrm{III}_{f;3;2}| \leq CL^3\rho^{\frac{7}{3} + n\beta}
\end{equation*}
Analogously, for $\mathrm{III}_{f;3;3}$, we can write 
\[
  \mathrm{III}_{f;3;3} = -\frac{1}{L^6}\sum_{k_1,k_2,k_3} (k_1)_j\hat{\varphi}(k_1)\varphi(-k_1) \hat{\delta}^>(k_1 + k_2)\hat{\delta}^>(k_1-k_3) ((k_2)_j \hat{\omega}(k_2))\hat{\omega}(k_3).
\]
Proceeding as for $\mathrm{III}_{f;3;3}$, one can prove that $|\mathrm{III}_{f;3;3}| \leq CL^3\rho^{\frac{7}{3} + n\beta}$. Putting all the estimates together, we get \eqref{eq: est err kin}.
\end{proof}
\subsection{Bounds for $\mathbb{Q}_1$ on interpolating states}
We now propagate the estimate for the operator $\mathbb{Q}_{1}$. Recall that 
\begin{equation*}
  \mathbb{Q}_{1} = \frac{1}{2}\sum_{\sigma,\sigma^\prime}\int\, dxdy\, V(x-y)  a^\ast_\sigma(u_x)a^\ast_{\sigma^\prime}(u_y)a_{\sigma^\prime}(u_y)a_\sigma(u_x) + \mathrm{h.c.},
\end{equation*}
\begin{equation*}
  \widetilde{\mathbb{Q}}_{1} = \frac{1}{2}\sum_{\sigma \neq \sigma^\prime}\int\, dxdy\, V(x-y)  a^\ast_\sigma(u_x)a^\ast_{\sigma^\prime}(u_y)a_{\sigma^\prime}(u_y)a_\sigma(u_x) + \mathrm{h.c.},
\end{equation*}
\begin{proposition}[Propagation estimate for $\mathbb{Q}_{1}, \widetilde{\mathbb{Q}}_1$]\label{pro: Q1}
Let $\lambda\in [0,1]$. Let $\xi_\lambda$ be a state as in \eqref{eq: def xi lambda}. Under the same assumptions of Theorem \ref{thm: optimal up bd} and Theorem \ref{thm: lower bound}, it holds  
\begin{equation}
    \frac{d}{d\lambda}\langle \xi_\lambda, \mathbb{Q}_{1}\xi_\lambda\rangle =  -\langle\xi_\lambda, \mathbb{T}_2\xi_\lambda\rangle + \mathcal{E}_{\mathbb{Q}_{1}}, \quad  \frac{d}{d\lambda}\langle \xi_\lambda, \widetilde{\mathbb{Q}}_{1}\xi_\lambda\rangle =  -\langle\xi_\lambda, \mathbb{T}_2\xi_\lambda\rangle + \mathcal{E}_{\widetilde{\mathbb{Q}}_{1}},
\end{equation}  
where 
\begin{equation}\label{eq: def T2}
  \mathbb{T}_2 := -\int\, dxdy\, V(x-y)\varphi(x-y)a_\uparrow(u_x)a_\uparrow(\overline{v}^r_x)a_\downarrow(u_y)a_\downarrow(\overline{v}^r_y) + \mathrm{h.c.},
\end{equation}
and
\begin{equation} 
  |\mathcal{E}_{\mathbb{Q}_{1}}| \leq  CL^{\frac{3}{2}}\rho^{\frac{4}{3}}\|\mathbb{Q}_1^{\frac{1}{2}}\xi_\lambda\| +  C\rho^{\frac{1}{2}}\|\mathbb{Q}^{\frac{1}{2}}_{1}\xi_\lambda\|\|\mathbb{H}_0^{\frac{1}{2}}\xi_\lambda\| + CL^{\frac{3}{2}}\rho^{\frac{4}{3}}\|\widetilde{\mathbb{Q}}^{\frac{1}{2}}_{1}\xi_\lambda\|,
\end{equation}
\begin{equation} 
  |\mathcal{E}_{\widetilde{\mathbb{Q}}_{1}}| \leq  C\rho^{\frac{1}{2}}\|\widetilde{\mathbb{Q}}^{\frac{1}{2}}_{1}\xi_\lambda\|\|\mathbb{H}_0^{\frac{1}{2}}\xi_\lambda\| + CL^{\frac{3}{2}}\rho^{\frac{4}{3}}\|\widetilde{\mathbb{Q}}^{\frac{1}{2}}_{1}\xi_\lambda\|.
\end{equation}
\end{proposition}
\begin{proof} We only do the proof for $\mathbb{Q}_1$, the one for $\widetilde{\mathbb{Q}}_1$ can be done in the same way.
  As in the previous proposition, we have
  \begin{equation} \label{eq: der lambda Q1}
    \frac{d}{d\lambda}\langle \xi_\lambda, \mathbb{Q}_{1} \xi_\lambda\rangle = -\langle \xi_\lambda, [\mathbb{Q}_{1}, B]\xi_\lambda\rangle + \mathrm{c.c.}
  \end{equation}
  Moreover,
  \begin{multline*} 
    \hspace{-0.3cm}[\mathbb{Q}_{1},B] = \frac{1}{2}\sum_{\sigma,\sigma^\prime} \int dxdydzdz^\prime\, V(x-y)\varphi(z-z^\prime)
    \\
    \cdot [a^\ast_\sigma(u_x)a^\ast_{\sigma^\prime}(u_y)a_{\sigma^\prime}a(u_y)a_{\sigma}(u_x), a_\uparrow(u^r_z)a_\uparrow(\overline{v}^r_z)a_\downarrow(u^r_{z^\prime})a_\downarrow(\overline{v}^r_{z^\prime})].
    \end{multline*}
    It holds
    \begin{multline*}
      [a^\ast_\sigma(u_x)a^\ast_{\sigma^\prime}(u_y)a_{\sigma^\prime}(u_y)a_{\sigma}(u_x), a_\uparrow(u^r_z)a_\uparrow(\overline{v}^r_z)a_\downarrow(u^r_{z^\prime})a_\downarrow(\overline{v}^r_{z^\prime})]
      \\
      = [a^\ast_\sigma(u_x)a^\ast_{\sigma^\prime}(u_y), a_\uparrow(u^r_z)a_\downarrow(u^r_{z^\prime})]a_\uparrow(\overline{v}^r_z)a_\downarrow(\overline{v}^r_{z^\prime})a_{\sigma^\prime}(u_y)a_\sigma(u_x).
    \end{multline*}
  We then compute 
  \begin{eqnarray} 
    [a^\ast_\sigma(u_x)a^\ast_{\sigma^\prime}(u_y), a_\uparrow(u^r_z)a_\downarrow(u^r_{z^\prime})] &=& a^\ast_\sigma(u_x) \left(\delta_{\sigma^\prime,\uparrow} u^r_{\sigma^\prime}(z;y)a_\downarrow(u^r_{z^\prime}) - \delta_{\sigma^\prime,\downarrow}u^r_{\sigma^\prime}(z^\prime;y)a_\uparrow(u^r_z) \right)
    \\
    && + \left(\delta_{\sigma,\uparrow}u^r_\sigma(z;x)a_\downarrow(u^r_{z^\prime}) - \delta_{\sigma,\downarrow}u^r_{\sigma}(z^\prime;x)a_\uparrow(u^r_z)\right) a^\ast_{\sigma^\prime}(u_y)\nonumber 
    \\
    &=& \delta_{\sigma^\prime,\uparrow} u^r_{\sigma^\prime}(z;y)a^\ast_\sigma(u_x)a_\downarrow(u^r_{z^\prime}) - \delta_{\sigma^\prime,\downarrow}u^r_{\sigma^\prime}(z^\prime;y)a^\ast_\sigma(u_x)a_\uparrow(u^r_z)\nonumber 
    \\
    && -\delta_{\sigma,\uparrow}u^r_\sigma(z;x) a^\ast_{\sigma^\prime}(u_y)a_\downarrow(u^r_{z^\prime}) + \delta_{\sigma,\downarrow}u^r_{\sigma}(z^\prime;x) a^\ast_{\sigma^\prime}(u_y)a_\uparrow(u^r_z)\nonumber 
    \\
    &&+ \delta_{\sigma,\uparrow}\delta_{\sigma^\prime,\downarrow} u^r_\sigma(z;x)u^r_{\sigma^\prime}(z^\prime;y) - \delta_{\sigma,\downarrow}\delta_{\sigma^\prime,\uparrow} u^r_\sigma(z^\prime;x)u^r_{\sigma^\prime}(z;y)\nonumber.
  \end{eqnarray}
The terms in the first two lines in the right side  above can be estimated in the same way. We take into account only the first one, i.e.,
\begin{multline} 
  \mathrm{I} = \frac{1}{2}\sum_\sigma\int dxdydzdz^\prime\, V(x-y)\varphi(z-z^\prime) u^r_\uparrow(z;y) \langle \xi_\lambda, a^\ast_\sigma(u_x)a_\downarrow(u^r_{z^\prime}) a_\uparrow(\overline{v}^r_z)a_\downarrow(\overline{v}^r_{z^\prime})a_{\uparrow}(u_y)a_\sigma(u_x)\xi_\lambda\rangle
  \\
  = - \sum_\sigma\int dxdydz\, V(x-y)u^r_\uparrow(z;y) \langle\xi_\lambda, a_\sigma^\ast(u_x)b_\downarrow(\varphi_z)a_\uparrow(\overline{v}^r_z) a_\uparrow(u_y)a_\sigma(u_x)\xi_\lambda,\rangle
\end{multline}
where to get the last equality we used the definition of the operator $b_\downarrow(\varphi_z)$ as in \eqref{eq: def b phi}.
We now write $a^\ast_\sigma(u_x) = a^\ast_\sigma(u^<_x) + a^\ast_\sigma(u^>_x)$, where $\hat{u}^<_\sigma(k)$ is supported for $k\notin\mathcal{B}_F^\sigma$ and $|k| < 2k_F^\sigma$ and $\hat{u}^>_\sigma(k)$ is supported for $ |k| \geq 2k_F^\sigma$. Correspondingly we rewrite the term $\mathrm{I}$ as $\mathrm{I} = \mathrm{I}_< + \mathrm{I}_>$. We start with estimating $\mathrm{I}_<$. From Proposition \ref{pro: L1, L2 v} we can bound $\|u^r_{x,\sigma}\|\leq C$ and from  Lemma \ref{lem: bound b phi}, we get 
\begin{eqnarray}\label{eq: est I< Q1}
  |\mathrm{I}_<| &\leq& C\int dxdydz\, V(x-y) |u^r_\uparrow(z;y)|\|a_\sigma^\ast(u^<_x)\|\|b_\downarrow(\varphi_z)\| \|a_\uparrow(\overline{v}^r_z)\|\|a_\uparrow(u_y)a_\sigma(u_x)\xi_\lambda\|\nonumber
  \\
  &\leq& CL^{\frac{3}{2}} \rho^{1+\frac{1}{3}} \|\mathbb{Q}_{1}^{\frac{1}{2}}\xi_\lambda\|,
\end{eqnarray}
where we also used that $\|a_\downarrow(u^<_x)\|\leq C\rho^{1/2}$ and  $\|a_\uparrow(\overline{v}^r_z)\| \leq C\rho^{1/2}$.
We now take into account the operator $I_>$. Using again Proposition \ref{pro: L1, L2 v} and  Lemma \ref{lem: bound b phi}, we get
\begin{eqnarray} \label{eq: est I> Q1}
  |\mathrm{I}_>| &\leq& C\rho^{\frac{1}{3}}\int dxdydz\, V(x-y)|u^r_\uparrow(z;y)| \|a_\uparrow(\overline{v}^r_z)\|\|a_\sigma(u_x^>)\xi_\lambda\| \|a_\uparrow(u_y)a_\sigma(u_x)\xi_\lambda\|\nonumber
  \\
  &\leq& C\rho^{\frac{1}{2} +\frac{1}{3}}\int dxdydz\, V(x-y)|u^r_\uparrow(z;y)| \|a_\sigma(u_x^>)\xi_\lambda\| \|a_\uparrow(u_y)a_\sigma(u_x)\xi_\lambda\|\nonumber 
  \\
  &\leq& C\rho^{\frac{1}{2} +\frac{1}{3}}\|\mathbb{Q}^{\frac{1}{2}}_{1}\xi_\lambda\|\|\mathcal{N}^{\frac{1}{2}}_>\xi_\lambda\| \leq C\rho^{\frac{1}{2}}\|\mathbb{Q}^{\frac{1}{2}}_{1}\xi_\lambda\|\|\mathbb{H}_0^{\frac{1}{2}}\xi_\lambda\|, 
\end{eqnarray}
where we used also that $\|\mathcal{N}_>^{\frac{1}{2}}\xi_\lambda\| \leq C\rho^{-1/3}\|\mathbb{H}_0^{1/2}\xi_\lambda\|$.
The leading term in \eqref{eq: der lambda Q1} is equal to 
\begin{equation} 
  \mathrm{I}_{\mathrm{main}} = -\int dxdydzdz^\prime\, V(x-y)\varphi(z-z^\prime)u^r_\uparrow(z;x)u^r_\downarrow(z^\prime;y) a_\uparrow(u_x) a_\uparrow(\overline{v}^r_z)a_\downarrow(u_y) a_\downarrow(\overline{v}^r_{z^\prime}).
\end{equation}
As in \cite[Proposition 5.5]{FGHP}, we now replace the product $u^r_\uparrow(z;x)u^r_\downarrow(z^\prime;y)$ in $\mathrm{I}_{\mathrm{main}}$ with $\delta_\uparrow^r(z;x)\delta_\downarrow^r(z^\prime;y)$. To do that we rewrite
\begin{equation} 
  u^r_\uparrow(z;x) = \delta^r_\uparrow(z;x) -\nu^r_\uparrow(z;x), \qquad u^r_\downarrow(z^\prime;y) = \delta^r_\downarrow(z^\prime;y) -\nu_\downarrow(z^\prime;y),
\end{equation}
where $\delta^r$ and $\nu^r$ are as in \eqref{eq: def delta nu}. We then have
\begin{equation} 
  \mathrm{I}_{\mathrm{main}} = -\int dxdydzdz^\prime V(x-y)\varphi(z-z^\prime)\delta^r_\uparrow(z;x)\delta^r_\downarrow(z^\prime;y) a_\uparrow(u_x)a_\uparrow(\overline{v}^r_z)a_\downarrow(u_y)a_\downarrow(\overline{v}^r_{z^\prime}) + \mathcal{E}_{\mathrm{I}_{\mathrm{main}}},
\end{equation}
where in $\mathcal{E}_{\mathrm{I}_{\mathrm{main}}}$ there are two types of terms.  To estimate them, one can proceed as in \cite[Proposition 5.5]{FGHP}, we omit the details. Using the bound for $\|\varphi\|_1$ proved in Lemma \ref{lem: bound phi}, one gets
 \begin{equation}
  |\mathcal{E}_{\mathrm{main}}|\leq CL^{\frac{3}{2}}\rho^{\frac{4}{3}}\|\widetilde{\mathbb{Q}}_{1}^{\frac{1}{2}}\xi_\lambda\|.
 \end{equation}
Moreover, proceeding similarly as in \cite[Lemma 5.6]{FGHP} (using the decay estimates proved in Lemma \ref{lem: decay phi}) and the support properties of $\varphi(z-z^\prime)$ together with the assumptions on the potential (see Assumption \ref{asu: potential V}), we can replace the regularized $\delta^r$ with the periodic Dirac delta distribution up to a small error, i.e., 
\begin{multline}
 -\int dxdydzdz^\prime V(x-y)\varphi(z-z^\prime)\delta^r_\uparrow(z;x)\delta^r_\downarrow(z^\prime;y) a_\uparrow(\overline{v}^r_z)a_\uparrow(u_x)a_\downarrow(\overline{v}^r_{z^\prime})a_\downarrow(u_y)
 \\
 = -\int dxdy\, V(x-y)\varphi(x-y)a_\uparrow(\overline{v}^r_x)a_\uparrow(u_x)a_\downarrow(\overline{v}^r_{y})a_\downarrow(u_y) + \hat{\mathcal{E}}_{\mathbb{Q}_1},
 \end{multline}
 where, for all $n\geq 4$, taking $V\in C^k$ with $k$ large enough,
 \begin{equation}\label{eq: est I main Q1}
  | \hat{\mathcal{E}}_{\mathbb{Q}_{1}}| \leq C_n\rho^{n\beta(n-3)}\left(CL^3\rho^2 + \langle \xi_\lambda, \widetilde{\mathbb{Q}}_{1}\xi_\lambda\rangle\right).
 \end{equation}
 Combining the estimates in \eqref{eq: est I< Q1}, \eqref{eq: est I> Q1} and \eqref{eq: est I main Q1} and taking $n$ large enough, we are done.
\end{proof}

\subsection{Bounds for $\widetilde{\mathbb{Q}}_4$, $\widetilde{\mathbb{Q}}_4^r$ on interpolating states}\label{sec: bounds Q4}
We now prove some propagation estimates for the following operator
\begin{equation}\label{eq: def tilde Q4}
	\widetilde{\mathbb{Q}}_4 = \frac{1}{2}\sum_{\sigma\neq\sigma^\prime} \int dxdy\, V(x-y) a^\ast_\sigma(u_x)a^\ast_\sigma(u_y)a^\ast_{\sigma^\prime}(\overline{v}_y)a^\ast_{\sigma^\prime}(\overline{v}_x) + \mathrm{h.c.}
\end{equation}
We also take into account the regularized version of it, i.e., we define
\begin{equation}
	\widetilde{\mathbb{Q}}^r_4 = \frac{1}{2}\sum_{\sigma\neq\sigma^\prime} \int dxdy\, V(x-y) a^\ast_\sigma(u^r_x)a^\ast_\sigma(u_y^r)a^\ast_{\sigma^\prime}(\overline{v}_y^r)a^\ast_{\sigma^\prime}(\overline{v}_x^r) + \mathrm{h.c.}
\end{equation}
\begin{proposition}[Propagation estimate for $\mathbb{Q}_4, \widetilde{\mathbb{Q}}^r_4$]\label{pro: Q4} Let $\lambda\in [0,1]$ and $0< \epsilon \leq 1$. Let $\xi_\lambda:= T^\ast_\lambda R^\ast \psi$ be a state as in \eqref{eq: def xi lambda}. Under the same assumptions of Theorem \ref{thm: optimal up bd} and Theorem \ref{thm: lower bound}, it holds  
 \begin{equation} 
      \frac{d}{d\lambda}\langle \xi_\lambda, \widetilde{\mathbb{Q}}_4^r\xi_\lambda\rangle = 2\rho_\uparrow \rho_\downarrow\int dxdy V(x-y)\varphi(x-y) + \mathcal{E}_{\widetilde{\mathbb{Q}}^r_4}(\xi_\lambda), 
    \end{equation}
    \begin{equation} 
      \frac{d}{d\lambda}\langle \xi_\lambda, \widetilde{\mathbb{Q}}_4\xi_\lambda\rangle = 2\rho_\uparrow \rho_\downarrow\int dxdy V(x-y)\varphi(x-y) + \mathcal{E}_{\widetilde{\mathbb{Q}}_4}(\xi_\lambda), 
    \end{equation}
    where,  
    \begin{equation} 
  |\mathcal{E}_{\widetilde{\mathbb{Q}}_4}(\xi_\lambda)| \leq CL^3\rho^{\frac{7}{3}} + CL^3\rho^{2+\frac{\epsilon}{3}} +  C\rho^{\frac{5}{6} -\frac{\epsilon}{3}}\|\widetilde{\mathbb{Q}}_1^{\frac{1}{2}}\xi_\lambda\|\|\mathcal{N}_>^{\frac{1}{2}}\xi_\lambda\| +C\rho\langle \xi_\lambda, \mathcal{N}\xi_\lambda\rangle ,
    \end{equation}
    \begin{equation}\label{eq: err Q tilde 4 r}
      |\mathcal{E}_{\widetilde{\mathbb{Q}}^r_4}(\xi_\lambda)| \leq CL^3\rho^{\frac{7}{3}} + CL^3\rho^{2+\frac{\epsilon}{3}} +  C\rho^{\frac{5}{6} -\frac{\epsilon}{3}}\|\widetilde{\mathbb{Q}}_1^{\frac{1}{2}}\xi_\lambda\|\|\mathcal{N}_>^{\frac{1}{2}}\xi_\lambda\| +C\rho\langle \xi_\lambda, \mathcal{N}\xi_\lambda\rangle.
    \end{equation}
\end{proposition}

\begin{proof} 
  We only discuss the proof for $\widetilde{\mathbb{Q}}_4^r$, the one for the operator $\widetilde{\mathbb{Q}}_4$ is simpler. To start, we have
  \begin{equation} 
    \frac{d}{d\lambda}\langle \xi_\lambda, \widetilde{\mathbb{Q}}^r_4\xi_\lambda\rangle = -\langle \xi_\lambda, [\widetilde{\mathbb{Q}}^r_4, B]\xi_\lambda\rangle + \mathrm{h.c.}
  \end{equation}
  Moreover, 
  \begin{multline} 
    [\widetilde{\mathbb{Q}}^r_4,B]
    \\ = \frac{1}{2}\sum_{\sigma\neq \sigma^\prime} \int dxdydzdz^\prime\, V(x-y)\varphi(z-z^\prime)[a^\ast_{\sigma}(u_x^r)a^\ast_{\sigma^\prime}(u_y^r)a^\ast_{\sigma^\prime}(\overline{v}_y^r)a^\ast_\sigma(\overline{v}_x^r), a_\uparrow(u^r_z)a_\uparrow(\overline{v}^r_z)a_\downarrow(u^r_{z^\prime})a_\downarrow(\overline{v}^r_{z^\prime})].
  \end{multline}
  We can rewrite the commutator as 
  \begin{eqnarray} \label{eq: comm Q4}
    &&[a^\ast_{\sigma}(u_x^r)a^\ast_{\sigma^\prime}(u_y^r)a^\ast_{\sigma^\prime}(\overline{v}_y^r)a^\ast_\sigma(\overline{v}_x^r), a_\uparrow(u^r_z)a_\uparrow(\overline{v}^r_z)a_\downarrow(u^r_{z^\prime})a_\downarrow(\overline{v}^r_{z^\prime})]
    \\
    && =-a^\ast_{\sigma}(u_x^r)a^\ast_{\sigma^\prime}(u_y^r)[a^\ast_{\sigma^\prime}(\overline{v}_y^r)a^\ast_\sigma(\overline{v}_x^r), a_\uparrow(\overline{v}^r_z)a_\downarrow(\overline{v}^r_{z^\prime})]a_\uparrow(u^r_z)a_\downarrow(u^r_{z^\prime})\nonumber
    \\
    &&\quad - a_\uparrow(\overline{v}^r_z)a_\downarrow(\overline{v}^r_{z^\prime}) [a^\ast_{\sigma}(u_x^r)a^\ast_{\sigma^\prime}(u_y^r), a_\uparrow(u^r_z)a_\downarrow(u^r_{z^\prime})]a^\ast_{\sigma^\prime}(\overline{v}_y^r)a^\ast_\sigma(\overline{v}_x^r)\nonumber
    \\
    &&= -a^\ast_{\sigma}(u_x^r)a^\ast_{\sigma^\prime}(u_y^r)[a^\ast_{\sigma^\prime}(\overline{v}_y^r)a^\ast_\sigma(\overline{v}_x^r), a_\uparrow(\overline{v}^r_z)a_\downarrow(\overline{v}^r_{z^\prime})]a_\uparrow(u^r_z)a_\downarrow(u^r_{z^\prime})\nonumber
    \\
    &&\quad - a^\ast_{\sigma^\prime}(\overline{v}_y^r)a^\ast_\sigma(\overline{v}_x^r)[a^\ast_{\sigma}(u_x^r)a^\ast_{\sigma^\prime}(u_y^r), a_\uparrow(u^r_z)a_\downarrow(u^r_{z^\prime})]a_\uparrow(\overline{v}^r_z)a_\downarrow(\overline{v}^r_{z^\prime})\nonumber
    \\
    &&\quad -  [a^\ast_{\sigma}(u_x^r)a^\ast_{\sigma^\prime}(u_y^r), a_\uparrow(u^r_z)a_\downarrow(u^r_{z^\prime})][a_\uparrow(\overline{v}^r_z)a_\downarrow(\overline{v}^r_{z^\prime}),a^\ast_{\sigma^\prime}(\overline{v}_y^r)a^\ast_\sigma(\overline{v}_x^r)]\nonumber.
  \end{eqnarray}
The terms which are going to give the main contributions are those which are completely anti-normal ordered. We take into account the first line in the right hand side above and we compute
\begin{eqnarray}\label{eq: comm 4v Q4}
  [a^\ast_{\sigma^\prime}(\overline{v}^r_y)a^\ast_\sigma(\overline{v}^r_x), a_\uparrow(\overline{v}^r_z)a_\downarrow(\overline{v}^r_{z^\prime})] &=& \delta_{\sigma,\uparrow}\omega^r_\sigma(x;z)a^\ast_{\sigma^\prime}(\overline{v}_y)a_\downarrow(\overline{v}^r_{z^\prime}) - \delta_{\sigma, \downarrow} \omega^r_\sigma(x;z^\prime)a^\ast_{\sigma^\prime}(\overline{v}_y)a_\uparrow(\overline{v}^r_z)\nonumber 
  \\
  && + \delta_{\sigma^\prime, \uparrow} \omega^r_{\sigma^\prime}(y;z) a_\downarrow(\overline{v}^r_{z^\prime})a^\ast_{\sigma}(\overline{v}_x) - \delta_{\sigma^\prime, \downarrow} \omega^r_{\sigma^\prime}(y;z^\prime)a_\uparrow(\overline{v}^r_{z})a^\ast_\sigma(\overline{v}_x).\nonumber
\end{eqnarray}
All of the corresponding terms can be estimated in the same way, it is not necessary to put them in normal order. We bound for instance 
\begin{equation} \label{eq: term I Q4}
  \mathrm{I}:=  \int dxdydzdz^\prime\, V(x-y)\varphi(z-z^\prime)\omega^r_\uparrow(x;z)\langle \xi_\lambda, a^\ast_\uparrow(u^r_x)a^\ast_{\downarrow}(u^r_y) a^\ast_{\downarrow}(\overline{v}^r_y)a_\downarrow (\overline{v}^r_{z^\prime})a_\downarrow(u^r_{z^\prime})a_\uparrow(u^r_{z})\xi_\lambda\rangle.
\end{equation}
Following the same ideas as in \cite[Proposition 6.1]{FGHP}, we first replace $u^r_x$, $u^r_y$ with $u_x$, $u_y$. 
We can then rewrite $\mathrm{I}$ as: 
\begin{equation}
  \mathrm{I} = \int dxdydz\, V(x-y)\omega^r_\uparrow(x;z) \langle \xi_\lambda, a_\uparrow^\ast(u_x)a_\downarrow^\ast(u_y) a^\ast_\downarrow(\overline{v}^r_y)b_\downarrow(\varphi_z)a_\uparrow(u^r_{z})\xi_\lambda\rangle + \widetilde{\mathrm{I}} \equiv \mathrm{I}_1 + \widetilde{\mathrm{I}},
\end{equation}
where $\widetilde{\mathrm{I}}$ is an error term. We start by estimating $\mathrm{I}_1$. 
From Proposition \ref{pro: L1, L2 v} (to bound $\|\overline{v}^r_{y,\sigma}\|\leq C\rho^{1/2}$ and $\|\omega^r_{x,\sigma}\|_1\leq C\rho^{-\epsilon/3}$) and Lemma \ref{lem: bound b phi} (to control $\|b_\sigma(\varphi_z)\|\leq C\rho^{1/3}$), we get
\begin{eqnarray}
  |\mathrm{I}_1| &\leq&  \int dxdydz\, |V(x-y)||\omega^r_\uparrow(x;z) | \| a^\ast_\downarrow(\overline{v}^r_y)\|\|b_\downarrow(\varphi_z)\|a_\uparrow(u_x)a_\downarrow(u_y)\xi_\lambda\|\|a_\uparrow(u^r_{z})\xi_\lambda\|\nonumber 
  \\
  &\leq& C\rho^{\frac{1}{2} +\frac{1}{3}-\frac{\epsilon}{3}}\|\widetilde{\mathbb{Q}}_{1}^{\frac{1}{2}}\xi_\lambda\|\|\mathcal{N}^{\frac{1}{2}}_>\xi_\lambda\|.
\end{eqnarray}
We now take into account the error term $\widetilde{\mathrm{I}}$. As we already did in Lemma \ref{lem: error kinetic energy}, we write $\hat{u}_\sigma(k) = \hat{u}^r_\sigma(k) + \hat{\alpha}_\sigma(k) + \hat{\delta}^<_\sigma(k)$, where $\hat{\alpha}_\sigma(k)$ is supported for $k\notin\mathcal{B}_F^\sigma$ and $|k| \leq 3k_F^\sigma$ and $\hat{\delta}^>_\sigma(k)$ is supported for $|k| \geq \rho^{-\beta}$. The estimate of this error term is similar as the corresponding one in \cite[Appendix C.3]{FGHP}. More precisely, as discussed in \cite[Appendix C.3]{FGHP}, when at least one between $a^\ast(u_x^r)$, $a^\ast(u_y^r)$ is replaced by either $a^\ast (\delta^>_x)$ or $a^\ast (\delta^>_y)$, the corresponding contribution can be proven to be smaller than any power of $\rho^\beta$ (see also the proof of \cite[Proposition 6.1]{FGHP} and the discussion for the estimate of the term $\mathrm{III}_f$ in Lemma \ref{lem: error kinetic energy}). We then discuss the case in which in $\widetilde{\mathrm{I}}$ at least one between $a^\ast(u_x^r)$, $a^\ast(u_y^r)$ is replaced by either $a^\ast(\alpha_x)$ or $a^\ast(\alpha_y)$. We consider for instance: 
\begin{equation}
  \widetilde{\mathrm{I}}_a := \int dxdydz\, V(x-y)\omega^r_\uparrow(x;z) \langle \xi_\lambda, a_\uparrow^\ast(\alpha_x)a_\downarrow^\ast(u_y) a^\ast_\downarrow(\overline{v}^r_y)b_\downarrow(\varphi_z)a_\uparrow(u^r_{z})\xi_\lambda\rangle.
\end{equation}
Then, we can bound
\begin{eqnarray}
  |\widetilde{\mathrm{I}}_a| &\leq& \int dxdydz\, V(x-y)|\omega^r_\uparrow(x;z)| \|a_\uparrow^\ast(\alpha_x)\| \|a^\ast_\downarrow(\overline{v}^r_y)\|\|b_\downarrow(\varphi_z)\|\|a_\downarrow^\ast(u_y)\xi_\lambda\|\|a_\uparrow(u^r_{z})\xi_\lambda\|
  \\
&\leq& C\rho^{1+\frac{1}{3} -\frac{\epsilon}{3}}\langle \xi_\lambda, \mathcal{N}\xi_\lambda\rangle.
\end{eqnarray}
where we used again Proposition \ref{pro: L1, L2 v}, Lemma \ref{lem: bound b phi} together with $\|a_\uparrow^\ast(\alpha_x)\|\leq \|\alpha_x\|_2 \leq C\rho^{1/2}$. The estimate for the remaining error terms in $\widetilde{\mathrm{I}}$ is the same, we omit the details. Note that the estimate of the term $\widetilde{\mathrm{I}}$ is necessary only when we propagate $\widetilde{\mathbb{Q}}_4$.

We now look to the second line and to the third line in the ride side of \eqref{eq: comm Q4}, from there we extract the leading term plus errors. To bound all the other error terms one can calculate 
\begin{eqnarray}
  [a_\uparrow^\ast(u^r_x) a_\downarrow^\ast(u^r_y), a_\uparrow(u^r_z)a_\downarrow(u^r_{z^\prime})] &=& (u^r)^2_\uparrow(z;x) a_\downarrow(u^r_{z^\prime})a^\ast_\downarrow(u^r_y) - (u^r)^2_\downarrow(z^\prime;y)a^\ast_\uparrow(u^r_x)a_\uparrow(u^r_z)
  \\
  &=&  (u^r)^2_\downarrow(z^\prime;y)(u^r)^2_\uparrow(z;x) - (u^r)^2_\uparrow(z;x) a^\ast_\downarrow(u^r_y)a_\downarrow(u^r_{z^\prime}) \nonumber
  \\
  &&-(u^r)^2_\downarrow(z^\prime;y)a^\ast_\uparrow(u^r_x)a_\uparrow(u^r_z), \nonumber
\end{eqnarray}
and use \eqref{eq: comm 4v Q4}. In particular, putting all of these error terms in normal order, we get a constant term $\mathrm{I}_{\mathrm{main}}$ plus an error $\mathcal{E}$. The error can be bounded by proceeding as in \cite[Proposition 6.1]{FGHP} and using the estimate proved in Proposition \ref{pro: L1, L2 v} and in Lemma \ref{lem: bound phi}. We have
\[
  |\mathcal{E}| \leq  C\rho\langle \xi_\lambda, \mathcal{N}\xi_\lambda\rangle.
\]
In particular, the main contribution coming from the last line of \eqref{eq: comm Q4} is the constant one, given by 
\begin{equation} \label{eq: I main Q4}
  \mathrm{I}_{\mathrm{main}}  = 2 \int dxdydzdz^\prime V(x-y) \varphi(z-z^\prime) (u^r)^2_\uparrow(z;x)(u^r)^2_\downarrow(z^\prime;y)\omega_\uparrow^r(x;z)\omega^r_\downarrow(y;z^\prime),
\end{equation}
where we also used that $\varphi(x) = \varphi(-x)$. 
All together, we find
\begin{equation} 
  \frac{d}{d\lambda}\langle \xi_\lambda, \widetilde{\mathbb{Q}}^r_4\xi_\lambda\rangle = 2 \int dxdydzdz^\prime V(x-y) \varphi(z-z^\prime) (u^r)^2_\uparrow(z;x)(u^r)^2_\downarrow(z^\prime;y)\omega_\uparrow^r(x;z)\omega^r_\downarrow(y;z^\prime)  + \mathrm{Err}_{\widetilde{\mathbb{Q}}^r_4}(\xi_\lambda), 
\end{equation}
with, for some $0<\epsilon\leq 1$,
\begin{equation} 
  |\mathrm{Err}_{\widetilde{\mathbb{Q}}^r_4}(\xi_\lambda)| \leq  C\rho^{\frac{5}{6} -\frac{\epsilon}{3}}\|\widetilde{\mathbb{Q}}_1^{\frac{1}{2}}\xi_\lambda\|\|\mathcal{N}_>^{\frac{1}{2}}\xi_\lambda\| +C\rho\langle \xi_\lambda, \mathcal{N}\xi_\lambda\rangle
\end{equation}
\\ 
\underline{Leading constant contribution}
We now analyze the constant term in \eqref{eq: I main Q4}. The idea is first to replace each $u^r$ with the periodic Dirac delta distribution. To do so, one can write, similarly as in Lemma \ref{lem: error kinetic energy} (see the estimate for the term in \eqref{eq: term IIIf H0}), $u^r = \delta^r - \nu^r$ (see \eqref{eq: def delta nu}). One then gets 
\begin{multline}\label{eq: uu - dd + er}
  2 \int dxdydzdz^\prime V(x-y) \varphi(z-z^\prime) (u^r)^2_\uparrow(z;x)(u^r)^2_\downarrow(z^\prime;y)\omega_\uparrow^r(x;z)\omega^r_\downarrow(y;z^\prime) 
  \\
  = 2 \int dxdydzdz^\prime V(x-y) \varphi(z-z^\prime) \delta^r_\uparrow(z;x)\delta^r_\downarrow(z^\prime;y)\omega_\uparrow^r(x;z)\omega^r_\downarrow(y;z^\prime) + \mathcal{E}.
\end{multline}
Using the fact that $\|\delta^r\|_1\leq C$, the support properties of $\hat{\omega}^r$ and $\hat{\nu}^r$ together with Lemma \ref{lem: bound phi}, one can prove that $\mathcal{E} = \mathcal{O}(L^3\rho^{7/3})$. We omit the details since the calculations can be performed as in \cite[Appendix C.4]{FGHP}(see also the estimate of the constant term  in Lemma \ref{lem: error kinetic energy}).
The idea then is then to replace $\delta^r(x;z)\delta^r(y;z^\prime)$ with $\delta(x-z)\delta(y-z^\prime)$. To do that one has to use the regularity of the interaction potential (see Assumption \ref{asu: potential V}) together with the properties of $\varphi$ stated in Lemma \ref{lem: bound phi} and Lemma \ref{lem: decay phi} and the bounds proved in Proposition \ref{pro: L1, L2 v}. More precisely, these terms are $o(\rho^{7/3})$ thanks to the regularity of the interaction potential and the decay properties of $\varphi$ (see Lemma \ref{lem: decay phi}), we omit the details for this replacement.  Thus, proceeding as in \cite[Appendix C.4]{FGHP}, one can prove that  
\begin{multline}
2 \int\, dxdydzdz^\prime V(x-y) \varphi(z-z^\prime) (u^r)^2_\uparrow(z;x)(u^r)^2_\downarrow(z^\prime;y)\omega_\uparrow^r(x;z)\omega^r_\downarrow(y;z^\prime
)
\\ = 2\rho_\uparrow^r\rho_\downarrow ^r\int\, dxdy\, V(x-y)\varphi(x-y) + CL^3\rho^{\frac{7}{3}},
\end{multline}
where $\rho_\uparrow^r = \omega_\uparrow^r (x;x)$ and $\rho_\downarrow^r = \omega_\downarrow^r (y;y)$.
The second step is to remove the regularization from the densities. As a consequence of the regularization we put on $\overline{v}^r$, we have
\begin{equation}
  |\rho_\sigma - \rho^r_\sigma| \leq C\rho^{1+\frac{\varepsilon}{3}}.
\end{equation}
It then follows that 
\begin{equation}
  2\rho_\uparrow^r\rho_\downarrow ^r\int dxdy\, V(x-y)\varphi(x-y) = 2\rho_\uparrow\rho_\downarrow \int dxdy\, V(x-y)\varphi(x-y) + CL^3\rho^{2+\frac{\epsilon}{3}},
\end{equation}
In conclusion, we proved that
 \begin{equation} 
  \frac{d}{d\lambda}\langle \xi_\lambda, \widetilde{\mathbb{Q}}^r_4\xi_\lambda\rangle = 2\rho_\uparrow\rho_\downarrow \int dxdy\, V(x-y)\varphi(x-y) + \mathcal{E}_{\widetilde{\mathbb{Q}}^r_4}(\xi_\lambda), 
\end{equation}
with, for some $0<\delta<1$,
\begin{equation} 
  |\mathcal{E}_{\widetilde{\mathbb{Q}}^r_4}(\xi_\lambda)| \leq CL^3\rho^{\frac{7}{3}}   +  CL^3\rho^{2+\frac{\epsilon}{3}} +C\rho^{\frac{5}{6} -\frac{\epsilon}{3}}\|\widetilde{\mathbb{Q}}_1^{\frac{1}{2}}\xi_\lambda\|\|\mathcal{N}_>^{\frac{1}{2}}\xi_\lambda\| +C\rho\langle \xi_\lambda, \mathcal{N}\xi_\lambda\rangle
\end{equation}
Note that the analysis for the constant term which can be extracted propagating $\mathbb{Q}_4$ can be done in the same way.
\end{proof}
\subsection{Regularization of $\mathbb{T}_2$ and $\mathbb{Q}_4$}
In order to use the scattering equation satisfied by $\varphi_\infty$, we need to regularize both $\mathbb{T}_2$ and $\mathbb{Q}_4$, to combine these operators with $\mathbb{T}_1$. We then define
\begin{equation}\label{eq: def T2r}
  \mathbb{T}_2^r = -\int dxdy\, V(x-y)\varphi(x-y) a_\uparrow(u_x^r)a_\uparrow(\overline{v}^r_x)a_\downarrow(u_y^r) a_\downarrow(\overline{v}^r_y)+ \mathrm{h.c.},
\end{equation}
and we recall that 
\begin{equation}
  \widetilde{\mathbb{Q}}_{4}^r = \int dxdy\, V(x-y) a_\uparrow(u_x^r)a_\uparrow (\overline{v}^r_x)a_\downarrow (u_y^r)  a_\downarrow(\overline{v}^r_y) + \mathrm{h.c.}
\end{equation}
\begin{proposition}\label{pro: reg} Let $\lambda\in [0,1]$. Let $\xi_\lambda$ be a state as in \eqref{eq: def xi lambda}. Under the same assumptions of Theorem \ref{thm: optimal up bd} and Theorem \ref{thm: lower bound}, it holds  
\begin{equation}\label{eq: T2reg up bd}
  |\langle \xi_\lambda, (\mathbb{T}_2 - \mathbb{T}_2^r)\xi_\lambda\rangle| \leq CL^{\frac{3}{2}}\rho^{\frac{3}{2}}\|\mathcal{N}^{\frac{1}{2}}\xi_\lambda\|. 
\end{equation}
Moreover, let $\eta >0$ and let 
\begin{equation}\label{eq: def tilde N}
  \widetilde{\mathcal{N}}:= \sum_{\sigma}\sum_{|k| \geq k_F + \rho^{\frac{1}{3} +\frac{\eta}{3}}} \hat{a}_{k,\sigma}\hat{a}_{k,\sigma}.
\end{equation}
It holds
\begin{equation}\label{eq: reg lw bd T2}
  |\langle \xi_\lambda, (\mathbb{T}_2 - \mathbb{T}_2^r) \xi_\lambda \rangle | \leq CL^{\frac{3}{2}}\rho^{\frac{3}{2}}\|\mathcal{N}_>^{\frac{1}{2}}\xi_\lambda\| + CL^{\frac{3}{2}}\rho^{\frac{3}{2} +\frac{\eta}{6}}\|\mathcal{N}^{\frac{1}{2}}\xi_\lambda\| + CL^{\frac{3}{2}}\rho^{\frac{3}{2}}\|\widetilde{\mathcal{N}}^{\frac{1}{2}}\xi_\lambda\|.
\end{equation} 
and
\begin{equation}\label{eq: reg lw bd Q4}
  |\langle \xi_\lambda, (\widetilde{\mathbb{Q}}_4 - \widetilde{\mathbb{Q}}_4^r) \xi_\lambda \rangle | \leq CL^3\rho^{2+\frac{\epsilon}{3}}  + \delta\langle \xi_\lambda, \widetilde{\mathbb{Q}}_1 \xi_\lambda\rangle +CL^{\frac{3}{2}}\rho^{\frac{3}{2}}\|\mathcal{N}_>^{\frac{1}{2}}\xi_\lambda\| + CL^{\frac{3}{2}}\rho^{\frac{3}{2} +\frac{\eta}{6}}\|\mathcal{N}^{\frac{1}{2}}\xi_\lambda\| + CL^{\frac{3}{2}}\rho^{\frac{3}{2}}\|\widetilde{\mathcal{N}}^{\frac{1}{2}}\xi_\lambda\|.
\end{equation}
\end{proposition}
\begin{remark}
The estimate in \eqref{eq: T2reg up bd} is sufficient to prove the optimal upper bound stated in Theorem \ref{thm: optimal up bd}. The ones in \eqref{eq: reg lw bd T2} and in \eqref{eq: reg lw bd Q4} instead are useful for the proof of Theorem \ref{thm: lower bound}.
We stress that the proof of \eqref{eq: T2reg up bd} can be done  as in \cite[Appendix C.2]{FGHP} using the decay estimates stated in Lemma \ref{lem: decay phi}. Below we mainly focus on the bounds \eqref{eq: reg lw bd T2} and \eqref{eq: reg lw bd Q4}.
\end{remark}
\begin{proof}
  We start by proving the bound for $\mathbb{T}_2 - \mathbb{T}_2^r$. As we already did, we can write
  \[
    \hat{u}_\sigma(k) = \hat{u}^r_\sigma(k) + \hat{\alpha}_\sigma(k) + \hat{\delta}^>_\sigma(k),
   \]
   with $\hat{\alpha}_\sigma(k)$ supported for $k\notin\mathcal{B}_F^\sigma$ and $|k| \leq 3k_F^\sigma$, and $\hat{\delta}^>_\sigma(k)$ supported for $|k| \geq \rho^{-\beta}$.
  In particular, proceeding similarly as in \cite[Appendix C.2]{FGHP}, one can prove that 
  \begin{eqnarray}
    \langle \xi_\lambda, (\mathbb{T}_2 - \mathbb{T}_2^r)\xi_\lambda\rangle &=& \int dxdy\, V(x-y)\varphi(x-y)\langle \xi_\lambda, a_\uparrow(\overline{v}^r_x)a_\uparrow(\alpha_x)a_\downarrow(\overline{v}^r_y)a_\downarrow(u^r_y)\xi_\lambda\rangle 
    \\
    && + \int dxdy\, V(x-y)\varphi(x-y)\langle \xi_\lambda, a_\uparrow(\overline{v}^r_x)a_\uparrow(\alpha_x)a_\downarrow(\overline{v}^r_y)a_\downarrow(\alpha_y)\xi_\lambda\rangle + \mathcal{E}\nonumber,
  \end{eqnarray}
  with, for any $n\in\mathbb{N}$ large enough,
   \begin{equation}\label{eq: est E reg lw bd}|\mathcal{E}| \leq C_n L^3\rho^{\beta n + \frac{13}{12}}.\end{equation}
   Note that in $\mathcal{E}$ we collect all the error terms in which there is at least one between $a_\uparrow(\delta^>_x)$ and $a_\downarrow(\delta^>_y)$.
   We underline once more that, as in \cite[Appendix C.2]{FGHP}, the bound above is a consequence of the decay estimates for $\varphi$ (see Lemma \ref{lem: decay phi}) and of the regularity of the interaction potential.
  We then estimate
  \begin{equation}
    \mathrm{I} := \int dxdy\, V(x-y)\varphi(x-y)\langle \xi_\lambda, a_\uparrow(\overline{v}^r_x)a_\uparrow(\alpha_x)a_\downarrow(\overline{v}^r_y)a_\downarrow(u^r_y)\xi_\lambda\rangle ,
  \end{equation}
  \begin{equation}
    \mathrm{II} := \int dxdy\, V(x-y)\varphi(x-y)\langle \xi_\lambda, a_\uparrow(\overline{v}^r_x)a_\uparrow(\alpha_x)a_\downarrow(\overline{v}^r_y)a_\downarrow(\alpha_y)\xi_\lambda\rangle.
  \end{equation}
In particular, using that $\|\alpha_{x,\sigma}\|_2 \leq C\rho^{1/2}$ and Lemma \ref{lem: bound phi}, we have 
\begin{equation}\label{eq: est I and II}
  |\mathrm{I}| \leq CL^{\frac{3}{{}2}}\|\varphi\|_\infty\|V\|_1\rho^{\frac{3}{2}}\|\mathcal{N}_>^{\frac{1}{2}}\xi_\lambda\|, \qquad |\mathrm{II}| \leq CL^{\frac{3}{{}2}}\|\varphi\|_\infty\|V\|_1\rho^{\frac{3}{2}}\|\mathcal{N}^{\frac{1}{2}}\xi_\lambda\|.
\end{equation}
Notice that the estimates above are sufficient to prove \eqref{eq: T2reg up bd}. 
We now prove the estimate \eqref{eq: reg lw bd T2}. In particular, we rewrite $\alpha_y$ in $\mathrm{II}$ as $\alpha_{y}^< + \alpha_y^>$, where $\hat{\alpha}^<_\sigma(k)$ is supported for $k_F^\sigma < |k| < k_F^\sigma + \rho^{1/3 + \eta/3}$ (for some $\eta>0$ to be chosen later) and $\hat{\alpha}^>_\sigma(k)$ is supported for $k_F + \rho^{1/3 + \eta/3} \leq |k| \leq 3k_F$. Correspondingly, we rewrite $\mathrm{II} = \mathrm{II}_< + \mathrm{II}_>$. Being $\|\alpha^<_{y,\sigma}\|_2 \leq C\rho^{1/2 + \eta/6}$, we can bound the term $\mathrm{II}_<$ as 
\begin{equation}\label{eq: est II< lw bd reg}
  |\mathrm{II}_< | \leq CL^{\frac{3}{2}}\rho^{\frac{3}{2} + \frac{\eta}{6}}\|\mathcal{N}^{\frac{1}{2}}\xi_\lambda\|. 
\end{equation} 
We now take into account $\mathrm{II}_>$. Being the momenta in $\hat{\alpha}^>(k)$ be such that $|k| \geq k_F + \rho^{1/3 + \eta/3}$, we get 
\begin{equation}\label{eq: bound alpha y}
  \int dy\, \|a_\downarrow(\alpha_y^>)\xi_\lambda\|^2 \leq \sum_{|k| \geq k_F + \rho^{\frac{1}{3} +\frac{\eta}{3}}}\langle \xi_\lambda,\hat{a}^\ast_{k, \downarrow} \hat{a}_{k,\downarrow}\xi_\lambda\rangle \leq \langle \xi_\lambda, \widetilde{\mathcal{N}}\xi_\lambda\rangle,
\end{equation}
according to the definition in \eqref{eq: def tilde N}. We can then estimate $\mathrm{II}_>$ as follows
\begin{equation}
  |\mathrm{II}_>| \leq\int dxdy\, V(x-y)\varphi(x-y) \|a_\uparrow(\overline{v}^r_x)\|\|a_\uparrow(\alpha_x)\|\|a_\downarrow(\overline{v}^r_y)\|\|a_\downarrow(\alpha_y)\xi_\lambda\|\leq CL^{\frac{3}{2}}\rho^{\frac{3}{2}}\|\widetilde{\mathcal{N}}^{\frac{1}{2}}\xi_\lambda\|.
\end{equation}
Thus, we get 
\begin{equation}\label{eq: est II lw bd reg}
  |\mathrm{II}| \leq CL^{\frac{3}{2}}\rho^{\frac{3}{2} +\frac{\eta}{6}}\|\mathcal{N}^{\frac{1}{2}}\xi_\lambda\| + CL^{\frac{3}{2}}\rho^{\frac{3}{2}}\|\widetilde{\mathcal{N}}^{\frac{1}{2}}\xi_\lambda\|.
\end{equation}
Combining the estimates in \eqref{eq: est E reg lw bd}, the one in \eqref{eq: est I and II} for the term $\mathrm{I}$ and \eqref{eq: est II lw bd reg} for the contribution $\mathrm{II}$, we get, taking $n \in \mathbb{N}$ large enough, 
\begin{equation}
  |\langle \xi_\lambda, (\mathbb{T}_2 - \mathbb{T}_2^r)\xi_\lambda\rangle | \leq CL^{\frac{3}{2}}\rho^{\frac{3}{2}}\|\mathcal{N}_>^{\frac{1}{2}}\xi_\lambda\| + CL^{\frac{3}{2}}\rho^{\frac{3}{2} +\frac{\eta}{6}}\|\mathcal{N}^{\frac{1}{2}}\xi_\lambda\| + CL^{\frac{3}{2}}\rho^{\frac{3}{2}}\|\widetilde{\mathcal{N}}^{\frac{1}{2}}\xi_\lambda\|,
\end{equation}
which concludes the proof of \eqref{eq: reg lw bd T2}.
The bound for $\langle \xi_\lambda, (\widetilde{\mathbb{Q}}^r_4 - \widetilde{\mathbb{Q}}_4)\xi_\lambda\rangle$ can be proven similarly. In particular, the proof for the estimate in \eqref{eq: reg lw bd Q4} can be done similarly as above. Note that the contribution $\langle \xi_\lambda, \mathbb{Q}_1\xi_\lambda\rangle + CL^3\rho^{2+\epsilon/3}$ in the right hand side of \eqref{eq: reg lw bd Q4} comes from the replacement of $a_\uparrow(\overline{v}_x) a_\downarrow(\overline{v}_y)$ with $a_\uparrow(\overline{v}_x^r) a_\downarrow(\overline{v}_y^r)$. We omit the details, see \cite[Appendix C.2]{FGHP}.
\end{proof}
\subsection{Scattering equation cancellation}
We can now use the equation satisfied by $\varphi_\infty$ (see \eqref{eq: equation for tilde phi}) to get a cancellation up to small error terms.
\begin{proposition}[Scattering equation cancellation]\label{pro: scatt canc} Let $\lambda\in [0,1]$. Let $\xi_\lambda$ be a state as in \eqref{eq: def xi lambda}. Under the same assumptions of Theorem \ref{thm: optimal up bd} and Theorem \ref{thm: lower bound}, it holds
\begin{equation}
  |\langle\xi_\lambda, (\mathbb{T}_1 + \mathbb{T}_2^r + \widetilde{\mathbb{Q}}_4^r)\xi_\lambda \rangle|\leq CL^{\frac{3}{2}}\rho^{\frac{3}{2}}\|\mathcal{N}_>^{\frac{1}{2}}\xi_\lambda\|.
\end{equation}
\end{proposition}
\begin{proof} We have 
\begin{equation}
  \langle\xi_\lambda, (\mathbb{T}_1 + \mathbb{T}_2^r + \widetilde{\mathbb{Q}}_4^r)\xi_\lambda \rangle 
  =  \int dxdy  (2\Delta\varphi + V(1-\varphi)) \langle  \xi_\lambda, a_\uparrow(u^r_x)a_\uparrow(\overline{v}^r_x)a_\downarrow(u^r_y)a_\downarrow(\overline{v}^r_y)\xi_\lambda\rangle + \mathrm{c.c.}
\end{equation}
Similarly as in Lemma \ref{lem: error kinetic energy} we now use the scattering equation. More precisely, being all the quantities we are integrating periodic and using that for each $x\in \Lambda_L$, for $L$ large enough, $\Delta\varphi(x)= \Delta\varphi_\infty(x)$, $V(x) = V_\infty(x)$, $\varphi(x) = \varphi_\infty(x)$ and $\mathcal{E}_{\varphi_0, \chi_{\sqrt[3]\rho}}(x) = \mathcal{E}^\infty_{\varphi_0, \chi_{\sqrt[3]\rho}}(x)$, we get 
\begin{multline*}
  \langle\xi_\lambda, (\mathbb{T}_1 + \mathbb{T}_2^r + \widetilde{\mathbb{Q}}_4^r)\xi_\lambda \rangle 
  \\
  = \int dxdy \, \mathcal{E}_{\varphi_0, \chi_{\sqrt[3]\rho}}(x-y) \langle  \xi_\lambda, a_\uparrow(u^r_x)a_\uparrow(\overline{v}^r_x)a_\downarrow(u^r_y)a_\downarrow(\overline{v}^r_y)\xi_\lambda\rangle + \mathrm{c.c.} 
\end{multline*}
We can now write that
\begin{multline*}
 \int dxdy  \,\mathcal{E}_{\varphi_0, \chi_{\sqrt[3]\rho}}(x-y) \langle  \xi_\lambda, a_\uparrow(u^r_x)a_\uparrow(\overline{v}^r_x)a_\downarrow(u^r_y)a_\downarrow(\overline{v}^r_y)\xi_\lambda\rangle 
  \\
   = \int dy \langle \xi_\lambda, b_\uparrow((\mathcal{E}_{\varphi_0, \chi_{\sqrt[3]\rho}})_y) a_\downarrow(u^r_y)a_\downarrow(\overline{v}^r_y)\xi_\lambda\rangle,
\end{multline*}
where we used the notations introduced in Lemma \ref{lem: b error scatt}.
We then get, using Proposition \ref{pro: L1, L2 v} and Lemma \ref{lem: b error scatt}
\begin{equation}
  | \langle\xi_\lambda, (\mathbb{T}_1 + \mathbb{T}_2^r + \widetilde{\mathbb{Q}}_4^r)\xi_\lambda \rangle| \leq \int dy\, \|b_\uparrow((\mathcal{E}_{\varphi_0, \chi_{\sqrt[3]\rho}})_y)\|\|a_\downarrow(\overline{v}^r_y)\|\|a_\downarrow(\overline{v}^r_y)\xi_\lambda\| \leq  CL^{\frac{3}{2}}\rho^{\frac{3}{2}}\|\mathcal{N}_>^{\frac{1}{2}}\xi_\lambda\|,
\end{equation}
and the result holds.
\end{proof}
\subsection{Propagation of the estimates: conclusions}
In this section we prove some propagation estimates, which allow us to close all the estimates in Proposition \ref{pro: H0}, Lemma \ref{lem: error kinetic energy}, Proposition \ref{pro: Q1}, Proposition \ref{pro: Q4} and Proposition \ref{pro: reg}.
\begin{proposition}\label{pro: propagation est}
  Let $\lambda\in [0,1]$. Let $\xi_\lambda$ be as in \eqref{eq: def xi lambda}. Under the same assumptions of Theorem \ref{thm: optimal up bd} and Theorem \ref{thm: lower bound}, it holds
  \begin{equation}\label{eq: propagation H0 and Q1}
   \langle\xi_\lambda,\mathbb{H}_0\xi_\lambda\rangle \leq CL^3\rho^2, \qquad  \langle\xi_\lambda,\mathbb{Q}_{1}\xi_\lambda\rangle \leq CL^3\rho^2, \qquad \langle\xi_\lambda,\widetilde{\mathbb{Q}}_{1}\xi_\lambda\rangle \leq CL^3\rho^2.
  \end{equation}
\end{proposition}
\begin{proof}
  The bound for $\widetilde{\mathbb{Q}}_1$  in \eqref{eq: propagation H0 and Q1} follow directly from the fact that $\widetilde{\mathbb{Q}}_{1} \leq \mathbb{Q}_{1}$. We then prove the first two bounds. The proof follows the ideas in \cite[Proposition 5.8]{FGHP}. We recall here the main steps.
  We have, by Proposition \ref{pro: H0} and Proposition \ref{pro: Q1},  
  \begin{eqnarray}
    \frac{d}{d\lambda}\langle\xi_\lambda, (\mathbb{H}_0 + \mathbb{Q}_{1})\xi_\lambda\rangle  &=& - \langle\xi_\lambda, (\mathbb{T}_1 + \mathbb{T}_2)\xi_\lambda\rangle + \mathcal{E}_{\mathbb{H}_0}(\xi_\lambda) + \mathcal{E}_{\mathbb{Q}_{1}}(\xi_\lambda)
    \\
    &=& - \langle\xi_\lambda, (\mathbb{T}_1 + \mathbb{T}_2^r)\xi_\lambda\rangle + \mathcal{E}_{\mathbb{H}_0}(\xi_\lambda)+ \langle\xi_\lambda (\mathbb{T}_2^r - \mathbb{T}_2)\xi_\lambda\rangle + \mathcal{E}_{\mathbb{Q}_{1}}(\xi_\lambda).\nonumber
      \end{eqnarray}
    From Proposition \ref{pro: reg} together with the non-optimal bound in \eqref{eq: non optimal bound N}, we get
    \begin{equation}
        |\langle \xi_\lambda, (\mathbb{T}_2^r  - \mathbb{T}_2)\xi_\lambda\rangle | \leq  CL^{\frac{3}{2}}\rho^{\frac{3}{2}}\|\mathcal{N}^{\frac{1}{2}}\xi_\lambda\| \leq C L^3\rho^{2+\frac{1}{12}}.
\end{equation} 
 Moreover, we also know from Proposition \ref{pro: H0}
      \begin{equation}
        \mathcal{E}_{\mathbb{H}_0}(\xi_\lambda)= 2\sum_{j=1}^3\int dxdy\, \partial_j \varphi(x-y) (a_\uparrow(u^r_x)a_\uparrow(\partial_j\overline{v}^r_x)a_\downarrow(u^r_y) a_\downarrow(\overline{v}^r_y) - a_\uparrow(u^r_x)a_\downarrow(\overline{v}^r_x)a_\downarrow(u^r_y) a_\downarrow(\partial_j\overline{v}^r_y)) + \mathrm{h.c.}\nonumber
      \end{equation}
     We can then prove the following non optimal bound
      \begin{equation} \label{eq: err kin}
        |\mathcal{E}_{\mathbb{H}_0}(\xi_\lambda)| \leq  C\sum_{j=1}^3\int dx\, |\langle \xi_\lambda, b_\downarrow( (\partial_j\varphi)_x)a_\uparrow(\partial_j\overline{v}^r_x) a_\uparrow(u^r_x)\xi_\lambda\rangle| \leq CL^{\frac{3}{2}}\rho^{\frac{4}{3}}\|\mathcal{N}_>^{\frac{1}{2}} \xi_\lambda\| \leq CL^3\rho^2 + C\langle \xi_\lambda, \mathbb{H}_0\xi_\lambda\rangle,
      \end{equation}
      where we used the bound from Proposition \ref{lem: bound N>}.
      Note that the bound above is not optimal but is enough to conclude the proof.
Moreover, from Proposition \ref{pro: Q1}, we get
      \begin{eqnarray}\label{eq: err Q1}
        |\mathcal{E}_{\mathbb{Q}_{1}}(\xi_\lambda)| &\leq& C\rho^{\frac{1}{2}}\|\mathbb{Q}^{\frac{1}{2}}_{1}\xi_\lambda\|\|\mathbb{H}_0^{\frac{1}{2}}\xi_\lambda\| + CL^{\frac{3}{2}}\rho^{\frac{4}{3}}\|\mathbb{Q}^{\frac{1}{2}}_{1}\xi_\lambda\|\nonumber
        \\
        &\leq& CL^3\rho^{\frac{8}{3}} + C\rho^{\frac{1}{2}}\langle \xi_\lambda, (\mathbb{H}_0 + \mathbb{Q}_{1})\xi_\lambda\rangle. 
      \end{eqnarray}
      Using the cancellation via to the scattering equation (see Proposition \ref{pro: scatt canc}), we have using the non-optimal estimate $\langle \xi_\lambda, \mathcal{N}\xi_\lambda\rangle \leq CL^3\rho^{7/6}$,
      \begin{equation}\label{eq: main prop}
        \langle \xi_\lambda,(\mathbb{T}_1 + \mathbb{T}_2^r)\xi_\lambda \rangle = \langle \xi_\lambda,(\mathbb{T}_1 + \mathbb{T}_2^r + \widetilde{\mathbb{Q}}_4^r)\xi_\lambda \rangle - \langle \xi_\lambda,\widetilde{\mathbb{Q}}_4^r\xi_\lambda \rangle = - \langle \xi_\lambda,\widetilde{\mathbb{Q}}_4^r\xi_\lambda \rangle + \mathcal{O}(L^3\rho^{2+\frac{1}{12}}),
      \end{equation}
      where we used the notations introduced in Section \ref{sec: bounds Q4}.  
      It is easy to prove that $|\langle \xi_\lambda, \widetilde{\mathbb{Q}}_4\xi_\lambda\rangle| \leq CL^3\rho^2$. Indeed, from Proposition \ref{pro: reg} (see the bound in \eqref{eq: reg lw bd Q4}), we have (using the non optimal bound $\langle \xi_\lambda, \mathcal{N}\xi_\lambda\rangle \leq CL^3\rho^{7/6}$), 
      \begin{eqnarray}
        |\langle\xi_\lambda, \widetilde{\mathbb{Q}}_4^r\xi_\lambda\rangle| &\leq& |\langle\xi_\lambda, (\widetilde{\mathbb{Q}}_4^r - \widetilde{\mathbb{Q}}_4)\xi_\lambda\rangle| + |\langle\xi_\lambda, \widetilde{\mathbb{Q}}_4\xi_\lambda\rangle|\nonumber
        \\
        &\leq& CL^3\rho^{2} + C\langle \xi_\lambda, \mathbb{Q}_1 \xi_\lambda\rangle + |\langle\xi_\lambda, \widetilde{\mathbb{Q}}_4\xi_\lambda\rangle|\nonumber
        \\
        &\leq& CL^3\rho^2 +  C\langle \xi_\lambda, \mathbb{Q}_1 \xi_\lambda\rangle,\nonumber
        \end{eqnarray}
      where the last inequality can be proved  via Cauchy-Schwarz as in \cite[Proposition 3.3]{FGHP}.
Combining all the bounds together, we find
\begin{equation} 
\frac{d}{d\lambda}\langle \xi_\lambda, (\mathbb{H}_0 + \mathbb{Q}_{1})\xi_\lambda\rangle \leq CL^3\rho^{2} + \langle\xi_\lambda(\mathbb{H}_0 +\mathbb{Q}_{1})\xi_\lambda\rangle.
\end{equation}
From the bounds in Lemma \ref{lem: a priori}, i.e., 
\begin{equation}
  \langle\xi_0, \mathbb{H}_0 \xi_0\rangle \equiv  \langle R^\ast\psi, \mathbb{H}_0 R^\ast\psi \rangle \leq CL^3\rho^2,\qquad \langle\xi_0, \mathbb{Q}_{1} \xi_0\rangle \equiv  \langle R^\ast\psi, \mathbb{Q}_{1} R^\ast\psi \rangle \leq CL^3\rho^2,
\end{equation}
and using Gr\"onwall's Lemma, we conclude that
\begin{equation}
   |\langle\xi_\lambda,\mathbb{H}_0\xi_\lambda\rangle| \leq CL^3\rho^2, \quad  |\langle\xi_\lambda,\mathbb{Q}_{1}\xi_\lambda\rangle| \leq CL^3\rho^2.
\end{equation}
\end{proof}
%

\section{{An optimal upper bound on the ground state energy density}}\label{sec: up bd}
In this section we conclude the proof of the upper bound. As a trial state we take $\psi = RT\Omega$, with $R$ being the particle-hole transformation introduced in Definition \ref{def: ferm bog} and $T$ being the almost bosonic Bogoliubov transformation in Definition \ref{def: transf T}. On can easily prove (see \cite[Section 7]{FGHP}) that $\psi$ is an $N-$particle state with $N_\uparrow$ particles with spin $\uparrow$ and $N_\downarrow$ particles with spin $\downarrow$. Moreover, being $R$ and $T$ unitary operators, we also have $\|\psi\| = \|\Omega\|=1$. In the sequel we are going to use the following bound for the number operator (which is a consequence of \eqref{eq: est N}):
\begin{equation}\label{eq: est N up bd}
  \langle\xi_\lambda, \mathcal{N}\xi_\lambda\rangle \leq CL^3 \rho^{\frac{5}{3}}, \qquad \xi_\lambda := T_{1-\lambda}\Omega.
\end{equation}
We then can write 
\begin{eqnarray}
  E_L(N_\uparrow, N_\downarrow) &\leq& E_{HF}(\omega) + \langle T\Omega, \mathbb{H}_0 T\Omega\rangle + \langle T\Omega, \mathbb{X}T\Omega\rangle + \langle T\Omega, \mathbb{Q} T\Omega\rangle
  \\
  &\leq& E_{\mathrm{HF}}(\omega) + \langle T\Omega, (\mathbb{H}_0 + \mathbb{Q}_{1} + \mathbb{Q}_4) T\Omega\rangle + \mathcal{E}_1,
\end{eqnarray}
where 
\begin{equation}
  \mathcal{E}_1 := \langle T\Omega, (\mathbb{X} + \mathbb{Q}_2 + \mathbb{Q}_3)T\Omega\rangle.
\end{equation}
As in \cite[Section 7]{FGHP}, one can prove that $\langle T\Omega, \mathbb{Q}_3 T\Omega \rangle = 0$ (this is due to the spin effect). Moreover, via Proposition \ref{pro: X,Q2} together with the bound in \eqref{eq: est N up bd}, we know that 
\begin{equation}
  |\mathcal{E}_1| \leq C\rho\langle T\Omega, \mathcal{N}T\Omega\rangle \leq CL^3\rho^{\frac{8}{3}}.
\end{equation}
Before proceeding further, we also notice that  
\begin{equation}
  \langle T\Omega, \mathbb{Q}_4T\Omega\rangle = \langle T\Omega, \widetilde{\mathbb{Q}}_4 T\Omega\rangle + \langle  T\Omega, \widehat{\mathbb{Q}}_4 T\Omega\rangle = \langle T\Omega, \widetilde{\mathbb{Q}}_4 T\Omega\rangle, 
\end{equation}
where $\widehat{\mathbb{Q}}_4$ is the contribution in $\mathbb{Q}_4$ coming from aligned spin and $\widetilde{\mathbb{Q}}_4$ is as in \eqref{eq: def tilde Q4}. The equality above follows from the fact that the transformation $T$ is commuting with the spin operator $S= \sum_\sigma \sigma\mathcal{N}_\sigma$, see \cite[Section 7]{FGHP} for more details.
We then have
\begin{equation}
  E_L(N_\uparrow, N_\downarrow) \leq E_{\mathrm{HF}}(\omega) + \langle T\Omega, (\mathbb{H}_0 + \mathbb{Q}_{1} + \widetilde{\mathbb{Q}}_4)T\Omega\rangle + CL^3\rho^{\frac{8}{3}}.
\end{equation}
We can now write
\begin{equation} \label{eq: 1st step duhamel}
\langle T\Omega, (\mathbb{H}_0 + \mathbb{Q}_{1} + \widetilde{\mathbb{Q}}_4)T\Omega \rangle = \int_0^1 d\lambda\, \langle \xi_\lambda, (\mathbb{T}_1 + \mathbb{T}_2)\xi \lambda \rangle + \langle T\Omega, \widetilde{\mathbb{Q}}_4 T\Omega \rangle +\mathcal{E}_2.
\end{equation}
where $\mathbb{T}_1$ and $\mathbb{T}_2$ are as in \eqref{eq: def T1} and \eqref{eq: def T2}, respectively, i.e., 
\begin{equation} 
  \mathbb{T}_1 = 2\int dxdy\, \Delta\varphi(x-y) a_\uparrow(u^r_x)a_\uparrow(\overline{v}^r_x)a_\downarrow(u^r_y) a_\downarrow(\overline{v}^r_y) + \mathrm{h.c.}
\end{equation}
\begin{equation}
  \mathbb{T}_2 = -\int dxdy\, V(x-y)\varphi(x-y) a_\uparrow(\overline{v}^r_z)a_\uparrow(u_x)a_\downarrow(\overline{v}^r_{z^\prime})a_\downarrow(u_y) + \mathrm{h.c.}
\end{equation}
Moreover, from Proposition \ref{pro: H0}, Lemma \ref{lem: error kinetic energy} and Proposition \ref{pro: Q1}, we can bound 
\begin{equation}\label{eq: est E2 up bd}
  |\mathcal{E}_2|\leq CL^3\rho^{\frac{7}{3}} + CL^3\rho^{\frac{7}{3} + \frac{1}{3} - \frac{\epsilon}{3}},
\end{equation}
where we used that $\xi_1 = \Omega$ together with the propagation estimates proved in Proposition \ref{pro: propagation est} and the estimate for the number operator in \eqref{eq: est N up bd}.
We now write 
\begin{equation}
  \langle \xi_\lambda, (\mathbb{T}_1 + \mathbb{T}_2) \xi_\lambda\rangle = \langle \xi_\lambda, (\mathbb{T}_2 - \mathbb{T}_2^r) \xi_\lambda\rangle + \langle \xi_\lambda, (\mathbb{T}_1 + \mathbb{T}^r_2 + \widetilde{\mathbb{Q}}_4^r)\xi_\lambda \rangle - \langle \xi_\lambda, \widetilde{\mathbb{Q}}_4^r\xi_\lambda\rangle, 
\end{equation}
From Proposition \ref{pro: reg}, we can bound 
\begin{equation}
  |\langle\xi_\lambda, (\mathbb{T}_2 - \mathbb{T}_2^r) \xi_\lambda\rangle | \leq CL^{\frac{3}{2}}\rho^{\frac{3}{2}}\|\mathcal{N}^{\frac{1}{2}}\xi_\lambda\| \leq CL^3\rho^{\frac{7}{3}},
\end{equation}
where we also used the bound for the number operator in \eqref{eq: est N up bd} and from Proposition \ref{pro: scatt canc} we get
\begin{equation}
  |\langle \xi_\lambda, (\mathbb{T}_1 + \mathbb{T}^r_2 + \widetilde{\mathbb{Q}}_4^r)\xi_\lambda \rangle| \leq CL^3\rho^{\frac{7}{3}},
\end{equation}
where we used again \eqref{eq: est N up bd}.
 We then have
\begin{multline}
  \langle T\Omega, \widetilde{\mathbb{Q}}_4T\Omega\rangle + \int_0^1 \, d\lambda \langle\xi_\lambda, (\mathbb{T}_1 + \mathbb{T}_2)\xi_\lambda\rangle
  \\
  =\langle T\Omega, \widetilde{\mathbb{Q}}_4 T\Omega\rangle - \int_0^1\langle \xi_{\lambda}, \widetilde{\mathbb{Q}}_4^r \xi_{\lambda}\rangle + \mathcal{E}_3
  \\
  = -\int_0^1 d\lambda\, \frac{d}{d\lambda}\, \langle \xi_\lambda, \widetilde{\mathbb{Q}}_4\xi_\lambda\rangle + \int_0^1 d\lambda\int_\lambda^1 d\lambda^\prime\, \frac{d}{d\lambda^\prime}\langle \xi_{\lambda^\prime},\widetilde{\mathbb{Q}}_4^r \xi_{\lambda^\prime}\rangle + \mathcal{E}_3,
\end{multline}
with $|\mathcal{E}_3| \leq CL^3 \rho^{7/3}$.
From Proposition \ref{pro: Q4}, we find 
\begin{multline}\label{eq: Q4 upper bd}
  -\int_0^1 d\lambda\, \frac{d}{d\lambda}\, \langle \xi_\lambda, \widetilde{\mathbb{Q}}_4\xi_\lambda\rangle + \int_0^1d\lambda\int_\lambda^1 d\lambda^\prime\, \frac{d}{d\lambda^\prime}\langle \xi_{\lambda^\prime},\widetilde{\mathbb{Q}}_4^r\xi_{\lambda^\prime}\rangle
  \\
  = 
- \rho_\uparrow\rho_\downarrow \int dxdy \, V(x-y)\varphi(x-y)  + \mathcal{E}_4,
\end{multline}
with 
\begin{equation}
  |\mathcal{E}_4| \leq CL^3\rho^{\frac{7}{3}} + CL^3\rho^{\frac{7}{3} + \frac{1}{3} -\frac{\epsilon}{3}} + CL^3\rho^{2+\frac{\epsilon}{3}},
\end{equation}
where we used again Proposition \ref{pro: propagation est} and the bound for the number operator in  \eqref{eq: est N up bd}.
Thus, putting all the estimates together, we find that
\begin{eqnarray}
  E_{L}(N_\uparrow, N_\downarrow) &\leq& E_{\mathrm{HF}}(\omega) -\rho_\uparrow\rho_\downarrow\int\, dxdy\, V(x-y)\varphi(x-y)  \nonumber
  \\
  && + CL^3\rho^{\frac{7}{3}} + CL^3\rho^{2+\frac{\epsilon}{3}} + CL^3\rho^{\frac{7}{3} + \frac{1}{3} -\frac{\epsilon}{3}}\nonumber
  \\
  &\leq& E_{\mathrm{HF}}(\omega) -\rho_\uparrow\rho_\downarrow\int dxdy\, V(x-y)\varphi(x-y) + CL^3\rho^{\frac{7}{3}}, 
\end{eqnarray}
where the last step follows after an optimization over $\epsilon$, which yields $\epsilon = 1$. Moreover, for the constant term, using the periodicity of $V$ and $\varphi$ and that for each $x\in \Lambda_L$, for $L$ large enough, $V(x) = V_\infty(x)$, $\varphi(x) = \varphi_\infty(x)$, we have
\begin{equation}
  \int_{\Lambda_L \times \Lambda_L}\, dxdy\, V(x-y)\varphi(x-y) = L^3\int_{\mathbb{R}^3} dx V_\infty (x)\varphi_\infty(x) = L^3\int_{\mathbb{R}^3}dx\, V_\infty (x)\varphi_0(x),
\end{equation}
where we also used that $V_\infty$ has compact support (see Assumption \ref{asu: potential V}) and that $\varphi_\infty = \varphi_0$ in the support of the interaction potential (recall that we denote by $\varphi_0$ the solution of the zero energy scattering equation in $\mathbb{R}^3$, i.e., \eqref{eq: def phi not app}).
We can conclude that, for $L$ large enough,
\begin{equation}
 e_L(\rho_\uparrow, \rho_\downarrow) = \frac{E_{L}(N_\uparrow, N_\downarrow)}{L^3} \leq \frac{3}{5}(6\pi^2)^{\frac{2}{3}}(\rho_\uparrow^{\frac{5}{3}} + \rho_\downarrow^{\frac{5}{3}}) + 8\pi a \rho_\uparrow \rho_\downarrow+ C\rho^{\frac{7}{3}},
 \end{equation}
 where we also used the expression for $E_{\mathrm{HF}}(\omega)/L^3$ in \eqref{eq: HF unit volume} and the fact that (see Lemma \ref{lem: bound phi0}) 
\[
   \int_{\mathbb{R}^3} V_\infty(x) (1-\varphi_0)(x) = 8\pi a
\]
This concludes the proof of the upper bound.
\section{An improved lower bound on the ground state energy density}\label{sec: lw bd}
In this section we prove Theorem \ref{thm: lower bound}. In what follows, $\psi$ will be an approximate ground state, in the sense of Definition \ref{def: approx gs}. Before discussing the proof of Theorem \ref{thm: lower bound}, we will
\begin{itemize}
  \item Improve the bounds for some error terms (Section \ref{sec: improv N lw bd}).
  \item Estimate the cubic term $\widetilde{\mathbb{Q}}_3$ (Section \ref{sec: Q3}).
\end{itemize}
\subsection{Improved bounds for the lower bound}\label{sec: improv N lw bd}
In this section we take into account the errors coming from Lemma \ref{lem: error kinetic energy}, Proposition \ref{pro: Q4}, Proposition \ref{pro: reg} and Proposition \ref{pro: scatt canc} and we will improve them making use of Proposition \ref{pro: impr est N}, Proposition \ref{lem: bound N>} and Corollary \ref{cor: N tilde}. In what follows $0<\delta<1$ denotes a small constant whose value can change from line to line.

From Lemma \ref{lem: error kinetic energy}, we can bound 
\begin{eqnarray*}
  |\mathcal{E}_{\mathbb{H}_0}(\xi_\lambda)| &\leq& \delta\langle\xi_1, \mathbb{H}_0\xi_1\rangle + CL^3\rho^{\frac{7}{3}}  + C\rho^{\frac{7}{6} -\frac{\epsilon}{3}}\langle \xi_\lambda, \mathcal{N}_> \xi_\lambda\rangle + C\rho^{\frac{2}{3}-\frac{\epsilon}{3}}\langle \xi_\lambda, \mathbb{H}_0 \xi_\lambda\rangle  + C\rho\langle \xi_\lambda, \mathcal{N}\xi_\lambda\rangle
  \\
  && + C\rho^{\frac{5}{6} -\frac{\epsilon}{3}}\|\mathcal{N}_>^{\frac{1}{2}}\xi_\lambda\|\big(\|\widetilde{\mathbb{Q}}_1^{\frac{1}{2}}\xi_\lambda\| + \|\mathbb{H}_0^{\frac{1}{2}}\xi_\lambda\|\big).
\end{eqnarray*}
From Proposition \ref{pro: impr est N}, using the propagation estimates proved in Proposition \ref{pro: propagation est} and proceeding as in \cite[Proposition 5.9]{FGHP}, we get 
\begin{equation}\label{eq: improv N lw bd}
  \langle \xi_\lambda, \mathcal{N}\xi_\lambda\rangle \leq CL^{\frac{3}{2}}\rho^{\frac{1}{6}}\|\mathbb{H}_0^{\frac{1}{2}}\xi_1\| + CL^3\rho^{\frac{5}{3}}. 
\end{equation}
As a consequence, there exists $0<\delta<1$ such that 
\begin{equation}
  C\rho\langle \xi_\lambda, \mathcal{N}\xi_\lambda\rangle\leq \delta\langle \xi_1, \mathbb{H}_0 \xi_1\rangle + CL^3\rho^{\frac{7}{3}}.
\end{equation}
Moreover, from Proposition \ref{lem: bound N>}, we know that 
\begin{equation}
  C\rho^{\frac{7}{6} -\frac{\epsilon}{3}}\langle \xi_\lambda, \mathcal{N}_> \xi_\lambda\rangle \leq C\rho^{\frac{7}{6} -\frac{2}{3} -\frac{\epsilon}{3}}\langle \xi_1,\mathbb{H}_0 \xi_1\rangle + CL^3\rho^{\frac{7}{3} +\frac{1}{2} -\frac{\epsilon}{2}} \leq C\rho^{\frac{1}{2} -\frac{\epsilon}{3}}\langle \xi_1,\mathbb{H}_0 \xi_1\rangle + CL^3\rho^{\frac{7}{3} +\frac{1}{2} -\frac{\epsilon}{3}},
\end{equation}
and, similarly, for $0<\delta <1$, we have
\begin{eqnarray}
  C\rho^{\frac{5}{6} -\frac{\epsilon}{3}}\|\mathcal{N}_>^{\frac{1}{2}}\xi_\lambda\|\big(\|\widetilde{\mathbb{Q}}_1^{\frac{1}{2}}\xi_\lambda\| + \|\mathbb{H}_0^{\frac{1}{2}}\xi_\lambda\|\big)&\leq& \big(CL^{\frac{3}{2}}\rho^{\frac{5}{3} -\frac{\epsilon}{3}} + C\rho^{\frac{5}{6} -\frac{1}{3} -\frac{\epsilon}{3}}\|\mathbb{H}_0^{\frac{1}{2}}\xi_1\|\big)(\|\widetilde{\mathbb{Q}}_1^{\frac{1}{2}}\xi_\lambda\| + \|\mathbb{H}_0^{\frac{1}{2}}\xi_\lambda\|\big)\nonumber
  \\
  &\leq& CL^3\rho^{\frac{7}{3} + \frac{1}{3} -\frac{\epsilon}{3}} + CL^3\rho^{\frac{7}{3} + \frac{2}{3} -\frac{2\epsilon}{3}} + \delta\langle \xi_1, \mathbb{H}_0 \xi_1\rangle ,
\end{eqnarray}
where in the last inequality we used Cauchy-Schwarz inequality together with Proposition \ref{pro: propagation est}. All together, there exists $0 <\delta <1$, such that, for $0< \epsilon \leq 1$:
\begin{equation}\label{eq: conclusive H0 lw bd}
 |\mathcal{E}_{\mathbb{H}_0}(\xi_\lambda)| \leq CL^3\rho^{\frac{7}{3}} + CL^3 \rho^{\frac{7}{3} +\frac{1}{3} -\frac{\epsilon}{3}} + \delta \langle\xi_1, \mathbb{H}_0 \xi_1 \rangle.
\end{equation}
From Proposition \ref{pro: Q4}, we have 
 \begin{equation}
 |\mathcal{E}_{\widetilde{\mathbb{Q}}_4}(\xi_\lambda)|,  |\mathcal{E}_{\widetilde{\mathbb{Q}}^r_4}(\xi_\lambda)|\leq CL^3\rho^{\frac{7}{3}} + CL^3\rho^{2+\frac{\epsilon}{3}} +  C\rho^{\frac{5}{6} -\frac{\epsilon}{3}}\|\widetilde{\mathbb{Q}}_1^{\frac{1}{2}}\xi_\lambda\|\|\mathcal{N}_>^{\frac{1}{2}}\xi_\lambda\| +C\rho\langle \xi_\lambda, \mathcal{N}\xi_\lambda\rangle
\end{equation}
Proceeding similar as above, we get, for $0<\epsilon \leq 1$:
\begin{equation}\label{eq: conclusive Q4 lw bd}
  |\mathcal{E}_{\widetilde{\mathbb{Q}}_4}(\xi_\lambda)|, |\mathcal{E}_{\widetilde{\mathbb{Q}}^r_4}(\xi_\lambda)|\leq CL^3\rho^{\frac{7}{3}} + CL^3 \rho^{\frac{7}{3} +\frac{1}{3} -\frac{\epsilon}{3}} + CL^3 \rho^{2+\frac{\epsilon}{3}} + \delta \langle\xi_1, \mathbb{H}_0 \xi_1 \rangle,
\end{equation}
We now take into account the errors we due with the regularizations, i.e., the ones estimated in Proposition \ref{pro: reg}. We have 
\begin{equation*}
  |\langle \xi_\lambda, (\mathbb{T}_2 - \mathbb{T}_2^r) \xi_\lambda \rangle | \leq CL^{\frac{3}{2}}\rho^{\frac{3}{2}}\|\mathcal{N}_>^{\frac{1}{2}}\xi_\lambda\| + CL^{\frac{3}{2}}\rho^{\frac{3}{2} +\frac{\eta}{6}}\|\mathcal{N}^{\frac{1}{2}}\xi_\lambda\| + CL^{\frac{3}{2}}\rho^{\frac{3}{2}}\|\widetilde{\mathcal{N}}^{\frac{1}{2}}\xi_\lambda\|.
\end{equation*} 
From Proposition \ref{lem: bound N>}, we get 
\[
  CL^{\frac{3}{2}}\rho^{\frac{3}{2}}\|\mathcal{N}_>^{\frac{1}{2}}\xi_\lambda\| \leq CL^3\rho^{\frac{7}{3}} +\delta \langle \xi_1,\mathbb{H}_0 \xi_1 \rangle,
\]
for some $0<\delta<1$.
Moreover, using \eqref{eq: improv N lw bd} together with Young's inequality, we find 
\begin{equation*}
 CL^{\frac{3}{2}}\rho^{\frac{3}{2} +\frac{\eta}{6}}\|\mathcal{N}^{\frac{1}{2}}\xi_\lambda\|\leq  CL^3\rho^{2+ \frac{2\eta}{9} + \frac{1}{9}} + \delta\langle \xi_1, \mathbb{H}_0 \xi_1\rangle, 
\end{equation*} 
for some $0<\delta<1$. To estimate the error term with respect to $\widetilde{\mathcal{N}}$, we recall that
\[
  \widetilde{\mathcal{N}}= \sum_{\sigma}\sum_{|k| \geq k_F + \rho^{\frac{1}{3} +\frac{\eta}{3}}} \hat{a}_{k,\sigma}\hat{a}_{k,\sigma}.
\]
Thus, via Gronwall's lemma (see Corollary \ref{cor: N tilde}) one gets
\begin{equation}
  \langle \xi_\lambda, \widetilde{\mathcal{N}}\xi_\lambda\rangle \leq \langle \xi_1, \widetilde{\mathcal{N}}\xi_1\rangle + CL^3\rho^{\frac{5}{3}}\leq C\rho^{-\frac{2}{3} -\frac{\eta}{3}}\langle \xi_1,\mathbb{H}_0 \xi_1\rangle + CL^3\rho^{\frac{5}{3}}
\end{equation}
We can then estimate 
\[
   CL^{\frac{3}{2}}\rho^{\frac{3}{2}}\|\widetilde{\mathcal{N}}^{\frac{1}{2}}\xi_\lambda\|\leq  CL^3\rho^{\frac{7}{3}} + CL^{\frac{3}{2}}\rho^{\frac{3}{2} -\frac{1}{3}-\frac{\eta}{6}}\|\mathbb{H}_0^{\frac{1}{2}}\xi_1\| \leq CL^3\rho^{\frac{7}{3} -\frac{\eta}{3}} + \delta \langle\xi_1, \mathbb{H}_0 \xi_1\rangle,
\]
for some $0<\delta<1$.
Thus, an optimization over $\eta$ ($\eta = 2/5$) yields
\begin{equation}\label{eq: final reg T2 lw bd}
  |\langle \xi_\lambda, (\mathbb{T}_2 - \mathbb{T}_2^r) \xi_\lambda \rangle |  \leq CL^3\rho^{2+ \frac{1}{5}}  +  \delta\langle \xi_1, \mathbb{H}_0 \xi_1\rangle,
\end{equation}
for some $0<\delta<1$. Similary, one can also prove that 
\begin{equation}\label{eq: concl Q4 reg lw bd}
   |\langle \xi_\lambda, (\mathbb{Q}_4 - \mathbb{Q}_4^r) \xi_\lambda \rangle |  \leq CL^3\rho^{2+\frac{\epsilon}{3}} + CL^3\rho^{2+ \frac{1}{5}}  +  \delta\langle \xi_1, \mathbb{H}_0 \xi_1\rangle + \delta \langle \xi_\lambda,\widetilde{\mathbb{Q}}_1 \xi_\lambda\rangle.
\end{equation}

To conclude, we can bound, for some $0<\delta<1$, the error in Proposition \ref{pro: scatt canc} due to the scattering equation cancellation as 
\begin{equation}\label{eq: concl scatt canc lw bd}
  |\langle\xi_\lambda, (\mathbb{T}_1 + \mathbb{T}_2^r + \widetilde{\mathbb{Q}}_4^r)\xi_\lambda \rangle|\leq CL^{\frac{3}{2}}\rho^{\frac{3}{2}}\|\mathcal{N}_>^{\frac{1}{2}}\xi_\lambda\| \leq CL^3\rho^{\frac{7}{3}} +\delta \langle \xi_1,\mathbb{H}_0 \xi_1 \rangle,
\end{equation}
where we used again \eqref{eq: improv N lw bd}.

\subsection{Estimate for $\widetilde{\mathbb{Q}}_3$}\label{sec: Q3}
In this section, we estimate $\langle R^\ast \psi, \widetilde{\mathbb{Q}}_3 R^\ast \psi\rangle$. Note that differently than in \cite{FGHP} where $\langle R^\ast \psi, \widetilde{\mathbb{Q}}_3 R^\ast\rangle$ was estimated via a priori bounds (see \cite[Proposition 3.1]{FGHP}, here a refined analysis is required.
We start by recalling that 
\begin{eqnarray}\label{eq: def Q3}
  \widetilde{\mathbb{Q}}_3 &=& -\sum_{\sigma\neq\sigma^\prime} \int\, dxdy\, V(x-y)\left(a^\ast_\sigma(u_x) a^\ast_{\sigma^\prime}(u_y) a^\ast_{\sigma}(\overline{v}_x)a_{\sigma^\prime}(u_y) - a^\ast_\sigma(u_x) a^\ast_{\sigma^\prime}(\overline{v}_y) a^\ast_{\sigma}(\overline{v}_x)a_{\sigma^\prime}(\overline{v}_y)\right) + \mathrm{h.c.} \nonumber
  \\
  &\equiv& \widetilde{\mathbb{Q}}_{3;1} + \widetilde{\mathbb{Q}}_{3;2},\nonumber
\end{eqnarray}
where 
\begin{equation}
  \widetilde{\mathbb{Q}}_{3;1} = -\sum_{\sigma\neq\sigma^\prime} \int\, dxdy\, V(x-y) a^\ast_\sigma(u_x) a^\ast_{\sigma^\prime}(u_y) a^\ast_{\sigma}(\overline{v}_x)a_{\sigma^\prime}(u_y) + \mathrm{h.c.},
\end{equation}
and
\begin{equation}
  \widetilde{\mathbb{Q}}_{3;2} = \sum_{\sigma\neq\sigma^\prime} \int\, dxdy\, V(x-y) a^\ast_\sigma(u_x) a^\ast_{\sigma^\prime}(\overline{v}_y) a^\ast_{\sigma}(\overline{v}_x)a_{\sigma^\prime}(\overline{v}_y) + \mathrm{h.c.}
\end{equation}
Before proceeding further, we split $\widetilde{\mathbb{Q}}_{3;1}$ it in two parts, according to the support of the momenta in $a_{\sigma^\prime}(u_y)$. In particular, we write $a_{\sigma^\prime}(u_y) = a_{\sigma^\prime}(u^<_y) + a_{\sigma^\prime}(u_y^>)$, where $\hat{u}^<_\sigma(k)$ and $\hat{u}^>_\sigma(k)$ are such that 
\[
  \hat{u}^<_\sigma(k) = \begin{cases} 1, &k\notin\mathcal{B}_F^\sigma, \, |k| < 2k_F^\sigma, \\ 0, &|k| \geq 2k_F^\sigma. \end{cases}, \qquad \hat{u}^>_\sigma(k) = \begin{cases} 0, &k\notin\mathcal{B}_F^\sigma, \, |k| < 2k_F^\sigma, \\ 1, &|k| \geq 2k_F^\sigma.\end{cases}
\]

Correspondingly, we write $\widetilde{\mathbb{Q}}_{3;1} = \widetilde{\mathbb{Q}}_{3;1}^< + \widetilde{\mathbb{Q}}_{3;1}^>$, i.e., 
\begin{equation}\label{eq: Q 31 <}
  \widetilde{\mathbb{Q}}_{3;1}^< := \sum_{\sigma\neq\sigma^\prime}\int\, dxdy\, V(x-y) \langle R^\ast \psi, a^\ast_\sigma(u_x) a^\ast_{\sigma^\prime}(u_y) a^\ast_\sigma(\overline{v}_x) a_{\sigma^\prime}(u_y^<)R^\ast\psi\rangle, 
\end{equation}

\begin{equation}\label{eq: Q 31 >}
  \widetilde{\mathbb{Q}}_{3;1}^> := \sum_{\sigma\neq\sigma^\prime}\int\, dxdy\, V(x-y) \langle R^\ast \psi, a^\ast_\sigma(u_x) a^\ast_{\sigma^\prime}(u_y) a^\ast_\sigma(\overline{v}_x) a_{\sigma^\prime}(u_y^>)R^\ast\psi\rangle.
\end{equation}
\begin{proposition}[A first bound for $\widetilde{\mathbb{Q}}_3$]\label{pro: Q3 first}
Let $\psi$ be an approximate ground state. Under the same assumptions as in Theorem \ref{thm: lower bound}, it holds
\begin{equation}
  \langle R^\ast\psi, \widetilde{\mathbb{Q}}_3 R^\ast \psi\rangle = \langle R^\ast\psi, \widetilde{\mathbb{Q}}_{3;1}^< R^\ast\psi\rangle + \mathcal{E}_{1;\widetilde{\mathbb{Q}}_3}(\psi),
\end{equation}
where $\widetilde{\mathbb{Q}}_{3;1}^<$ is as in \eqref{eq: Q 31 <} and
\begin{equation}
  |\mathcal{E}_{1;\widetilde{\mathbb{Q}}_3}(\psi)| \leq CL^3 \rho^{\frac{7}{3}} + \delta\langle T^\ast R^\ast\psi, \mathbb{H}_0 T^\ast R^\ast\psi \rangle,
\end{equation}
for some $0< \delta <1$.
\end{proposition}
\begin{proof}

We estimate $\widetilde{\mathbb{Q}}_{3;1}^>$ and $\widetilde{\mathbb{Q}}_{3;2}$ separately. We first take into account $\widetilde{\mathbb{Q}}_{3;2}$. Using that $\|a^\ast_{\sigma^\prime}(\overline{v}_y)a_{\sigma}^\ast(\overline{v}_x)\| \leq C\rho$, we can bound 
\begin{equation}
  |\langle R^\ast \psi, \widetilde{\mathbb{Q}}_{3;2} R^\ast\psi\rangle | \leq C \rho \|V\|_1 \langle R^\ast\psi, \mathcal{N} R^\ast\psi\rangle \leq CL^3\rho^{\frac{7}{3}} + \delta\langle \xi_1, \mathbb{H}_0 \xi_1 \rangle,
\end{equation}
where we used the estimate in \eqref{eq: improv N lw bd}. 
We now take into account $\widetilde{\mathbb{Q}}_{3;1}^>$. We have
\begin{eqnarray}
  |\langle R^\ast \psi, \widetilde{\mathbb{Q}}_{3;1}^> R^\ast\psi \rangle| &\leq& \sum_{\sigma\neq\sigma^\prime}\int\, dxdy\, V(x-y)|\langle R^\ast \psi, a_{\sigma}^\ast(u_x)a^\ast_{\sigma^\prime}(u_y)a^\ast_{\sigma}(\overline{v}_x)a_{\sigma^\prime}(u_y^>)R^\ast\psi\rangle|
\\
  &\leq& C\rho^{\frac{1}{2}} \|\widetilde{\mathbb{Q}}_{1}^{\frac{1}{2}}R^\ast\psi\| \|V\|_1^{\frac{1}{2}}\|\mathcal{N}_>^{\frac{1}{2}}R^\ast \psi\|.\nonumber
\end{eqnarray}
Using now \eqref{eq: est N>}, i.e., 
\[
  \langle R^\ast\psi, \mathcal{N}_> R^\ast \psi\rangle \leq C\rho^{-\frac{2}{3}}\langle \xi_1,\mathbb{H}_0 \xi_1\rangle + CL^3\rho^{\frac{5}{3}}, 
\]
and the a priori estimate for $\langle R^\ast \psi, \widetilde{\mathbb{Q}}_1R^\ast\psi\rangle$ stated in Lemma \ref{lem: a priori}, we can bound 
\begin{eqnarray}
|\langle R^\ast \psi, \widetilde{\mathbb{Q}}_{3;1}^> R^\ast\psi \rangle| &\leq& CL^{\frac{3}{2}}\rho^{\frac{4}{3}}\|\widetilde{\mathbb{Q}}_1^{\frac{1}{2}}R^\ast\psi\| + C\rho^{\frac{1}{2} -\frac{1}{3}}\|\widetilde{\mathbb{Q}}_1^{\frac{1}{2}} R^\ast\psi\|\|\mathbb{H}_0^{\frac{1}{2}}T^\ast R^\ast\psi\|,\nonumber
\\
&\leq& CL^3\rho^{\frac{7}{3}} +\delta\langle T^\ast R^\ast\psi, \mathbb{H}_0 T^\ast R^\ast\psi \rangle.
\end{eqnarray}
This concludes the proof.
\end{proof}
\begin{proposition}[Propagation of $\widetilde{\mathbb{Q}}_{3;1}^<$] \label{pro: Q3<} Let $\psi$ be an approximate ground state. Let $0 <\epsilon \leq 1$. Under the same assumptions as in Theorem \ref{thm: lower bound}, it holds
\begin{equation}\label{eq: est Q31<}
  |\langle R^\ast \psi, \widetilde{\mathbb{Q}}_{3;1}^< R^\ast \psi\rangle | \leq CL^3\rho^{\frac{7}{3}} + CL^3\rho^{\frac{7}{3} + \frac{1}{3} -\frac{\epsilon}{3}} + \delta\langle T^\ast R^\ast \psi, (\mathbb{H}_0 + \widetilde{\mathbb{Q}}_{1}) T^\ast R^\ast \psi \rangle,
\end{equation}
for some $0<\delta<1$.
\end{proposition}
\begin{proof}
  To prove the estimate in \eqref{eq: est Q31<} we need to use the almost bosonic Bogoliubov transformation introduced in Definition \ref{def: transf T}. We have 
  \begin{eqnarray}
    \langle R^\ast\psi, \widetilde{\mathbb{Q}}_{3;1}^< R^\ast \psi\rangle &=& \langle T^\ast R^\ast\psi, \widetilde{\mathbb{Q}}_{3;1}^< T^\ast R^\ast\psi\rangle - \int_0^1\, d\lambda\, \frac{d}{d\lambda}\langle \xi_\lambda, \widetilde{\mathbb{Q}}_{3;1}^< \xi_\lambda\rangle
    \\
    &=&\langle T^\ast R^\ast\psi, \widetilde{\mathbb{Q}}_{3;1}^< T^\ast R^\ast\psi\rangle + \int_0^1\, d\lambda\,\langle \xi_\lambda, [\widetilde{\mathbb{Q}}_{3;1}^<, B - B^\ast] \xi_\lambda\rangle,\nonumber
  \end{eqnarray}
  where we recall the notation $\xi_\lambda = T^\ast_\lambda R^\ast \psi$. We now estimate the first contribution in the right side above. We have, for some $0<\delta<1$
  \begin{eqnarray}
    \langle T^\ast R^\ast\psi, \widetilde{\mathbb{Q}}_{3;1}^< T^\ast R^\ast\psi\rangle &\leq& C\rho^{\frac{1}{2}}\sum_{\sigma\neq\sigma^\prime}\int\, dxdy\, V(x-y)\|a_\sigma(u_x)a_{\sigma^\prime}(u_y)T^\ast R^\ast \psi\|\|a_{\sigma^\prime}(u^<_y) T^\ast R^\ast \psi\| \nonumber
    \\
    &\leq& C\rho^{\frac{1}{2}}\|\widetilde{\mathbb{Q}}_{1}^{\frac{1}{2}}T^\ast R^\ast \psi\|\|\mathcal{N}^{\frac{1}{2}}T^\ast R^\ast \psi\| \nonumber
    \\
    &\leq& \delta \langle T^\ast R^\ast \psi, \mathbb{Q}_{1} T^\ast R^\ast \psi\rangle + C\rho\langle T^\ast R^\ast \psi, \mathcal{N} T^\ast R^\ast \psi \rangle \nonumber
    \\
    &\leq& \delta\langle T^\ast R^\ast \psi, (\mathbb{H}_0 + \widetilde{\mathbb{Q}}_{1}) T^\ast R^\ast \psi \rangle + CL^3\rho^{\frac{7}{3}},\nonumber
  \end{eqnarray}
  where in the last inequality we used the estimate in \eqref{eq: est N}. We now take into account the commutator. We have 
  \begin{multline}\label{eq: two parts comm Q3}
    [\widetilde{\mathbb{Q}}_{3;1}^<, B-B^\ast]
    \\
         = \int\, dxdydzdz^\prime\, V(x-y)\varphi(z-z^\prime) [a^\ast_\sigma(u_x) a^\ast_{\sigma^\prime}(u_y) a^\ast_\sigma(\overline{v}_x) a_{\sigma^\prime}(u_y^<), a_\uparrow(u^r_z)a_\uparrow(\overline{v}^r_z)a_\downarrow(u^r_{z^\prime})a_\downarrow(\overline{v}^r_{z^\prime})]
    \\
     -\int\, dxdydzdz^\prime\, V(x-y)\varphi(z-z^\prime) [a^\ast_\sigma(u_x) a^\ast_{\sigma^\prime}(u_y) a^\ast_\sigma(\overline{v}_x) a_{\sigma^\prime}(u_y^<), a^\ast_\downarrow(\overline{v}^r_{z^\prime}) a_\downarrow^\ast(u^r_{z^\prime}) a^\ast_\uparrow(\overline{v}^r_z)a_\uparrow^\ast(u^r_z)].
  \end{multline} 
  We first take into account the second commutator in the right side above:
\begin{multline*}
  [a^\ast_\sigma(u_x) a^\ast_{\sigma^\prime}(u_y) a^\ast_\sigma(\overline{v}_x) a_{\sigma^\prime}(u_y^<), a^\ast_\downarrow(\overline{v}^r_{z^\prime}) a_\downarrow^\ast(u^r_{z^\prime}) a^\ast_\uparrow(\overline{v}^r_z)a_\uparrow^\ast(u^r_z)] 
  \\
  = -a^\ast_\sigma(u_x) a^\ast_{\sigma^\prime}(u_y) a^\ast_\sigma(\overline{v}_x)a^\ast_\downarrow(\overline{v}^r_{z^\prime}) a^\ast_\uparrow(\overline{v}^r_z)  [a_{\sigma^\prime}(u_y^<), a_\downarrow^\ast(u^r_{z^\prime}) a_\uparrow^\ast(u^r_z)].
\end{multline*}
This commutator is vanishing since:
\begin{equation}
    [a_{\sigma^\prime}(u_y^<), a_\downarrow^\ast(u^r_{z^\prime}) a_\uparrow^\ast(u^r_z)] =  0,
\end{equation}
which is due to the fact that 
\[
\hat{u}^<_\sigma(k) = \begin{cases} 1, &k\notin\mathcal{B}_F^\sigma, \, |k| < 2k_F^\sigma, \\ 0, &|k| \geq 2k_F^\sigma. \end{cases}, \qquad \hat{u}^{r}_{\sigma}(k) = \begin{cases} 0, &\mbox{for}\quad |k| \leq 2k_F^\sigma, \\ 
    1, &\mbox{for}\quad  3k_F^\sigma \leq |k| \leq \rho^{-\beta}_\sigma,
    \\
    0,  &\mbox{for}\quad |k|\geq 2\rho_\sigma^{-\beta},\end{cases}
\] 
All the error terms are then coming form the first line in the right side in \eqref{eq: two parts comm Q3}. We then calculate:
  \begin{multline}\label{eq: 1 part comm Q3}
    [a^\ast_\sigma(u_x) a^\ast_{\sigma^\prime}(u_y) a^\ast_\sigma(\overline{v}_x) a_{\sigma^\prime}(u_y^<), a_\uparrow(u^r_z)a_\uparrow(\overline{v}^r_z)a_\downarrow(u^r_{z^\prime})a_\downarrow(\overline{v}^r_{z^\prime})] 
    \\
    =  -[a^\ast_\sigma(u_x) a^\ast_{\sigma^\prime}(u_y) a^\ast_\sigma(\overline{v}_x), a_\uparrow(u^r_z)a_\downarrow(u^r_{z^\prime})a_\uparrow(\overline{v}^r_z)a_\downarrow(\overline{v}^r_{z^\prime})]a_{\sigma^\prime}(u_y^<)
    \\
    = -[a^\ast_\sigma(u_x) a^\ast_{\sigma^\prime}(u_y) , a_\uparrow(u^r_z)a_\downarrow(u^r_{z^\prime})]a_\uparrow(\overline{v}^r_z)a_\downarrow(\overline{v}^r_{z^\prime})a^\ast_\sigma(\overline{v}_x) a_{\sigma^\prime}(u_y^<)
    \\
    - a^\ast_\sigma(u_x) a^\ast_{\sigma^\prime}(u_y)a_\uparrow(u^r_z)a_\downarrow(u^r_{z^\prime})[a_\sigma^\ast(\overline{v}_x), a_\uparrow(\overline{v}^r_z)a_\downarrow(\overline{v}^r_{z^\prime})]a_{\sigma^\prime}(u_y^<).
  \end{multline}
  We now calculate, for $\sigma \neq \sigma^\prime$,
  \begin{eqnarray}\label{eq: Q3 comm calcul}
    [a^\ast_\sigma(u_x) a^\ast_{\sigma^\prime}(u_y) , a_\uparrow(u^r_z)a_\downarrow(u^r_{z^\prime})]
    &=& u^r_\uparrow(z;x)u^r_\downarrow(z^\prime;y) - u^r_\downarrow(z^\prime, x) u^r_\uparrow(z;y) \nonumber
    \\
    && - u^r_\uparrow(z;x)a^\ast_\downarrow(u_y)a_\downarrow(u^r_{z^\prime}) + u^r_\downarrow(z^\prime;x) a^\ast_\uparrow(u_y)a_\uparrow(u^r_z) \nonumber
    \\
    && + u^r_\uparrow(z;y) a_\downarrow^\ast(u_x)a_\downarrow(u^r_{z^\prime}) - u^r_\downarrow(z^\prime;y) a_\uparrow^\ast(u_x)a_\uparrow(u^r_{z}).
  \end{eqnarray}
  The commutator above gives rise to two different error terms. The first one is, for instance,  
\begin{eqnarray}
  \mathrm{I} &=& \int\, dxdydzdz^\prime\, V(x-y)\varphi(z - z^\prime) u^r_\uparrow(z;x) \langle \xi_\lambda, a^\ast_\downarrow(u_y) a_\downarrow(u^r_{z^\prime})a_\uparrow(\overline{v}^r_z)a_\downarrow(\overline{v}^r_{z^\prime})a^\ast_\uparrow(\overline{v}^r_x)a_\downarrow(u^<_y)\xi_\lambda\rangle\nonumber
  \\
  &=& -\int\, dxdydz\, V(x-y)u^r_\uparrow(z;x)\langle \xi_\lambda, a^\ast_\downarrow(u_y)b_\downarrow(\varphi_{z})a_\uparrow(\overline{v}^r_z)a^\ast_\uparrow(\overline{v}^r_x)a_\downarrow(u^<_y)\xi_\lambda\rangle.
\end{eqnarray}
We then get, using Proposition \ref{pro: L1, L2 v} (to bound $\|u^r_{x,\sigma}\|_1 \leq C$, $\|\overline{v}^r_{x,\sigma}\|_2 \leq C\rho^{1/2}$) and Lemma \ref{lem: bound b phi} (to estimate $\|b_\sigma(\varphi_{z})\|\leq C\rho^{1/3}$), that
\begin{eqnarray}
  |\mathrm{I}| &\leq& \int\, dxdydz\, |V(x-y)||u^r_\uparrow(z;x)|\|b_\downarrow(\varphi_{z})\|\|a_\uparrow(\overline{v}^r_z)\|\|a^\ast_\uparrow(\overline{v}^r_x)\|a_\downarrow(u_y)\xi_\lambda\|\|a_\downarrow(u^<_y)\xi_\lambda\| \nonumber\\
  &\leq&  C\rho^{1+\frac{1}{3}}\langle \xi_\lambda, \mathcal{N}\xi_\lambda\rangle \leq CL^3\rho^3 + \delta\langle \xi_1, \mathbb{H}_0 \xi_1 \rangle,
\end{eqnarray}
where we used also the estimate in \eqref{eq: improv N lw bd}. The other type of error is coming from the first two terms in the right hand side of \eqref{eq: Q3 comm calcul}. They can be estimated in the same way, we take into account: 
  \begin{eqnarray}
    \mathrm{II} &=& -\int\, dxdydzdz^\prime\, V(x-y)\varphi(z-z^\prime) u^r_\uparrow(z;x)u^r_\downarrow(z^\prime;y) \langle \xi_\lambda, a_\uparrow(\overline{v}^r_z) a_\downarrow(\overline{v}^r_{z^\prime}) a^\ast_\uparrow(\overline{v}_x)a_\downarrow(u^<_y)\xi_\lambda\rangle
    \\
    &=& + \int\, dxdydzdz^\prime\, V(x-y)\varphi(z-z^\prime) u^r_\uparrow(z;x)u^r_\downarrow(z^\prime;y)\langle \xi_\lambda, a^\ast_\uparrow(\overline{v}_x)a_\uparrow(\overline{v}^r_z)a_\downarrow(\overline{v}^r_{z^\prime})a_\downarrow(u^<_y)\xi_\lambda\rangle\nonumber
    \\
    && -\int\, dxdydzdz^\prime\, V(x-y)\varphi(z-z^\prime) u^r_\uparrow(z;x)u^r_\downarrow(z^\prime;y)\widetilde{\omega}^r_\uparrow(x;z) \langle \xi_\lambda, a_\downarrow(\overline{v}^r_{z^\prime}) a_\downarrow(u^<_y)\xi_\lambda\rangle\equiv \mathrm{II}_{\mathrm{A}} + \mathrm{II}_{\mathrm{B}}\nonumber,
  \end{eqnarray}
  where $\widetilde{\omega}^r$ is defined as in Remark \ref{rem: u2 and tilde omega}.
  We now take into account $\mathrm{II}_{\mathrm{A}}$, we have
  \begin{eqnarray}
    |\mathrm{II}_{\mathrm{A}}| &\leq& C\rho \int\, dxdydzdz^\prime\, V(x-y)|\varphi(z-z^\prime)||u^r_\uparrow(z;x)||u^r_\downarrow(z^\prime;y)|\|\|a_\uparrow(\overline{v}_x)\xi_\lambda\|\|a_\downarrow(u^<_y)\xi_\lambda\|
    \\
    &\leq& C \rho  \langle \xi_\lambda, \mathcal{N}\xi_\lambda\rangle \leq CL^3\rho^{\frac{7}{3}} + \delta\langle \xi_1, \mathbb{H}_0 \xi_1\rangle, \nonumber
  \end{eqnarray}
where we used the estimate in\eqref{eq: improv N lw bd} together with the fact that $\|u^r_{x,\sigma}\|_1 , \|\varphi\|_\infty \leq C$ (see Proposition \ref{pro: L1, L2 v} and Lemma \ref{lem: bound phi}). We now take into account the term $\mathrm{II}_{\mathrm{B}}$. In what follows, we prove that $\mathrm{II}_{\mathrm{B}}$ vanishes due to the momentum conservation. We have 
\begin{equation}
  \mathrm{II}_{\mathrm{B}} = \int\, dxdydzdz^\prime\, V(x-y)\varphi(z-z^\prime) u^r_\uparrow(z;x)u^r_\downarrow(z^\prime;y)\widetilde{\omega}^r_\uparrow(x;z) \langle \xi_\lambda, a_\downarrow(\overline{v}^r_{z^\prime}) a_\downarrow(u^<_y)\xi_\lambda\rangle.
\end{equation}
We can indeed rewrite it in momentum space and we get that 
\begin{eqnarray}
  \mathrm{II}_{\mathrm{B}} \hspace{-0.3cm}&=&\hspace{-0.3cm} \int\, dxdydzdz^\prime\, \frac{1}{(L^3)^6}\sum_{p,q}\hat{V}(p)e^{ip\cdot (x-y)}\hat{\varphi}(q) e^{iq\cdot(z-z^\prime)}\nonumber
  \\
  &&\cdot\sum_{k_1, k_2, k_3} \hat{u}^r_\uparrow(k_1) \hat{u}_\downarrow(k_2)\hat{\widetilde{\omega}}^r(k_3)e^{i(k_1-k_3)\cdot(z-x)}e^{ik_2\cdot(z^\prime - y)}\cdot\nonumber
  \\
  &&\cdot \int\, dtdt^\prime\, \sum_{k_4, k_5} \hat{v}^r(k_4) \hat{u}^<(k_5) e^{-ik_4\cdot(z^\prime + t)} e^{-ik_5\cdot(t^\prime -y)} \, \langle\xi_\lambda, a_{t,\downarrow}a_{t^\prime; \downarrow}\xi_\lambda\rangle.
\end{eqnarray}
Taking the integrals in $x,y,z,z^\prime$, imply that  $k_4 = k_5$. This is impossible since $k_4\in\mathcal{B}_F^\downarrow$ and $k_5\notin\mathcal{B}_F^\downarrow$. 
We now take into account the error terms coming from the second line in the right hand side of \eqref{eq: 1 part comm Q3}. To begin we calculate: 
\begin{equation}
  [a_\sigma^\ast(\overline{v}_x) , a_\uparrow(\overline{v}^r_z)a_\downarrow(\overline{v}^r_{z^\prime})] = \widetilde{\omega}_\uparrow(x;z) a_\downarrow(\overline{v}^r_{z^\prime}) - \widetilde{\omega}^r_\downarrow(x;z^\prime)a_\uparrow(\overline{v}^r_z).
\end{equation}
The corresponding error term is then of the following form: 
\begin{equation}
  \mathrm{III} = -\int\, dxdydzdz^\prime\, V(x-y)\varphi(z-z^\prime)\widetilde{\omega}^r_\uparrow(z;x)\langle \xi_\lambda, a^\ast_\uparrow(u_x)a^\ast_\downarrow(u_y) a_\uparrow(u^r_z)a_\downarrow(u^r_{z^\prime}) a_\downarrow(\overline{v}^r_{z^\prime}) a_\downarrow(u^<_y)\xi_\lambda\rangle.
\end{equation}
Using that creation/annihilation operators with different spin commute, we can write 
\begin{eqnarray}
  |\mathrm{III}| &\leq& \int\, dxdydz\, V(x-y)|\widetilde{\omega}^r(z;x)| \|a_\uparrow(u_x)a_\downarrow(u_y)\xi_\lambda\|\|b_\downarrow(\varphi_z)\|\|a_\downarrow(u^<_y)a_\uparrow(u^r_z)\xi_\lambda\|
  \\
  &\leq &C\rho^{\frac{1}{2} + \frac{1}{3}}\int\, dxdydz\, V(x-y)|\widetilde{\omega}^r(z;x)|\|a_\uparrow(u_x)a_\downarrow(u_y)\xi_\lambda\|\|a_\uparrow(u^r_z)\xi_\lambda\|
  \\
  &\leq& C\rho^{\frac{1}{2} + \frac{1}{3}-\frac{\epsilon}{3}} \|\widetilde{\mathbb{Q}}_{1}^{\frac{1}{2}}\xi_\lambda\|\|\mathcal{N}_>^{\frac{1}{2}}\xi_\lambda\| \leq CL^{\frac{3}{2}}\rho^{\frac{3}{2} + \frac{1}{3} - \frac{\epsilon}{3}}\|\mathcal{N}_>^{\frac{1}{2}}\xi_\lambda\|,
\end{eqnarray}
where we used that $\|a_\downarrow(u^<_y)\|\leq \|u^<_y\|_2\leq C\rho^{1/2}$, together with $\|\widetilde{\omega}^{r}_{x,\sigma}\|_1 \leq C\rho^{-\epsilon/3}$ (see Remark \ref{rem: u2 and tilde omega}), and the propagation estimates for $\widetilde{\mathbb{Q}}_{1}$. Using now Proposition \ref{lem: bound N>}, we find 
\begin{equation}
  |\mathrm{III}| \leq CL^3\rho^{\frac{3}{2} + \frac{5}{6} + \frac{1}{3} - \frac{\epsilon}{3}} + CL^{\frac{3}{2}}\rho^{\frac{3}{2} + \frac{1}{3} - \frac{\epsilon}{3} - \frac{1}{3}}\|\mathbb{H}_0^{\frac{1}{2}}\xi_1\| \leq CL^3\rho^{\frac{7}{3} + \frac{1}{3} - \frac{\epsilon}{3}} + CL^3\rho^{3-\frac{2}{3}\epsilon} + \delta\langle \xi_1, \mathbb{H}_0 \xi_1 \rangle.
\end{equation}
Putting all the estimates together, we then proved that for some $0<\delta<1$ and $0< \epsilon \leq 1$, we have
\begin{equation}
  |\langle R^\ast\psi, \widetilde{\mathbb{Q}}_{3;1}^< R^\ast \psi\rangle | \leq CL^3\rho^{\frac{7}{3}}+ CL^3\rho^{\frac{7}{3} + \frac{1}{3} -\frac{\epsilon}{3}} + \delta\langle T^\ast R^\ast \psi, (\mathbb{H}_0 + \widetilde{\mathbb{Q}}_{1})T^\ast R^\ast\psi\rangle.
\end{equation}
\end{proof}
\begin{proposition}[Final estimate $\widetilde{\mathbb{Q}}_3$] \label{pro: Q3 final} Let $0<\delta<1$, $0 < \epsilon \leq 1$, and let $\psi$ be an approximate ground state, in the sense of Definition \ref{def: approx gs}. Let $\widetilde{\mathbb{Q}}_3$ be as in \eqref{eq: def Q3}. Under the same assumptions of Theorem \ref{thm: lower bound}, it holds 
\begin{equation}
  |\langle R^\ast\psi, \widetilde{\mathbb{Q}}_3 R^\ast\psi\rangle | \leq CL^3\rho^{\frac{7}{3}}+ CL^3\rho^{\frac{7}{3} + \frac{1}{3} -\frac{\epsilon}{3}} + \delta\langle T^\ast R^\ast \psi, (\mathbb{H}_0 + \widetilde{\mathbb{Q}}_{1}) T^\ast R^\ast \psi\rangle.
\end{equation} 
\begin{proof}
  The proof directly follows combing Proposition \ref{pro: Q3 first} together with Proposition \ref{pro: Q3<}.
\end{proof}
\end{proposition}
\subsection{Proof of Theorem \ref{thm: lower bound}}
To prove Theorem \ref{thm: lower bound}, we follow the same strategy of the proof of Theorem \ref{thm: optimal up bd}. We summarize the main steps. In the following $\psi$ denotes an approximate ground state, in the sense of Definition \ref{def: approx gs}. Recall that in our notations $\xi_\lambda = T^\ast_\lambda R^\ast \psi$.
Similarly as in the upper bound, we start by using Proposition \ref{pro: X,Q2} to write
\begin{equation}
  E_L(N_\uparrow, N_\downarrow)\geq E_{\mathrm{HF}}(\omega) + \langle \xi_0, (\mathbb{H}_0 + \widetilde{\mathbb{Q}}_1 + \widetilde{\mathbb{Q}}_4)\xi_0\rangle + \langle \xi_0, \widetilde{\mathbb{Q}}_3 \xi_0\rangle + \mathcal{E}_1(\psi),
\end{equation}
with 
\begin{equation}
  |\mathcal{E}_1(\psi)| \leq C\rho\langle \xi_0, \mathcal{N}\xi_0\rangle. 
\end{equation}
 Being $\xi_0 = R^\ast\psi$ and using the bound in \eqref{eq: improv N lw bd}, we get, for some $0<\delta <1$, 
\begin{equation}
  |\mathcal{E}_1(\psi)| \leq CL^3\rho^{\frac{7}{3}} + \delta\langle \xi_1, \mathbb{H}_0 \xi_1 \rangle. 
\end{equation}
Moreover, to bound $\langle \xi_0, \widetilde{\mathbb{Q}}_3 \xi_0 \rangle = \langle R^\ast\psi, \widetilde{\mathbb{Q}}_3 R^\ast \psi \rangle$, we can use Proposition \ref{pro: Q3 final}:
\begin{equation}
  |\langle R^\ast\psi, \widetilde{\mathbb{Q}}_3 R^\ast \psi\rangle | \leq CL^3\rho^{\frac{7}{3}} + CL^3\rho^{\frac{7}{3} +\frac{1}{3} -\frac{\epsilon}{3}} + \delta\langle \xi_1, (\mathbb{H}_0 + \widetilde{\mathbb{Q}}_1) \xi_1\rangle.
\end{equation}
As in the upper bound, we then extract the constant term from $\langle \xi_0, (\mathbb{H}_0 + \widetilde{\mathbb{Q}}_1 + \widetilde{\mathbb{Q}}_4)\xi_0\rangle$. More precisely, proceeding similarly as in \eqref{eq: 1st step duhamel}, we have 
\begin{eqnarray}
\langle \xi_0, (\mathbb{H}_0 + \widetilde{\mathbb{Q}}_1) \xi_0 \rangle &=& \langle\xi_1, (\mathbb{H}_0 + \widetilde{\mathbb{Q}}_1)\xi_1\rangle - \int_0^1 d\lambda\, \frac{d}{d\lambda}\langle \xi_\lambda, (\mathbb{H}_0 + \widetilde{\mathbb{Q}}_1)\xi_\lambda\rangle 
\\
&=& \langle\xi_1, (\mathbb{H}_0 + \widetilde{\mathbb{Q}}_1)\xi_1\rangle + \int_0^1 d\lambda\,\langle \xi_\lambda, (\mathbb{T}_1 + \mathbb{T}_2)\xi_\lambda\rangle + \mathcal{E}_2, \nonumber
\end{eqnarray}
where $\mathcal{E}_2$ can be bounded using \eqref{eq: conclusive H0 lw bd}, Proposition \ref{pro: Q1} and the propagation of the estimates discussed in Proposition \ref{pro: propagation est}, as follows:
\begin{equation}
  |\mathcal{E}_2| \leq CL^3\rho^{\frac{7}{3}} + CL^3\rho^{\frac{7}{3} +\frac{1}{3} -\frac{\epsilon}{3}} + \delta\langle\xi_1, \mathbb{H}_0 \xi_1\rangle,
\end{equation}
for some $0<\delta<1$. Now we write
\[
  \langle \xi_\lambda (\mathbb{T}_1 + \mathbb{T}_2)\xi_\lambda \rangle = \langle \xi_\lambda (\mathbb{T}_1 + \mathbb{T}_2^r)\xi_\lambda \rangle + \mathcal{E}_3,
\]
where $\mathcal{E}_3$ can be bounded as in \eqref{eq: final reg T2 lw bd}, i.e., 
\[
  |\mathcal{E}_3| = |\langle \xi_\lambda, (\mathbb{T}_2 - \mathbb{T}_2^r)\xi_\lambda \rangle| \leq CL^3\rho^{2+\frac{1}{5}} + \delta\langle \xi_1, \mathbb{H}_0 \xi_1\rangle.
\]
Thus, we get 
\begin{eqnarray}
  \langle \xi_0, (\mathbb{H}_0 + \widetilde{\mathbb{Q}}_1)\xi_0\rangle  = \langle\xi_1, (\mathbb{H}_0 + \widetilde{\mathbb{Q}}_1)\xi_1\rangle - \int_0^1 d\lambda \langle \xi_\lambda, \widetilde{\mathbb{Q}}_4^r \xi_\lambda \rangle + \mathcal{E}_2 +\mathcal{E}_3 + \mathcal{E}_4,
\end{eqnarray}
with 
\[
  |\mathcal{E}_4| = |\langle \xi_\lambda, (\mathbb{T}_1 + \mathbb{T}_2^r + \widetilde{\mathbb{Q}}_4^r)\xi_\lambda| \leq CL^3\rho^{\frac{7}{3}} + \delta\langle \xi_1, \mathbb{H}_0 \xi_1\rangle, 
\]
where we used \eqref{eq: concl scatt canc lw bd}.
All together we find that 
\begin{eqnarray}
  \langle \xi_0, (\mathbb{H}_0 + \widetilde{\mathbb{Q}}_1 + \widetilde{\mathbb{Q}}_4) \xi_0 \rangle\hspace{-0.2cm} &=& \hspace{-0.2cm}\langle\xi_1, (\mathbb{H}_0 + \widetilde{\mathbb{Q}}_1)\xi_1\rangle + \langle \xi_0, \widetilde{\mathbb{Q}}_4 \xi_0 \rangle -  \int_0^1 d\lambda \langle \xi_\lambda, \widetilde{\mathbb{Q}}_4^r \xi_\lambda \rangle + \mathcal{E}_2 + \mathcal{E}_3 + \mathcal{E}_{4}\nonumber
  \\
  \hspace{-0.2cm} &=& \hspace{-0.2cm} \langle\xi_1, (\mathbb{H}_0 + \widetilde{\mathbb{Q}}_1)\xi_1\rangle + \langle\xi_1, (\widetilde{\mathbb{Q}}_4 - \widetilde{\mathbb{Q}}_4^r)\xi_1\rangle\nonumber
  \\
  &&\hspace{-0.2cm} -\int_0^1 \frac{d}{d\lambda}\langle \xi_\lambda, \widetilde{\mathbb{Q}}_4\xi_\lambda\rangle + \int_0^1 d\lambda\int_\lambda^1 d\lambda^\prime\frac{d}{d\lambda^\prime}\langle \xi_{\lambda^\prime}, \widetilde{\mathbb{Q}}_4^r\xi_{\lambda^\prime}\rangle + \mathcal{E}_2 + \mathcal{E}_3 + \mathcal{E}_4\nonumber 
  \\
  \hspace{-0.2cm} &=& \hspace{-0.2cm} \langle\xi_1, (\mathbb{H}_0 + \widetilde{\mathbb{Q}}_1)\xi_1\rangle  -\int_0^1 \frac{d}{d\lambda}\langle \xi_\lambda, \widetilde{\mathbb{Q}}_4\xi_\lambda\rangle + \int_0^1 d\lambda\int_\lambda^1 d\lambda^\prime\frac{d}{d\lambda^\prime}\langle \xi_{\lambda^\prime}, \widetilde{\mathbb{Q}}_4^r\xi_{\lambda^\prime}\rangle + \mathcal{E}_{\mathrm{tot}},\nonumber
\end{eqnarray}
where we can estimated $\langle\xi_1, (\widetilde{\mathbb{Q}}_4 - \widetilde{\mathbb{Q}}_4^r)\xi_1\rangle$  as in \eqref{eq: concl Q4 reg lw bd}, therefore $\mathcal{E}_{\mathrm{tot}}$ is such that 
\[
  |\mathcal{E}_{\mathrm{tot}}| \leq CL^3\rho^{2+\frac{1}{5}} + CL^3 \rho^{2+\frac{\epsilon}{3}} + CL^3\rho^{\frac{7}{3} +\frac{1}{3} -\frac{\epsilon}{3}} + \delta\langle \xi_1, (\mathbb{H}_0 + \widetilde{\mathbb{Q}}_1)\xi_1\rangle, 
\]
for some $0<\delta<1$. From the estimate in \eqref{eq: conclusive Q4 lw bd} and  proceeding as in Section \ref{sec: up bd} (see \eqref{eq: Q4 upper bd}), we get that there exists $0<\widetilde{\delta}<1$ such that 
\begin{equation}
  E_{L}(N_\uparrow, N_\downarrow) \geq E_{\mathrm{HF}} + (1-\widetilde{\delta})\langle \xi_1, (\mathbb{H}_0 + \widetilde{\mathbb{Q}}_1)\xi_1\rangle - \rho_\uparrow\rho_\downarrow \int dxdy V(x-y)\varphi(x-y) + \widetilde{\mathcal{E}}_{\mathrm{tot}},
\end{equation}
where, 
\begin{equation}
  |\widetilde{\mathcal{E}}_{\mathrm{tot}}| \leq CL^3 \rho^{2+\frac{1}{5}} + CL^3\rho^{2+\frac{\epsilon}{3}} + CL^3\rho^{\frac{7}{3} +\frac{1}{3} - \frac{\epsilon}{3}} \leq CL^3\rho^{2+\frac{1}{5}}.
\end{equation}

The last inequality above follows after an optimization over $\epsilon = 1$. Proceeding now as in Section \ref{sec: up bd} and using the positivity of $\mathbb{H}_0$, $\widetilde{\mathbb{Q}}_1$, we then get for $L$ large enough,
\begin{equation}
  e_L(\rho_\uparrow, \rho_\downarrow) \geq \frac{3}{5}(6\pi^2)^{\frac{2}{3}}(\rho_\uparrow^{\frac{5}{3}} + \rho_\downarrow^{\frac{5}{3}}) + 8\pi a\rho_\uparrow \rho_\downarrow - C\rho^{2+\frac{1}{5}}.
\end{equation}

%

%
%
%
%
\appendix
\section{Properties of the scattering solution}\label{app: scattering}
Recall that the $\varphi$ appearing in the almost bosonic Bogoliubov transformation $T$ is defined as the periodization of  $\varphi_\infty= \varphi_0 \chi_{\sqrt[3]\rho}$ in the box $\Lambda_L$, where $\varphi_0$ is the solution of the zero energy scattering equation in $\mathbb{R}^3$:
\begin{equation}\label{eq: def phi0}
  2\Delta\varphi_0 + V_\infty(1-\varphi_0) = 0, \qquad \mbox{in}\,\,\,\mathbb{R}^3, \quad \mbox{with}\,\,\, \varphi(x)\rightarrow 0\quad \mbox{as}\,\,\, |x|\rightarrow \infty. 
\end{equation}
The solution of equation \eqref{eq: def phi0} is very well known in the literature, see e.g. \cite{LSSY, NT}. Before discussing the proof of Lemma \ref{lem: bound phi} and Lemma \ref{lem: decay phi}, we recall some useful bounds for $\varphi_0$. 
\begin{lemma}[Bounds for $\varphi_0$]\label{lem: bound phi0}
Let $V_\infty$ as in Assumption \ref{asu: potential V} and such that $\mathrm{supp} V_\infty \subset \{x\in\mathbb{R}^3\, \vert\, |x| < R_0\}$. Let $\varphi_0$ as in \eqref{eq: def phi0}. The following holds
\begin{itemize}
  \item[(i)] For all $x\in \mathbb{R}^3$, there exists a constant $C>0$ such that 
  \begin{equation}\label{eq: bound phi0}
    0\leq \varphi_0(x) \leq \frac{C}{1+|x|}, \qquad \varphi_0(x) = \frac{a}{|x|} \quad \forall x\in \mathbb{R}^3\setminus (\mathrm{supp}\, V_\infty),\qquad 8\pi a = \int_{\mathbb{R}^3} dx\, V_\infty(x)(1-\varphi_0)(x)
  \end{equation}
  where $a$ is the scattering length associated with the potential $V_\infty$, and 
  \begin{equation}\label{eq: gradient of phi0}
    |\nabla\varphi_0(x)|\leq \frac{C}{1 + |x|^2}.
  \end{equation}
\item[(ii)]
   Moreover, there exists $C>0$ such that 
  \begin{equation}\label{eq: L1 norm lapl phi}
   |\Delta\varphi_0(x)|\leq C, \qquad  \|\Delta\varphi_0\|_1 \leq C.
  \end{equation}
\item[(iii)] For all $n\geq 2$, it follows that for all $x\in B_{R}(0)$, with $R > R_0$
\begin{equation}\label{eq: bound n der}
|D^\alpha\varphi_0(x)|\leq C_n, \qquad \mbox{with}\,\,\,\alpha = (\alpha_1, \alpha_2,\alpha_3)\in \mathbb{N}_0^3\,\,\,\mbox{and}\,\,\, |\alpha| = n.
\end{equation}
\end{itemize}
\end{lemma}
\begin{proof}
  For the proof of \eqref{eq: bound phi0} and \eqref{eq: gradient of phi0} we refer to \cite[Section 2.2]{NT}. The bounds in \eqref{eq: L1 norm lapl phi} directly follow from the equation satisfied by $\varphi_0$ together with the regularity of the potential $V_\infty$. Moreover, the estimate for \eqref{eq: bound n der} follows from the bound for the case $n=1,2$ which are well-know together with the fact that $\varphi_0$ solves the scattering solution and that $V_\infty$ is a regular interaction potential.
\end{proof}

We now prove prove Lemma \ref{lem: bound phi}. Recall that 
\[
  \varphi(x) = \sum_{n\in\mathbb{Z}^3}\varphi_\infty(x+nL).
\]
\begin{proof}[Proof of Lemma \ref{lem: bound phi}] Taking $L$ large enough, from the definition of $\varphi$, we find that 
\[
  |\varphi(x)| \leq  |\varphi_\infty(x)| \leq |\varphi_0(x)||\chi_{\sqrt[3]\rho}(x)| \leq C,
\]
 where the last inequality follows from \eqref{eq: bound phi0} together with the fact that $|\chi_{\sqrt[3]\rho}(x)|\leq 1$.

 \noindent\underline{\textit{Bound for the $L^2$ norm of $\varphi$}}. From \eqref{eq: bound phi0}, we get, taking $L$ large enough:
\begin{equation}
  \|\varphi\|^2_{L^2(\Lambda_L)} \leq \int_{\mathbb{R}^3}dx\, |\varphi_\infty(x)|^2 \leq \int_{\mathbb{R}^3}dx\, |\varphi_0(x)|^2|\chi_{\sqrt[3]\rho}(x)|^2  \leq C\rho^{-\frac{1}{3}}.
\end{equation}

\noindent\underline{\textit{Bound for the $L^2$ norm of $\nabla\varphi$}}. From \eqref{eq: bound phi0}, we get, taking $L$ large enough:
\begin{eqnarray}
  \|\nabla\varphi\|^2_{L^2(\Lambda_L)} &\leq& \int_{\mathbb{R}^3}dx\, |\nabla\varphi_\infty(x)|^2 \leq C\int_{\mathbb{R}^3}dx\, |\nabla\varphi_0(x)|^2|\chi_{\sqrt[3]\rho}(x)|^2  + \int_{\mathbb{R}^3}dx\, |\varphi_0(x)|^2|\nabla\chi_{\sqrt[3]\rho}(x)|^2 \nonumber
  \\
  &\leq& C\int_{B_R(0)}|\nabla\varphi_0(x)|^2 + C\int_{B_R^c(0)}|\nabla\varphi_0(x)|^2|\chi_{\sqrt[3]\rho}(x)|^2 + \int_{\mathbb{R}^3}dx\, |\varphi_0(x)|^2|\nabla\chi_{\sqrt[3]\rho}(x)|^2,\nonumber
\end{eqnarray}
where $R>0$ is such that the interaction potential $V_\infty$ vanishes in $B_{R}^c(0):= \mathbb{R}^3\setminus B_R(0)$. Using then the properties of $\varphi_0$ stated in Lemma \ref{lem: bound phi0} together with the ones of $\chi_{\sqrt[3]\rho}$, we can conclude that $\|\nabla\varphi\|_{L^2(\Lambda_L)}\leq C$.\\

\noindent\underline{\textit{Estimate for the $L^1$ norm of $\varphi$}}. From \eqref{eq: bound phi0} and proceeding similar as above, we have for $L$ large enough, 
\begin{equation}
  \|\varphi\|_{L^1(\Lambda_L)}  \leq \int_{\mathbb{R}^3}dx\, |\varphi_\infty(x)| \leq \int_{\mathbb{R}^3}|\varphi_0(x)||\chi_{\sqrt[3]\rho}(x)| \leq C\rho^{-\frac{2}{3}}.
\end{equation}
\noindent\underline{\textit{Estimate for the $L^1$ norm of $\nabla\varphi$}}. We can write 
\begin{equation}
  \nabla\varphi_\infty(x) = \varphi_0(x)\nabla\chi_{\sqrt[3]\rho}(x) + \chi_{\sqrt[3]\rho}(x)\nabla\varphi_0(x).
\end{equation}
Using then the decay for $\varphi_0$ and $\nabla\varphi_0$ stated in Lemma \ref{lem: bound phi0}, we have, for $L$ large enough,
\begin{equation}\label{eq: L1 norm gradient phi}
  \int_{\Lambda_L} dx\, |\nabla\varphi(x)| \leq \int_{\mathbb{R}^3}dx\, |\nabla\varphi_\infty(x)| \leq C\int_{\mathbb{R}^3} |\varphi_0(x)\nabla\chi_{\sqrt[3]\rho}(x)| + C\int_{\mathbb{R}^3} |\chi_{\sqrt[3]\rho}(x)\nabla\varphi_0(x)| \leq C\rho^{-\frac{1}{3}},
\end{equation}
where we use also that $\|\chi_{\sqrt[3]\rho}\|_\infty\leq C$, $\|\nabla\chi_{\sqrt[3]\rho}\|_\infty\leq C\rho^{1/3}$.\\
\noindent\underline{\textit{Estimate for the $L^1$ norm of $\Delta\varphi$}}. We recall that $\varphi_\infty$ is such that 
\[
  2\Delta\varphi_\infty =  - V_\infty(1-\varphi_\infty) + \mathcal{E}_{\{\varphi_0, \chi_{\sqrt[3]\rho}\}},
\]
with 
\[
  \mathcal{E}_{\{\varphi_0, \chi_{\sqrt[3]\rho}\}} = -4\nabla\varphi_0 \nabla\chi_{\sqrt[3]\rho} - 2\varphi_0 \Delta\chi_{\sqrt[3]\rho}.
\]
It then follows, taking $L$ large enough and using the bounds for $\nabla\varphi_0$ and $\varphi_0$ stated in Lemma \ref{lem: bound phi0} together with the support properties of $\chi_{\sqrt[3]\rho}$, that
\[
  \|\Delta\varphi\|_{L^1(\Lambda_L)} \leq \|\Delta\varphi_\infty\|_1 \leq C\|V\|_1 + C\int_{\mathbb{R}^3} dx\, |\nabla\varphi_0(x)|\nabla\chi_{\sqrt[3]{\rho}}(x)| + C\int_{\mathbb{R}^3} dx\, |\varphi_0(x)||\Delta\chi_{\sqrt[3]{\rho}}(x)| \leq C,
\]
where we also used that $\|\nabla^n\chi_{\sqrt[3]\rho}\|_\infty\leq C\rho^{n/3}$.\\

\noindent\underline{\textit{Estimate for the $L^1$ norm of $D^2\varphi$ and $D^3\varphi$}.} The bound for $\|D^2\varphi\|_1$ easily follows from the properties of $\varphi_0$. We have taking $L$ large enough and $R>0$ large enough such that $\mathrm{supp}V \subset \{x\in \mathbb{R}^3\, \vert\, |x| \leq R_0\} \subset B_R(0)$
\begin{eqnarray}
  \int_{\Lambda_L}dx\, |D^2\varphi(x)| &\leq& \int_{\mathbb{R}^3}|D^2(\varphi_0\chi_{\sqrt[3]\rho})(x)| 
  \\
  &\leq& \int_{B_R(0)}|D^2(\varphi_0\chi_{\sqrt[3]\rho})(x)| + \int_{B_R^c(0)}|D^2(\varphi_0\chi_{\sqrt[3]\rho})(x)|\leq C + C|\log\rho|,
\end{eqnarray}
where we used the properties of $\chi_{\sqrt[3]\rho}$, \eqref{eq: bound n der} and the fact that in $B_R^c(0)$, $\varphi_0(x) = a/|x|$. The proof for $\|D^3\varphi\|_{L^1(\Lambda_L)}\leq C$ can be done in the same way, we omit the details.
\end{proof}
We now prove the decay estimates stated in Lemma \ref{lem: decay phi}.
\begin{proof}[Proof of Lemma \ref{lem: decay phi}]
  We first take into account the case $n=2$, we have that for each $p\in (2\pi/L)\mathbb{Z}^3$, $\hat{\varphi}(p) = \int_{\Lambda_L}\varphi(x) e^{ip\cdot x}$ is such that
  \begin{equation}\label{eq: est phi > k}
    |\hat{\varphi}(p)| \leq \|\varphi\|_1 \leq C\rho^{-\frac{2}{3}}.
  \end{equation}
  Moreover, from Lemma \ref{lem: bound phi0}, we get
  \[
    |p|^2 |\hat{\varphi}(p)| \leq C\|\Delta\varphi\|_1 \leq C.
  \]
  Thus we find 
  \[
    |\hat{\varphi}(p)| \leq C\frac{\rho^{-\frac{2}{3}}}{1 + \rho^{-\frac{2}{3}}|p|^2}.
  \]
More in general, for all $n\in\mathbb{N}$, $n>2$, we can proceed as follows. Taking $L$ large enough, we have
\begin{equation}
  |p|^{n} |\hat\varphi(p)| \leq C\int_{\mathbb{R}^3}dx\, |D^n(\varphi_0\chi_{\sqrt[3]\rho})(x)| \leq C_n + C\rho^{\frac{n-2}{3}}\leq C_n,
\end{equation}
where we used the fact that in the support of $\nabla^n\chi_{\sqrt[3]\rho}$, $\varphi_0(x) = a/|x|$ together with the bounds $\|\varphi\|_1 \leq C\rho^{-\frac{2}{3}}$ and $\|\nabla^n\chi_{\sqrt[3]\rho}\|_\infty\leq C\rho^{n/3}$ and the estimate in \eqref{eq: bound n der}.
We then get that for all $n\in\mathbb{N}$, $n\geq 2$:
\begin{equation}
    |\hat{\varphi}(p)| \leq C_n\frac{\rho^{-\frac{2}{3}}}{(1 + \rho^{-\frac{2}{3}}|p|^n)},
  \end{equation}
  which concludes the proof of \eqref{eq: est decay phi>}.
  \end{proof}
\section{Estimate of part of the terms in \eqref{eq: a+b cut-off}}\label{app: cut-off}
In this section we take into account part of the terms in \eqref{eq: a+b cut-off}. More precisely, we study the terms in $\mathrm{B}$ in \eqref{eq: a+b cut-off} for which at least one between $a_\uparrow (u^r_x)$ and $a_\downarrow(u^r_y)$ is replaced by $a_\uparrow(\delta^>_x)$ or $a_\downarrow(\delta^>_y)$. In the case in which we have $a_\downarrow(\delta^>_y)$, consider for instance 
\begin{equation}
  \mathrm{I} :=\int dxdydz\, V(1-\varphi)(x-y) \omega^r_\uparrow(x;z) \langle a_\uparrow (\overline{v}^r_y) a_\uparrow(u_x)a_\downarrow(\delta^>_y)\xi_\lambda, b_\downarrow(\varphi_z)a_\uparrow(u^r_z)\xi_\lambda,\rangle
\end{equation} 
we can use that for any $\xi\in\mathcal{F}$ and for $L$ large enough,
\begin{equation}\label{eq: decay Ib1 1}
  \left\|\int\, dy \, a_\downarrow(\delta^>_y)V(1 -\varphi) (x-y) a_\downarrow(\overline{v}^r_y)\xi\right\| \leq C_n \rho^{\beta n}\|\xi\|, \qquad \mbox{for}\,\,\,\mbox{any}\,\,\, n\in\mathbb{N}\,\,\,\mbox{large}\,\,\,\mbox{enough}.
\end{equation}
Using \eqref{eq: decay Ib1 1}, we can conclude that
\begin{equation}
  |\mathrm{I}| \leq C\rho^{\frac{1}{3} + \beta n}\int dxdz |\omega^r_\uparrow(x;z)|\|a_\uparrow(u_x)\xi_\lambda\|\|a_\uparrow(u^r_z)\|\leq C\rho^{\beta n + \frac{1}{3} -\frac{\epsilon}{3}}\langle \xi_\lambda, \mathcal{N}\xi_\lambda\rangle \leq C\rho^{\frac{3}{2}-\frac{\epsilon}{3} + \beta n }
\end{equation}
where we used the non optimal estimate for the number operator, i.e., $\langle \xi_\lambda, \mathcal{N}\xi_\lambda \rangle \leq CL^{3}\rho^{7/6}$.
For the proof of \eqref{eq: decay Ib1 1}, we refer to \cite[Eq. (C.24)]{FGHP}: it is a consequence of the decay estimates for $\widehat{V(1-\varphi)}(k)$ which follow from the regularity of the interaction potential and the properties of $\varphi$, see Lemma \ref{lem: decay phi}. As a consequence, taking $n\in\mathbb{N}$ large enough, all the terms in which $a_\downarrow(\delta^>_y)$ appear can be made smaller than $L^3\rho^{7/3}$. The case in which instead, there is no $a_\downarrow(\delta^>_y)$ but we do have $a_\uparrow(\delta^>_x)$ has to be treated differently. We can proceed still similarly as in \cite[Appendix C.2]{FGHP}, therefore here we only give the main ideas. All the terms in which we have $a_\uparrow(\delta^>_x)$ can be bounded in the same way. We take into account
\begin{equation}
  \mathrm{II} :=\int dxdydz\, V(1-\varphi)(x-y) \omega^r(x;z) \langle a_\uparrow (\overline{v}^r_y) a_\uparrow(\delta^>_x)a_\downarrow(u_y)\xi_\lambda, b_\downarrow(\varphi_z)a_\uparrow(u^r_z)\xi_\lambda,\rangle.
\end{equation}
We can rewrite $\mathrm{II}$ as follows 
\begin{equation}
  \mathrm{II} = -\int dy\,\langle a_\uparrow (\overline{v}^r_y) a_\downarrow(u_y)\xi_\lambda,\int dxdz V(1-\varphi)(x-y) \omega^r_\uparrow(x;z) a_\uparrow^\ast(\delta^>_x)b_\downarrow(\varphi_z)a_\uparrow(u^r_z)\xi_\lambda,\rangle.
\end{equation}
Moreover, 
\[
  \int dxdz V(1-\varphi)(x-y) \omega^r_\uparrow(x;z) a_\uparrow^\ast(\delta^>_x)b_\downarrow(\varphi_z)a_\uparrow(u^r_z) = \int dz\, a^\ast(A_{z,y})b_\downarrow(\varphi_z) a_\uparrow (u^r_z),
\]
where 
\[
  A_{z,y}(r) = \int dx\, \delta^>(r;x)V(1-\varphi)(x-y) \omega^r_\uparrow(x;z).
\]
It then follows that 
\begin{multline*}
  \left \|\int dxdz V(1-\varphi)(x-y) \omega^r_\uparrow(x;z) a_\uparrow^\ast(\delta^>_x)b_\downarrow(\varphi_z)a_\uparrow(u^r_z)\xi_\lambda\right\|\leq \int dz \|a^\ast(A_{z,y})\|\|b_\downarrow(\varphi_z) \|\|a_\uparrow (u^r_z)\|\|\xi_\lambda\|
  \\
  \leq C\rho^{\frac{1}{3} -\frac{3}{2}\beta}\int dz \|A_{z,y}\|_2,
\end{multline*}
where we used Lemma \ref{lem: bound b phi} and Proposition \ref{pro: L1, L2 v}.
Proceeding similarly as in \cite[Eq. (C.30) -- (C.35)]{FGHP}, i.e., using some decay estimates for $\widehat{V(1-\varphi)}(k)$ and bounds for the derivatives of $\hat{\omega}^r$ and $\hat{\delta}^>$, one can prove that\footnote{Note that the estimate in \eqref{eq: est inte dz norm A} depends on $\epsilon$, this is not the case in \cite[Eq. (C.35)]{FGHP}. This is due to the fact that here there is $\omega^r$ in the definition of the operator $A_{z,y}$ and the derivatives of $\hat{\omega}^r$ have a bound which depends on $\epsilon$. However, taking $n\in\mathbb{N}$ large enough, we can still prove the final result.}for $n\in\mathbb{N}$ large enough,
\begin{equation}\label{eq: est inte dz norm A}
  \int dz\, \|A_{z,y}\|_2 \leq C\rho^{\beta(n-3) -\frac{3}{2} -\epsilon}
\end{equation}
Then, we can bound 
\[
  |\mathrm{II}|\leq C\rho^{\beta(n-3) -\frac{7}{6} -\epsilon -\frac{3}{2}\beta}\int dy\, \|a_\uparrow (\overline{v}^r_y)a_\downarrow(u_y)\xi_\lambda\|\leq CL^{\frac{3}{2}} \rho^{\beta n - \frac{2}{3} -\epsilon -\frac{9}{2}\beta}\|\mathcal{N}^{\frac{1}{2}}\xi_\lambda\| \leq CL^3\rho^{\beta n + \frac{7}{12} -\frac{2}{3}-\epsilon -\frac{9}{2}\beta},
\]
where we used the non optimal estiamate $\|\mathcal{N}^{1/2}\xi_\lambda\|\leq CL^{3/2}\rho^{7/12}$. Taking then $n\in\mathbb{N}$ large enough, we can conclude that all the error terms containing at least one between $a_\uparrow(\delta^>_x)$ and $a_\downarrow(\delta^>_y)$, are negligible.

\end{document}